\documentclass[11pt]{article}

\usepackage{amsmath,amstext,amssymb,amsfonts}
\usepackage{tcolorbox}
\tcbuselibrary{breakable}
\usepackage[dvipsnames]{xcolor}
\usepackage{float}
\usepackage{enumerate}
\usepackage[T1]{fontenc}
\usepackage{mathpazo}
\usepackage{graphicx}
\usepackage{parskip}
\usepackage{accents}
\usepackage{dirtytalk}
\usepackage{bbold}
\usepackage{caption}
\usepackage{mathcomp}
\DeclareCaptionLabelFormat{nolabel}{}
\usepackage{chngcntr} 
\usepackage{geometry} 
\usepackage{tikz}

\usepackage{textcomp}
    \AtBeginDocument{%
        
    }
\usepackage{upquote} 
\usepackage{eurosym} 
\usepackage[mathletters]{ucs} 
\usepackage{inputenc} 
\usepackage{fancyvrb} 
\usepackage{grffile} 
\usepackage{longtable} 
\usepackage{booktabs} 

\usepackage[pdftex,pagebackref,colorlinks,filecolor=blue,linkcolor=blue,citecolor=magenta,urlcolor=JungleGreen]{hyperref}

\definecolor{urlcolor}{RGB}{205,0,0}
\definecolor{linkcolor}{RGB}{0,71,171}
\definecolor{citecolor}{RGB}{0,129,127}

\definecolor{ansi-black}{HTML}{3E424D}
\definecolor{ansi-black-intense}{HTML}{282C36}
\definecolor{ansi-red}{HTML}{E75C58}
\definecolor{ansi-red-intense}{HTML}{B22B31}
\definecolor{ansi-green}{HTML}{00A250}
\definecolor{ansi-green-intense}{HTML}{007427}
\definecolor{ansi-yellow}{HTML}{DDB62B}
\definecolor{ansi-yellow-intense}{HTML}{B27D12}
\definecolor{ansi-blue}{HTML}{208FFB}
\definecolor{ansi-blue-intense}{HTML}{0065CA}
\definecolor{ansi-magenta}{HTML}{D160C4}
\definecolor{ansi-magenta-intense}{HTML}{A03196}
\definecolor{ansi-cyan}{HTML}{60C6C8}
\definecolor{ansi-cyan-intense}{HTML}{258F8F}
\definecolor{ansi-white}{HTML}{C5C1B4}
\definecolor{ansi-white-intense}{HTML}{A1A6B2}

\DefineVerbatimEnvironment{Highlighting}{Verbatim}{commandchars=\\\{\}}

    \title{Introduction to python-igraph}
\makeatletter
\def\PY@reset{\let\PY@it=\relax \let\PY@bf=\relax%
    \let\PY@ul=\relax \let\PY@tc=\relax%
    \let\PY@bc=\relax \let\PY@ff=\relax}
\def\PY@tok#1{\csname PY@tok@#1\endcsname}
\def\PY@toks#1+{\ifx\relax#1\empty\else%
    \PY@tok{#1}\expandafter\PY@toks\fi}
\def\PY@do#1{\PY@bc{\PY@tc{\PY@ul{%
    \PY@it{\PY@bf{\PY@ff{#1}}}}}}}
\def\PY#1#2{\PY@reset\PY@toks#1+\relax+\PY@do{#2}}

\DeclareMathSymbol{\widehatsym}{\mathord}{largesymbols}{"62}

\expandafter\def\csname PY@tok@gd\endcsname{\def\PY@tc##1{\textcolor[rgb]{0.63,0.00,0.00}{##1}}}
\expandafter\def\csname PY@tok@gu\endcsname{\let\PY@bf=\textbf\def\PY@tc##1{\textcolor[rgb]{0.50,0.00,0.50}{##1}}}
\expandafter\def\csname PY@tok@gt\endcsname{\def\PY@tc##1{\textcolor[rgb]{0.00,0.27,0.87}{##1}}}
\expandafter\def\csname PY@tok@gs\endcsname{\let\PY@bf=\textbf}
\expandafter\def\csname PY@tok@gr\endcsname{\def\PY@tc##1{\textcolor[rgb]{1.00,0.00,0.00}{##1}}}
\expandafter\def\csname PY@tok@cm\endcsname{\let\PY@it=\textit\def\PY@tc##1{\textcolor[rgb]{0.25,0.50,0.50}{##1}}}
\expandafter\def\csname PY@tok@vg\endcsname{\def\PY@tc##1{\textcolor[rgb]{0.10,0.09,0.49}{##1}}}
\expandafter\def\csname PY@tok@vi\endcsname{\def\PY@tc##1{\textcolor[rgb]{0.10,0.09,0.49}{##1}}}
\expandafter\def\csname PY@tok@vm\endcsname{\def\PY@tc##1{\textcolor[rgb]{0.10,0.09,0.49}{##1}}}
\expandafter\def\csname PY@tok@mh\endcsname{\def\PY@tc##1{\textcolor[rgb]{0.40,0.40,0.40}{##1}}}
\expandafter\def\csname PY@tok@cs\endcsname{\let\PY@it=\textit\def\PY@tc##1{\textcolor[rgb]{0.25,0.50,0.50}{##1}}}
\expandafter\def\csname PY@tok@ge\endcsname{\let\PY@it=\textit}
\expandafter\def\csname PY@tok@vc\endcsname{\def\PY@tc##1{\textcolor[rgb]{0.10,0.09,0.49}{##1}}}
\expandafter\def\csname PY@tok@il\endcsname{\def\PY@tc##1{\textcolor[rgb]{0.40,0.40,0.40}{##1}}}
\expandafter\def\csname PY@tok@go\endcsname{\def\PY@tc##1{\textcolor[rgb]{0.53,0.53,0.53}{##1}}}
\expandafter\def\csname PY@tok@cp\endcsname{\def\PY@tc##1{\textcolor[rgb]{0.74,0.48,0.00}{##1}}}
\expandafter\def\csname PY@tok@gi\endcsname{\def\PY@tc##1{\textcolor[rgb]{0.00,0.63,0.00}{##1}}}
\expandafter\def\csname PY@tok@gh\endcsname{\let\PY@bf=\textbf\def\PY@tc##1{\textcolor[rgb]{0.00,0.00,0.50}{##1}}}
\expandafter\def\csname PY@tok@ni\endcsname{\let\PY@bf=\textbf\def\PY@tc##1{\textcolor[rgb]{0.60,0.60,0.60}{##1}}}
\expandafter\def\csname PY@tok@nl\endcsname{\def\PY@tc##1{\textcolor[rgb]{0.63,0.63,0.00}{##1}}}
\expandafter\def\csname PY@tok@nn\endcsname{\let\PY@bf=\textbf\def\PY@tc##1{\textcolor[rgb]{0.00,0.00,1.00}{##1}}}
\expandafter\def\csname PY@tok@no\endcsname{\def\PY@tc##1{\textcolor[rgb]{0.53,0.00,0.00}{##1}}}
\expandafter\def\csname PY@tok@na\endcsname{\def\PY@tc##1{\textcolor[rgb]{0.49,0.56,0.16}{##1}}}
\expandafter\def\csname PY@tok@nb\endcsname{\def\PY@tc##1{\textcolor[rgb]{0.00,0.50,0.00}{##1}}}
\expandafter\def\csname PY@tok@nc\endcsname{\let\PY@bf=\textbf\def\PY@tc##1{\textcolor[rgb]{0.00,0.00,1.00}{##1}}}
\expandafter\def\csname PY@tok@nd\endcsname{\def\PY@tc##1{\textcolor[rgb]{0.67,0.13,1.00}{##1}}}
\expandafter\def\csname PY@tok@ne\endcsname{\let\PY@bf=\textbf\def\PY@tc##1{\textcolor[rgb]{0.82,0.25,0.23}{##1}}}
\expandafter\def\csname PY@tok@nf\endcsname{\def\PY@tc##1{\textcolor[rgb]{0.00,0.00,1.00}{##1}}}
\expandafter\def\csname PY@tok@si\endcsname{\let\PY@bf=\textbf\def\PY@tc##1{\textcolor[rgb]{0.73,0.40,0.53}{##1}}}
\expandafter\def\csname PY@tok@s2\endcsname{\def\PY@tc##1{\textcolor[rgb]{0.73,0.13,0.13}{##1}}}
\expandafter\def\csname PY@tok@nt\endcsname{\let\PY@bf=\textbf\def\PY@tc##1{\textcolor[rgb]{0.00,0.50,0.00}{##1}}}
\expandafter\def\csname PY@tok@nv\endcsname{\def\PY@tc##1{\textcolor[rgb]{0.10,0.09,0.49}{##1}}}
\expandafter\def\csname PY@tok@s1\endcsname{\def\PY@tc##1{\textcolor[rgb]{0.73,0.13,0.13}{##1}}}
\expandafter\def\csname PY@tok@dl\endcsname{\def\PY@tc##1{\textcolor[rgb]{0.73,0.13,0.13}{##1}}}
\expandafter\def\csname PY@tok@ch\endcsname{\let\PY@it=\textit\def\PY@tc##1{\textcolor[rgb]{0.25,0.50,0.50}{##1}}}
\expandafter\def\csname PY@tok@m\endcsname{\def\PY@tc##1{\textcolor[rgb]{0.40,0.40,0.40}{##1}}}
\expandafter\def\csname PY@tok@gp\endcsname{\let\PY@bf=\textbf\def\PY@tc##1{\textcolor[rgb]{0.00,0.00,0.50}{##1}}}
\expandafter\def\csname PY@tok@sh\endcsname{\def\PY@tc##1{\textcolor[rgb]{0.73,0.13,0.13}{##1}}}
\expandafter\def\csname PY@tok@ow\endcsname{\let\PY@bf=\textbf\def\PY@tc##1{\textcolor[rgb]{0.67,0.13,1.00}{##1}}}
\expandafter\def\csname PY@tok@sx\endcsname{\def\PY@tc##1{\textcolor[rgb]{0.00,0.50,0.00}{##1}}}
\expandafter\def\csname PY@tok@bp\endcsname{\def\PY@tc##1{\textcolor[rgb]{0.00,0.50,0.00}{##1}}}
\expandafter\def\csname PY@tok@c1\endcsname{\let\PY@it=\textit\def\PY@tc##1{\textcolor[rgb]{0.25,0.50,0.50}{##1}}}
\expandafter\def\csname PY@tok@fm\endcsname{\def\PY@tc##1{\textcolor[rgb]{0.00,0.00,1.00}{##1}}}
\expandafter\def\csname PY@tok@o\endcsname{\def\PY@tc##1{\textcolor[rgb]{0.40,0.40,0.40}{##1}}}
\expandafter\def\csname PY@tok@kc\endcsname{\let\PY@bf=\textbf\def\PY@tc##1{\textcolor[rgb]{0.00,0.50,0.00}{##1}}}
\expandafter\def\csname PY@tok@c\endcsname{\let\PY@it=\textit\def\PY@tc##1{\textcolor[rgb]{0.25,0.50,0.50}{##1}}}
\expandafter\def\csname PY@tok@mf\endcsname{\def\PY@tc##1{\textcolor[rgb]{0.40,0.40,0.40}{##1}}}
\expandafter\def\csname PY@tok@err\endcsname{\def\PY@bc##1{\setlength{\fboxsep}{0pt}\fcolorbox[rgb]{1.00,0.00,0.00}{1,1,1}{\strut ##1}}}
\expandafter\def\csname PY@tok@mb\endcsname{\def\PY@tc##1{\textcolor[rgb]{0.40,0.40,0.40}{##1}}}
\expandafter\def\csname PY@tok@ss\endcsname{\def\PY@tc##1{\textcolor[rgb]{0.10,0.09,0.49}{##1}}}
\expandafter\def\csname PY@tok@sr\endcsname{\def\PY@tc##1{\textcolor[rgb]{0.73,0.40,0.53}{##1}}}
\expandafter\def\csname PY@tok@mo\endcsname{\def\PY@tc##1{\textcolor[rgb]{0.40,0.40,0.40}{##1}}}
\expandafter\def\csname PY@tok@kd\endcsname{\let\PY@bf=\textbf\def\PY@tc##1{\textcolor[rgb]{0.00,0.50,0.00}{##1}}}
\expandafter\def\csname PY@tok@mi\endcsname{\def\PY@tc##1{\textcolor[rgb]{0.40,0.40,0.40}{##1}}}
\expandafter\def\csname PY@tok@kn\endcsname{\let\PY@bf=\textbf\def\PY@tc##1{\textcolor[rgb]{0.00,0.50,0.00}{##1}}}
\expandafter\def\csname PY@tok@cpf\endcsname{\let\PY@it=\textit\def\PY@tc##1{\textcolor[rgb]{0.25,0.50,0.50}{##1}}}
\expandafter\def\csname PY@tok@kr\endcsname{\let\PY@bf=\textbf\def\PY@tc##1{\textcolor[rgb]{0.00,0.50,0.00}{##1}}}
\expandafter\def\csname PY@tok@s\endcsname{\def\PY@tc##1{\textcolor[rgb]{0.73,0.13,0.13}{##1}}}
\expandafter\def\csname PY@tok@kp\endcsname{\def\PY@tc##1{\textcolor[rgb]{0.00,0.50,0.00}{##1}}}
\expandafter\def\csname PY@tok@w\endcsname{\def\PY@tc##1{\textcolor[rgb]{0.73,0.73,0.73}{##1}}}
\expandafter\def\csname PY@tok@kt\endcsname{\def\PY@tc##1{\textcolor[rgb]{0.69,0.00,0.25}{##1}}}
\expandafter\def\csname PY@tok@sc\endcsname{\def\PY@tc##1{\textcolor[rgb]{0.73,0.13,0.13}{##1}}}
\expandafter\def\csname PY@tok@sb\endcsname{\def\PY@tc##1{\textcolor[rgb]{0.73,0.13,0.13}{##1}}}
\expandafter\def\csname PY@tok@sa\endcsname{\def\PY@tc##1{\textcolor[rgb]{0.73,0.13,0.13}{##1}}}
\expandafter\def\csname PY@tok@k\endcsname{\let\PY@bf=\textbf\def\PY@tc##1{\textcolor[rgb]{0.00,0.50,0.00}{##1}}}
\expandafter\def\csname PY@tok@se\endcsname{\let\PY@bf=\textbf\def\PY@tc##1{\textcolor[rgb]{0.73,0.40,0.13}{##1}}}
\expandafter\def\csname PY@tok@sd\endcsname{\let\PY@it=\textit\def\PY@tc##1{\textcolor[rgb]{0.73,0.13,0.13}{##1}}}

\makeatother
\definecolor{incolor}{rgb}{0.0, 0.0, 0.5}
\definecolor{outcolor}{rgb}{0.545, 0.0, 0.0}
  
\hypersetup{
    breaklinks=true,  %
    colorlinks=true,
    urlcolor=urlcolor,
    linkcolor=linkcolor,
    citecolor=citecolor,
    }
\usepackage{pstricks}
\usepackage{graphicx}
\usepackage{xcolor}
\usepackage{tabularx}
\usepackage{verbatim}
\usepackage{fullpage}
\usepackage[T1]{fontenc}
\usepackage{bbm}

\newtcolorbox{namedSDP}[1]{%
  breakable,
  colback=white,
  colframe=black!60,
  title={\textbf{#1}},
  fonttitle=\bfseries,
  top=1mm,
  bottom=1mm,
  left=1mm,
  right=1mm
}

\usepackage{amsthm}
\usepackage{ifdraft}
\usepackage{algorithm2e}
\usepackage{algorithmic}
\usepackage{appendix}
\usepackage{multicol}

\usepackage{nicefrac}

\newcommand{\sfP}{\mathsf{P}}
\newcommand{\sfQ}{\mathsf{Q}}
\newcommand{\sfA}{\mathsf{A}}
\newcommand{\maybe}{\mathsf{MAYBE}}
\newcommand{\Adv}{\mathsf{Adv}}
\newcommand{\Ber}{\mathsf{Ber}}

\usepackage[expansion=false]{microtype}

\newtheorem{theorem}{Theorem}
\newtheorem*{theorem*}{Theorem}

\newtheorem{claim}[theorem]{Claim}
\newtheorem*{claim*}{Claim}

\newtheorem{proposition}[theorem]{Proposition}
\newtheorem*{proposition*}{Proposition}
\newtheorem{lemma}[theorem]{Lemma}
\newtheorem*{lemma*}{Lemma}
\newtheorem{corollary}[theorem]{Corollary}

\newtheorem*{conjecture*}{Conjecture}

\newtheorem{fact}[theorem]{Fact}
\newtheorem*{fact*}{Fact}

\newtheorem*{hypothesis*}{Hypothesis}

\theoremstyle{definition}
\newtheorem{definition}[theorem]{Definition}
\newtheorem{construction}[theorem]{Construction}

\newtheorem{MP}[theorem]{MP}

\newtheorem{remark}[theorem]{Remark}

\usepackage{prettyref}
\newcommand{\savehyperref}[2]{\texorpdfstring{\hyperref[#1]{#2}}{#2}}

\newrefformat{parta}{\savehyperref{#1}{a}}
\newrefformat{partb}{\savehyperref{#1}{b}}
\newrefformat{eq}{\savehyperref{#1}{\textup{(\ref*{#1})}}}
\newrefformat{lem}{\savehyperref{#1}{Lemma~\ref*{#1}}}
\newrefformat{def}{\savehyperref{#1}{Definition~\ref*{#1}}}
\newrefformat{thm}{\savehyperref{#1}{Theorem~\ref*{#1}}}
\newrefformat{cor}{\savehyperref{#1}{Corollary~\ref*{#1}}}
\newrefformat{cha}{\savehyperref{#1}{Chapter~\ref*{#1}}}
\newrefformat{sec}{\savehyperref{#1}{Section~\ref*{#1}}}
\newrefformat{app}{\savehyperref{#1}{Appendix~\ref*{#1}}}
\newrefformat{tab}{\savehyperref{#1}{Table~\ref*{#1}}}
\newrefformat{fig}{\savehyperref{#1}{Figure~\ref*{#1}}}
\newrefformat{hyp}{\savehyperref{#1}{Hypothesis~\ref*{#1}}}
\newrefformat{alg}{\savehyperref{#1}{Algorithm~\ref*{#1}}}
\newrefformat{cp}{\savehyperref{#1}{CP~\ref*{#1}}}
\newrefformat{sdp}{\savehyperref{#1}{SDP~\ref*{#1}}}
\newrefformat{mp}{\savehyperref{#1}{MP~\ref*{#1}}}
\newrefformat{qp}{\savehyperref{#1}{QP~\ref*{#1}}}
\newrefformat{vp}{\savehyperref{#1}{VP~\ref*{#1}}}
\newrefformat{lp}{\savehyperref{#1}{LP~\ref*{#1}}}
\newrefformat{rem}{\savehyperref{#1}{Remark~\ref*{#1}}}
\newrefformat{item}{\savehyperref{#1}{Item~\ref*{#1}}}
\newrefformat{step}{\savehyperref{#1}{step~\ref*{#1}}}
\newrefformat{conj}{\savehyperref{#1}{Conjecture~\ref*{#1}}}
\newrefformat{fact}{\savehyperref{#1}{Fact~\ref*{#1}}}
\newrefformat{prop}{\savehyperref{#1}{Proposition~\ref*{#1}}}
\newrefformat{claim}{\savehyperref{#1}{Claim~\ref*{#1}}}
\newrefformat{relax}{\savehyperref{#1}{Relaxation~\ref*{#1}}}
\newrefformat{red}{\savehyperref{#1}{Reduction~\ref*{#1}}}
\newrefformat{part}{\savehyperref{#1}{Part~\ref*{#1}}}
\newrefformat{prob}{\savehyperref{#1}{Problem~\ref*{#1}}}
\newrefformat{ass}{\savehyperref{#1}{Assumption~\ref*{#1}}}
\newrefformat{cons}{\savehyperref{#1}{Construction~\ref*{#1}}}

\newcommand{\Sref}[1]{\hyperref[#1]{\S\ref*{#1}}}

\renewcommand{\leq}{\leqslant}
\renewcommand{\le}{\leqslant}
\renewcommand{\geq}{\geqslant}
\renewcommand{\ge}{\geqslant}

\usepackage{bm}

\usepackage{xspace}

\usepackage{boxedminipage}

\newcommand{\mper}{\,.}

\newcommand{\paren}[1]{\left(#1 \right )}

\newcommand{\brac}[1]{[#1 ]}
\newcommand{\Brac}[1]{\left[#1\right]}

\newcommand{\set}[1]{\left\{#1\right\}}

\newcommand{\abs}[1]{\left\lvert#1\right\rvert}

\newcommand{\ceil}[1]{\lceil #1 \rceil}
\newcommand{\floor}[1]{\lfloor #1 \rfloor}

\newcommand{\norm}[1]{\left\lVert#1\right\rVert}

\newcommand{\defeq}{\stackrel{\textup{def}}{=}}

\newcommand{\R}{\mathbb R}

\newcommand{\Esymb}{\mathbb{E}}
\newcommand{\Psymb}{\mathbb{P}}

\DeclareMathOperator*{\E}{\Esymb}
\DeclareMathOperator*{\Var}{{\sf Var}}

\DeclareMathOperator*{\ProbOp}{\Psymb}

\newcommand{\Prob}[1]{\ProbOp\Brac{#1}}

\newcommand{\Ex}[1]{\E\Brac{#1}}

\renewcommand{\Pr}[1]{\ProbOp\Brac{#1}}
\newcommand{\pr}[2]{\ProbOp_{#1}\Brac{#2}}

\newcommand{\one}{\mathbf{1}}

\definecolor{DSgray}{cmyk}{0,0,0,0.7}

\newcommand{\Tr}{\mathsf{Tr}}
\newcommand{\cA}{\mathcal A}

\newcommand{\cC}{\mathcal C}

\newcommand{\cE}{\mathcal E}
\newcommand{\cF}{\mathcal F}
\newcommand{\cG}{\mathcal G} 
\newcommand{\cH}{\mathcal H} 
\newcommand{\cI}{\mathcal I}

\newcommand{\cR}{\mathcal R}

\newcommand{\cU}{\mathcal U}

\newcommand{\bbR}{\mathbb R}

\newcommand{\bbE}{\mathbb E}

\newcommand{\Erdos}{Erd\H{o}s\xspace}
\newcommand{\Renyi}{R\'enyi\xspace}
\newcommand{\Lovasz}{Lov\'asz\xspace}

\newcommand{\Holder}{Hölder}

\newcommand{\bigO}{\mathcal{O}}

\newcommand{\polylog}{{\sf polylog}}

\newcommand{\no}{\textsc{No}\xspace}	%
\newcommand\numberthis{\addtocounter{equation}{1}\tag{\theequation}} %

\newcommand{\pE}{\widetilde{\mathbb{E}}}
\newcommand{\sos}{\mathsf{SoS}}
\usepackage{turnstile}
\newcommand{\SoSp}[2]{\sststile{#2}{#1}}
\newcommand{\St}{\mathsf{St}_2}
\newcommand{\sosalg}{\hyperref[par:sdp_recovery_algorithm]{\textup{SoS Recovery Algorithm}}}

\DeclareSymbolFont{eulerletters}{U}{eur}{m}{n}
\DeclareMathSymbol{\eulerpi}{\mathalpha}{eulerletters}{"19}\renewcommand{\leq}{\leqslant}

\author{
    Pravesh Kothari\footnote{Princeton University, Princeton}\\ \href{mailto:kothari@cs.princeton.edu}{kothari@cs.princeton.edu}
    \and
    Anand Louis\footnote{Indian Institute of Science, Bengaluru.}\\ \href{mailto:anandl@iisc.ac.in}{anandl@iisc.ac.in}
    \and 
    Rameesh Paul\footnotemark[2]\\ \href{mailto:rameeshpaul@iisc.ac.in}{rameeshpaul@iisc.ac.in}
    \and
    Prasad Raghavendra\footnote{University of California, Berkeley}\\ \href{mailto:raghavendra@berkeley.edu}{raghavendra@berkeley.edu}
}

\title{Improved Certificates for Independence Number\\ in Semirandom Hypergraphs}

\begin{document}
\maketitle

\begin{abstract}
    We study the problem of efficiently certifying upper bounds on the independence number of $\ell$-uniform hypergraphs. This is a notoriously hard problem, with efficient algorithms failing to approximate the independence number within $n^{1-\varepsilon}$ factor in the worst case \cite{Has99,Zuc07}. We study the problem in random and semirandom hypergraphs.
 
    There is a folklore reduction to the graph case achieving a certifiable bound of $\mathcal{O}\paren{\sqrt{n/p}}$. More recently, the work \cite{GKM22} improved this by constructing spectral certificates that yield a bound of $\bigO\paren{\sqrt{n}\cdot \mathsf{polylog}(n)/p^{1/(\ell/2)}}$. 
    In this work, we prove sharper bounds that get rid of the pesky logarithmic factors in $n$, and we obtain sharper bounds in $p$ that nearly attains the threshold of $\bigO\paren{\sqrt{n}/p^{1/\ell}}$. We also give evidence that this is the right scale in $p$ by showing low-degree polynomial lower bounds.
    
    Our certificate is designed by employing the \emph{proofs-to-algorithms} paradigm \cite{BS16,FKP19}, where we show an upper bound for a degree-$2\ell$ pseudo-expectation arising from a degree-$2\ell$ Sum-of-Squares (SoS) relaxation of the natural polynomial system formulation of maximum independent set problem. The more challenging case is the odd-arity hypergraphs, where we employ a tensor based analysis that reduces the problem to proving bounds on the operator norm of a natural class of random chaos matrices associated with $\ell$-uniform hypergraphs. Previous bounds \cite{AMP21,RT23} have a logarithmic dependence, which we remove by leveraging recent progress on matrix concentration inequalities \cite{BBvH23,BLNvH25}; we believe these maybe useful in other hypergraph problems. Since we deploy our certificates within the SoS framework, the bounds continue to hold in presence of monotone adversaries.
    
    Our low-degree polynomial lower bound involves constructing a \emph{quiet} planted distribution that is supported on an independent set of size $k=o\paren{\sqrt{n}/p^{1/\ell}}$ while remaining low-degree indistinguishable from the null distribution, here the null distribution is a random $\ell$-uniform hypergraph. Prior to this work, the problem of constructing a quiet planted distribution in the sparse regimes was open even for the graph case, see \cite{JPR+21,Pot22}.

    Our results are in a sharp contrast to the problem of recovering a planted independent set in a random hypergraph \cite{CZ20,AK26} where the algorithmic threshold is $k \gtrsim \sqrt{n}/p^{1/2(\ell-1)}$. In a concurrent work \cite{FS26}, they also prove matching low-degree polynomial lower bounds for weak recovery at the same threshold.

    As an application, we show our improved certificates can be combined with an SoS relaxation of a natural $r$-coloring polynomial system and a threshold rounding algorithm to recover an arbitrary planted $r$-colorable subhypergraph in a semirandom hypergraph model along the lines of \cite{LPR25} which also allows for strong adversaries.
\end{abstract}
\newpage

\tableofcontents
\newpage

\section{Introduction}
An independent set in a hypergraph $H=(V,E)$ is a set of vertices $S \subseteq V$ 
that contains no hyperedge $e\in E$ entirely contained inside S (i.e., $e \not\subset S$, $\forall e\in E$). The problem of finding the largest independent set, also called the Maximum Independent Set (MIS) is a central problem in theoretical computer science. Its worst-case intractability is well-established: it is one of the original NP-hard problems \cite{Kar72}, and notoriously hard to approximate, admitting no $n^{1-\varepsilon}$ approximation for any $\varepsilon>0$ unless P=NP \cite{Has99,Zuc07}. The maximum size of an independent set is a fundamental structural parameter of the hypergraph, known as its independence number, and denoted by $\alpha(H)$.

Given this severe worst-case intractability, it is natural to study the problem in average-case models. For random graphs, ($\ell=2$), the true value of $\alpha(H)$ is known to be  sharply concentrated at $\paren{2+o(1)}\log_{1/(1-p)}(n)$. However, despite this precise understanding of the ground truth, efficient algorithms are only able to certify (or refute) the existence of independent sets up to size $\bigO\paren{\sqrt{n/p}}$. This bound can be achieved by spectral techniques \cite{AKS98} and the \Lovasz-Theta SDP relaxation \cite{Lov79,Juh82,FK00}. The problem is also closely related to the \emph{planted clique} problem that has been extensively studied in many works including \cite{Jer92,Kuc95,AKS98,FK00,FK01,FGR+13,CSV17,BHK+16,MMT20,BKS23,BBKS24,GW25}.

In this work, we investigate the refutation task in the more general setting of random and semirandom hypergraphs. For an $\ell$-uniform random hypergraph $\cH_0(n,\ell,p)$ where each $\ell$-set is an edge independently with probability $p$, the work of Krivelevich and Sudakov \cite{KS98} shows that with high probability, the true value of independence number is bounded by
\begin{align*}
    \alpha(H) = \Theta\paren{\paren{\log n/p}^{1/(\ell-1)}}
\end{align*}
Our goal is to design efficient algorithms that certify upper bounds on $\alpha(H)$, that is as tight as possible. The problem for hypergraphs was first studied in the work \cite{CGL07} where they certify a bound of $\varepsilon n$ using Feige's XOR trick \cite{Fei02}, and later improved to $\widetilde{\bigO}_{\ell}\paren{n^{3/4}/\sqrt{p}}$ by \cite{AOW15}. More recently, the work of Guruswami, Kothari, and Manohar \cite{GKM22} improves on these bounds by directly constructing spectral certificates and certify a bound of $\bigO_{\ell}\paren{\sqrt{n}\cdot \mathsf{polylog}(n)/p^{1/(\ell/2)}}$. We improve on these previous results in two key aspects,
\begin{enumerate}
    \item \textbf{Sharper Refutation Bounds.} We provide significantly tighter certificates. Our bounds (a) eliminate the logarithmic factors achieving the conjectured tight bounds in terms of $n$, and (b) provide an improved dependence on $p$ that scales as $p^{-1/\ell}$.

    \item \textbf{Robustness.} Our certificates are deployed inside the Sum-of-Squares (SoS) proof system, and hence inherently robust. This allows our results to hold even in semirandom hypergraphs where an adversary can add additional edges.
\end{enumerate}

Additionally, we prove matching low-degree polynomial lower bounds that can be seen as an evidence that refuting below $k=\sqrt{n}/p^{1/\ell}$ is hard. This involves constructing a planted distribution that is supported on independent sets of size $k$, while not allowing any low-degree polynomial to separate the planted distribution from the null distribution of random hypergraphs $\cH_0\paren{n,\ell,p}$. A canonical planted distribution along the lines of \cite{BHK+16} only yields a bound of $k=o\paren{\sqrt{n}/p^{1/2\ell}}$. This difficulty also arises in sparse graph settings where the bound one gets is $\sqrt{n}/p^{1/4}$. Therefore, SoS results \cite{JPR+21,KPX24} have to work extra hard to show SoS lower bounds. To the best of our knowledge, a low-degree polynomial lower bound better than the scale of $\sqrt{n}/p^{1/4}$ in sparse graphs has so far remained elusive. This is also mentioned as an open problem in \cite{JPR+21,Pot22}. 

We give a construction for a \emph{quiet} planted distribution that resolves this open problem, and obtains a low-degree polynomial lower bounds for $k=o\paren{\sqrt{n}/p^{1/\ell}}$. We note that constructing quiet planted distributions to obtain sharper low-degree polynomial lower bounds is a problem specific task \cite{BKW20,BBK+21,KVWX23} and one that requires creative insights.

Our results are in sharp contrast to the problem of recovering planted independent sets in random hypergraphs. The recovery problem was implicitly studied in \cite{CZ20} where they study a projection matrix $A$ of size $n\times n$ with entry $A\Brac{i,j}$ counting the number of hyperedges passing through $i$ and $j$. More recently \cite{AK26} showed that a simple spectral algorithm based on this projection matrix can exactly recover the planted independent set when the size of the planted set is $k \gtrsim \sqrt{n}/p^{1/2(\ell-1)}$. In a concurrent work \cite{FS26}, they also show a matching low-degree polynomial lower bound for weak recovery. However, it is easy to show that this projection matrix does not give any sub-linear bounds for the certification problem we are interested in.

On the other hand, our refutation certificates could be used to design recovery algorithms when the planted set has size $k \gtrsim \sqrt{n}/p^{1/\ell}$. Importantly, the algorithms constructed from refutation certificates can handle strong adversaries such as those considered in \cite{LPR25}, whereas the projection matrix based analysis completely breaks down in this setting.

\subsection{Our Models and Results}
We consider two hypergraph generative models, a random hypergraph model denoted $\cH_0$, and a natural semirandom extension of it denoted $\cH_1$.

\begin{definition}[Random $\ell$-Uniform Hypergraph]
\label{def:random-hypergraph}
    Let $n,\ell \in \mathbb{N}$ and $p\in [0,1]$. The distribution $\cH_0\paren{n,\ell,p}$ defines a random $\ell$-uniform hypergraph on vertex set $V=[n]$ where, for each $\ell$-tuple of distinct vertices $\set{i_1,i_2,\dots,i_{\ell}} \subset V$, the hyperedge $\set{i_1,i_2,\dots,i_{\ell}}$ is included independently at random with probability $p$.
\end{definition}

\begin{definition} [Semirandom Hypergraph Model]
\label{def:semirandom-hypergraph}
Given $n,\ell \in \mathbb{N}$ and $p\in [0,1]$, the distribution $\cH_1(n,\ell,p)$ defines a semirandom $\ell$-uniform hypergraph where an instance from the model is constructed as follows:
\begin{enumerate}
    \item Sample a random $\ell$-uniform hypergraph $H=(V,E)$ where $H \sim \cH_0\paren{n,\ell,p}$.
    \item A monotone adversary may inspect the hypergraph instance and insert additional edges.
\end{enumerate}
\end{definition}

\begin{definition}[Refutation Algorithm for Hypergraph Independence Number]
    A refutation algorithm/certificate for independence number of an $\ell$-uniform hypergraph, takes a hypergraph instance $H=(V,E)$ and outputs a number $k \in [0,n]$ such that,
    \begin{itemize}
        \item \textbf{Correctness.} The output of the certificate is always a valid bound on true independence number of hypergraph i.e., $k \geq \alpha(H)$ with probability $1$.
        \item \textbf{Utility.} If $H \sim \cH_1(n,\ell,p)$, the output is a \say{small}value of $k$ w.h.p. (over the draw of $H$).
    \end{itemize}
\end{definition}
Our main result shows that for hypergraphs drawn from the semirandom model $\cH_1\paren{n,\ell,p}$, there exists an efficient refutation certificate, that w.h.p. (over draw of $H$) outputs a bound $k$ satisfying $k=\bigO_{\ell}\paren{\sqrt{n}/p^{1/\ell}}$, nearly matching the conjectured optimal polynomial-time scaling (in both $n$ and $p$), across a broad range of hyperedge probability parameter $p$.

\begin{figure}[htbp]
	\centering
	\begin{minipage}[b]{0.49\textwidth}
		\centering
		\resizebox{\linewidth}{!}{%
			\begin{tikzpicture}[x=1.1cm, y=0.8cm]
				\draw[->] (0,0) -- (11,0) node[right] {$p$};
				\draw[->] (0,0) -- (0,9) node[left] {$k$};
				\draw[thick] (10,0) -- (10,9) node[below] at (10,0) {$1/2$};
				\draw[dashed] (0,8) -- (10.5,8) node[left] at (0,8) {$n$};
				
				\draw (0.1, 1.5) -- (-0.1, 1.5) node[left] {$\sqrt{n}$};
				\draw (0.1, 2.5) -- (-0.1, 2.5) node[left] {$\sqrt{n}\log n$};
				\draw (0.1, 5.5) -- (-0.1, 5.5) node[left] {$n^{3/4}\log n$};
				\draw (2.0, 0.1) -- (2.0, -0.1) node[below] {\footnotesize $\frac{\log n}{n^3}$};
				\draw (5.0, 0.1) -- (5.0, -0.1) node[below] {\footnotesize $\frac{1}{n^{3/2}}$};
				\draw (6.5, 0.1) -- (6.5, -0.1) node[below] {\footnotesize $\frac{1}{n}$};
				\draw (8.0, 0.1) -- (8.0, -0.1) node[below] {\footnotesize $\frac{1}{\sqrt{n}}$};
				
				\draw[very thick, black] (2.0, 8) -- (10, 1.5);
				\draw[very thick, green!50!black] (5.0, 8) -- (10, 2.5);
				\draw[very thick, BrickRed] (6.5, 8) -- (10, 1.5);
				\draw[very thick, blue] (8.0, 8) -- (10, 5.5);
				
				\node[draw=black!20, fill=white, anchor=south west, inner sep=3pt, rounded corners=2pt] at (0.2, 0.2) {
					\scriptsize
					\begin{tabular}{@{}l@{\hspace{4pt}}l@{}}
						\tikz[baseline=-0.6ex]\draw[very thick, black] (0,0) -- (0.4,0); & \textbf{This work} \\
						\tikz[baseline=-0.6ex]\draw[very thick, green!50!black] (0,0) -- (0.4,0); & [GKM22]\\
						\tikz[baseline=-0.6ex]\draw[very thick, BrickRed] (0,0) -- (0.4,0); & Folklore\\
						\tikz[baseline=-0.6ex]\draw[very thick, blue] (0,0) -- (0.4,0); & [AOW15]\\
					\end{tabular}
				};
			\end{tikzpicture}%
		}
	\end{minipage}
	\hfill
	\begin{minipage}[b]{0.49\textwidth}
		\centering
		\resizebox{\linewidth}{!}{%
			\begin{tikzpicture}[x=1.1cm, y=0.8cm]
				\draw[->] (0,0) -- (11,0) node[right] {$p$};
				\draw[->] (0,0) -- (0,9) node[left] {$k$};
				\draw[thick] (10,0) -- (10,9) node[below] at (10,0) {$1/2$};
				\draw[dashed] (-0.5,8) -- (10.5,8) node[left] at (-0.5,8) {$n$};
				
				\draw (0.1, 1.5) -- (-0.1, 1.5) node[left] {$\sqrt{n}$};
				\draw (0.1, 2.5) -- (-0.1, 2.5) node[left] {$\sqrt{n}\log n$};
				\draw (0.1, 3.2) -- (-0.1, 3.2) node[left] {$\approx n^{4/7}$};
				\draw (0.1, 5.5) -- (-0.1, 5.5) node[left] {$n^{3/4}\log n$};
				
				\draw (1.5, 0.1) -- (1.5, -0.1) node[below] {\footnotesize $\frac{\log n}{n^{7/2}}$};
				\draw (4.0, 0.1) -- (4.0, -0.1) node[below] {\footnotesize $\frac{1}{n^2}$};
				\draw (5.5, 0.1) -- (5.5, -0.1) node[below] {\footnotesize $\frac{1}{n^{3/2}}$};
				\draw (6.5, 0.1) -- (6.5, -0.1) node[below] {\footnotesize $\frac{1}{n}$};
				\draw (8.0, 0.1) -- (8.0, -0.1) node[below] {\footnotesize $\frac{1}{\sqrt{n}}$};
				
				\draw[thick, blue] (10, 5.5) -- (8.0, 8);
				\draw[thick, green!50!black] (10, 2.5) -- (5.5, 8);
				\draw[thick, BrickRed] (10, 1.5) -- (6.5, 8);
				
				\draw[dashed, very thick, MidnightBlue] (10, 2.5) -- (8.0, 3.2);
				\node[MidnightBlue, anchor=south west, font=\scriptsize, align=left] at (8.1, 3.2) {Dense Case\\(Suboptimal)};
				\draw[dashed, very thick, olive] (8.0, 3.2) -- (4.0, 8);
				\node[olive, anchor=south west, font=\scriptsize, align=left] at (5.4, 5.8) {Sparse Case\\\,\,\,\,(Blowup)};
				
				\draw[very thick, black] (10, 1.5) -- (8.0, 3.2);
				\draw[very thick, black] (8.0, 3.2) -- (1.5, 8);
				\filldraw[black] (8.0, 3.2) circle (2.5pt);
				\draw[dashed, very thick, black] (10, 1.5) -- (1.5, 8);
				\node[black, below left] at (7.5, 3.5) {\footnotesize Conjectured};
				
				\node[draw=black!20, fill=white, anchor=south west, inner sep=3pt, rounded corners=2pt] at (0.2, 0.2) {
					\scriptsize
					\begin{tabular}{@{}l@{\hspace{4pt}}l@{}}
						\tikz[baseline=-0.6ex]\draw[very thick, black] (0,0) -- (0.4,0); & \textbf{This work} \\
						\tikz[baseline=-0.6ex]\draw[very thick, green!50!black] (0,0) -- (0.4,0); & [GKM22]\\
						\tikz[baseline=-0.6ex]\draw[very thick, BrickRed] (0,0) -- (0.4,0); & Folklore \\
						\tikz[baseline=-0.6ex]\draw[very thick, blue] (0,0) -- (0.4,0); & [AOW15] \\
					\end{tabular}
				};
			\end{tikzpicture}%
		}
	\end{minipage}
	\caption{Landscape of refutation certificates (log scale): Even case ($\ell=6$) and Odd case ($\ell=7$).}
	\label{fig:compare-results}
\end{figure}
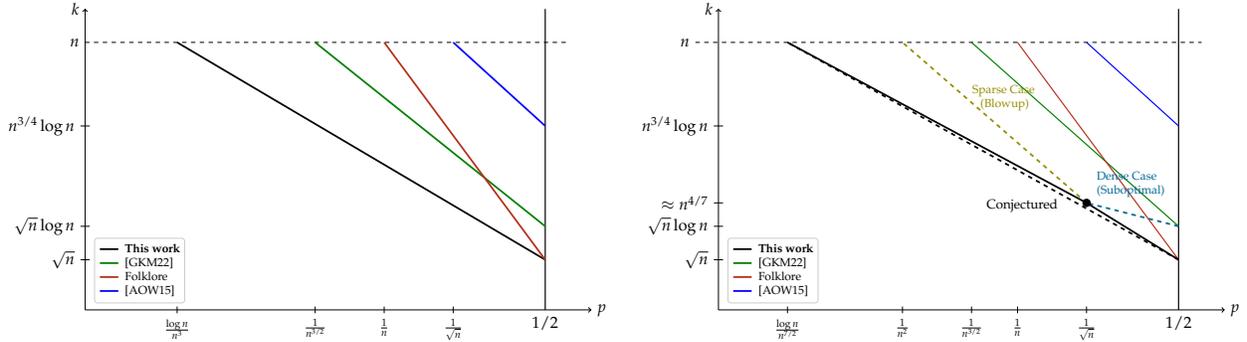

\begin{theorem}\label{thm:main-theorem-informal}
    For an instance of a random hypergraph $H=(V,E)$ generated from the semirandom hypergraph model $\cH_1\paren{n,\ell,p}$, there is an $n^{\bigO\paren{\ell}}$ running time refutation algorithm such that w.h.p.,
    \begin{itemize}
        \item \textbf{Even Arity Case}: For $\ell=2q$ the algorithm certifies,
        \begin{align*}
            k \leq \bigO_{\ell}\paren{\frac{\sqrt{n}}{p^{1/\ell}} + \frac{\paren{\log n}^{1/\ell}}{p^{2/\ell}}}
        \end{align*}
        \item \textbf{Odd Arity Case}: For $\ell=2q+1$ we define a threshold $p^{\star} \defeq \paren{\log n}^{5/2}/\sqrt{n}$ and certify,
        \begin{itemize}
            \item When $p \geq p^{\star}$ (moderate to dense regime),
                \begin{align*}
                    k \leq \bigO_{\ell}\paren{\frac{\sqrt{n}}{p^{1/\ell}}}
                \end{align*}
                \item When $p \leq p^{\star}$ (sparse regimes),
                    \begin{align*}
                        k \leq \bigO_{\ell}\paren{\frac{\sqrt{n}\paren{\log n}^{1/2\ell}}{p^{1/\ell}} + \frac{\paren{\log n}^{3/\ell}}{p^{2/\ell}}}.
                    \end{align*}
        \end{itemize}       
    \end{itemize}
\end{theorem}

An ingenuous folklore approach (see \cite{ZX18} for instance) to certifying bounds on $\alpha(H)$ is to fix $\paren{\ell-2}$ vertices, $T \subset V$ and consider their induced link graph $G_T$. Now if $S$ is independent in $H$, then for every $T \subset S, \abs{T}=\ell-2$, the set $S \setminus T$ is independent in $G_T$ and one can use the graph independent set certificates. However, this reduction is extremely lossy as it only retains $1/n^{\ell-2}$ fraction of edges, yielding a \say{weak} certificate with the bound,
\begin{align*}
    k \leq \paren{\ell-2} + \max_T \alpha_{\text{cert}}\paren{G_T} \leq \bigO_{\ell}\paren{\sqrt{n/p} + \sqrt{\log n}/p}.
\end{align*}

\begin{remark}
    We say the bounds from folklore graph certificates are \say{weak} as they turn vacuous (output a value larger than the trivial value $n$) when $p \leq \sqrt{\log n}/n$. 
\end{remark}
The spectral certificates of \cite{GKM22} are significantly better as we see in their bound below,
\begin{align*}
    k \leq \bigO_{\ell}\paren{\sqrt{n}\cdot \paren{\frac{\log^{4/\ell}(n)}{p^{2/\ell}}} + \sqrt{\frac{\log n}{p^{4/\ell}}}} \text{ is non-vacuous for all }p \gg 1/n^{\ell/4}.
\end{align*}
\begin{remark}
     The bounds obtained in \cite{GKM22} certificate are not tight in factors of $n$, especially for larger values of $p$ (e.g., $p=\tcohm(1)$), and also not known to be robust to presence of monotone adversaries. Our certificates simultaneously improve upon all these aspects (see table below).
\end{remark}

\begin{table}[H]
\centering
\small
\setlength{\tabcolsep}{5pt}
\renewcommand{\arraystretch}{1.25}
\scalebox{1.1}{
\begin{tabular}{lcccc}
\toprule
& \textbf{\cite{AOW15}} 
& \textbf{\textcolor{teal}{Folklore}}
& \textbf{\cite{GKM22}}  
& \textbf{\textcolor{teal}{Our Results}} \\
\midrule

\textbf{Scaling with n}
& $n^{3/4}\cdot \polylog(n)$
& $\sqrt{n}$
& $\sqrt{n}\cdot \polylog(n)$
& $\sqrt{n}$ \\

\textbf{Utility Threshold}
& \emph{$p \gg 1/\sqrt{n}$}
& $p \gg 1/n$
& $p \gg 1/n^{\ell/4}$
& $p \gg 1/n^{\ell/2}$ \\

\textbf{Robustness}
& Yes
& Yes
& No
& Yes\\

\bottomrule
\end{tabular}}
\label{tab:comparison-certificates}
\end{table}

Our certificates are derived from a natural $\bigO(\ell)$-degree Sum-of-Squares relaxation of Maximum Independent Set program, using the \say{proofs-to-algorithm} approach  (see \cite{BS16,FKP19}). However, there are several technical challenges in the setting of hypergraphs, we outline how we address them in \prettyref{sec:proof_overview}. Along the way, we establish logarithmic-free operator norm bounds for a class of matrices arising out of  $\ell$-uniform hypergraphs, which may be of independent interest in other hypergraph problems. Existing bounds for such  matrices with polynomial entries (e.g., in the works \cite{AMP21,RT23}), when applied off-the-shelf to our operator, incur logarithmic losses in $n$; we avoid this by relating our hypergraph operator to recent bounds for \emph{combinatorial chaos} matrices studied in \cite{BLNvH25}.

\begin{remark}
    The odd-arity bound falls slightly off the conjectured curve (see \prettyref{fig:compare-results}), while still strictly improving over all previous certificates. This gap arises from a known bottleneck in existing SoTA concentration bounds for chaos matrices (see Remark A.8 in \cite{BLNvH25}); in sparse regimes, the variance term (olive curve in \prettyref{fig:compare-results}) blows up. By combining chaos bounds with bounds obtained from matrix Bernstein, we track the strictly better black curve, which results in the two-regime behavior of \prettyref{thm:main-theorem-informal}.
\end{remark}

Finally, we give some evidence for why we conjecture the computational threshold is $\sqrt{n}/p^{1/\ell}$ by proving low-degree polynomial lower bounds. Our proofs involve constructing a \say{quiet} planted distribution that is hard to distinguish from the null distribution $\cH_0\paren{n,\ell,p}$, and successfully resolves the open problem in \cite{JPR+21,Pot22}

\begin{theorem}[Informal Version of \prettyref{thm:quiet-planting-hypergraph-theorem}]
    There exists a planted distribution $\sfP$ that is supported on hypergraphs having independent set of size atleast $k$, while no degree-$D$ polynomial $f$ can distinguish $\sfP$ from the null distribution $\sfQ =\cH_0\paren{n,\ell,p}$ for the regimes of $n,k,\ell,D,p$ satisfying,
    \begin{align*}
        k \leq {\frac{n^{1/2-\gamma}}{p^{1/\ell}}}, \quad \text{for $\gamma=o(1)$}, \quad \text{and} \quad D=o\paren{\paren{\frac{\log n}{\log \log n}}^{\ell/(\ell-1)}}.
    \end{align*}
\end{theorem}

\begin{remark}
    The degree $D$ above is close to the best one could hope for. The true independence number of $\cH_0\paren{n,\ell,p}$ is of order $\paren{\frac{\log n}{p}}^{1/(\ell-1)}$; indeed there is a brute force algorithm at this scale that checks candidate sets of this size, where checking a candidate set is an independent set requires checking all $\ell$-tuples inside it. This can be encoded as a polynomial of degree $\Theta\paren{\paren{\frac{\log n}{p}}^{\ell/(\ell-1)}}$.
\end{remark}

As an application, we show how we can use the Sum-of-Squares program and a natural threshold rounding algorithms along with the refutation certificates to recover a large fraction of vertices of planted independent sets, and more generally $r$-colorable planted hypergraphs. The algorithms continue to work in powerful semirandom models, such as those considered in \cite{LPR25}. 

\begin{corollary}\label{cor:planted-coloring-recover}
    Given an instance of an $\ell$-uniform semirandom hypergraph that contains a planted $r$-colorable arbitrary hypergraph on a set $S$ of $k$ vertices, and an adversary that can arbitrarily modify edges across $S \times \paren{V\setminus S}$ (along lines of \cite{LPR25}), there exists an $n^{\bigO\paren{\ell}}$ running-time algorithm which for regimes of $k=\tcohm_{p,l,\varepsilon}\paren{r\sqrt{n}}$ outputs a set $\widehat{S}$ such that with high probability (over draw of $H$) for any $\varepsilon>0$ satisfies,
    \begin{multicols}{2}
        \begin{itemize}
            \item $\abs{\widehat{S}} \leq \paren{1+ \varepsilon}k$
            \item $\abs{\widehat{S} \cap S} \geq \paren{1-\varepsilon}k$.
        \end{itemize}
    \end{multicols}
\end{corollary}

\begin{remark}
    The folklore reduction to the graph case can also be adapted to recover planted independent sets. However, this black-box approach fundamentally breaks down for planted $r$-colorable hypergraphs when $r\geq 2$ (see \prettyref{app:graph_reduction} for a detailed discussion).
\end{remark}

\subsection{Motivation and Related Works}
\paragraph{Independent Sets in Hypergraphs.}
The problem of finding independent sets in a hypergraphs is central in discrete mathematics and theoretical computer science with connections to multiple areas, including forbidden structures in extremal combinatorics \cite{ES46,Sze75,AKP+82,KMV14,BMS15}, coding theory \cite{ST20,HCTR21}, and constraint satisfaction problems \cite{GK01,GL03,CGL07,BDL+23}. The problem also arises in a diverse range of practical applications \cite{Lub93,BG98,GEC+07,KHT09,Yan14,ZSHZ18,GZCP21}.

From a computational perspective, the problem of finding the largest independent set, known as the Maximum Independent Set (MIS) problem is notoriously difficult. It is NP-hard \cite{Kar72}, even for the simpler case of graph ($\ell=2$). The problem is in fact hard to approximate within a factor of $n^{1-\varepsilon}$ for any $\varepsilon>0$ \cite{Has99,Zuc07}, and the best-known algorithm of \cite{Hal00} achieves only an $\bigO\paren{n/\log n}$ approximation. For hypergraphs, the situation is no better: even in bounded-degree hypergraphs, the best-known approximation guarantees scale linearly with $\Delta$, the maximum degree parameter \cite{HL09,AKS11,BK22}.

\paragraph{Certifying Independent Set in Random Hypergraphs.}
Given the severe worst-case intractability, it is natural to study the problem in random models, where typical instances exhibit more structure than worst-case ones.
In this work, we focus on average-case analogues of well-studied hypergraph problems, mirroring the recent trend of generalizing well-studied matrix problems to tensors, see \cite{KB09,ZX18,CLPC19,LZ20,LZ22,HLWZ22}.

In a random hypergraph, the true size of the largest independent set (the independence number) is tightly characterized in the work \cite{KS98} showing that,
\begin{align*}
    \alpha(H) =\paren{1 \pm o(1)}\paren{\ell \cdot \log n/p}^{1/(\ell-1)},
\end{align*}
generalizing the classical result for graphs by \cite{Mat76,Fri90}. From an algorithmic perspective, a recent work \cite{DW25} shows that low-degree polynomial algorithms can compute (with high probability) an independent set of size $\paren{1-\varepsilon}\paren{\log n/p}^{1/(\ell-1)}$, along with a low-degree hardness result, thus extending similar results obtained for graphs \cite{GM75,GS14,RV17,Wei22}.

However the task of refuting, i.e. efficiently certifying the non-existence of a large independent set is significantly more challenging.
For the case of graphs ($\ell=2$), efficient certificates can only certify a bound of $k=\bigO\paren{\sqrt{n/p}}$ \cite{Lov79,Juh82,AKS98,FK00}, which leads to the infamous \emph{statistical-computational gap}. The task of designing refutation algorithms that refute existence of a large planted clique/independent set in a random hypergraph, has seen slow gradual progress. The initial threshold is $\varepsilon n$ by \cite{CGL07} who studied this first, and was improved to $\widetilde{\bigO}\paren{n^{3/4}}$ in \cite{AOW15} using Feige's XOR trick (\cite{Fei02}). More recently this was improved to $\bigO\paren{\sqrt{n}\cdot \mathsf{polylog}(n)/p^{2/\ell}}$ using a direct spectral certificate by \cite{GKM22}.

\paragraph{Recovery for Planted Problems.}
The certification task is also closely related to the \emph{planted recovery problem}, where the task is to recover a specific structure hidden within a random graph or hypergraph. The central problem is the \emph{planted clique problem} that has been extensively studied in the graph setting \cite{AKS98,FK00,FK01,FGR+13,BHK+16,MMT20,BKS23,BBKS24,GW25}. There is also a growing body of work of other planted recovery problems such as planted bisection and Stochastic Block Models \cite{FK01,ABBS14,ABH16,CX16}, planted coloring \cite{AK97,DF16,CO04}, graph partitioning problems \cite{MMV12,MMV14,LV18,LV19}, planted dense subgraphs \cite{BCC+10,HWX16a,HWX16b,HWX16c}, planted bipartite subgraphs in random graphs \cite{KLP22,RHS24,LP26}, and planted $r$-colorable subgraphs in random graphs \cite{LPR25}.

In a sharp contrast, recovery problems on planted hypergraphs remain far less explored. The existing literature is comparatively scarce, with only a handful of works, including hypergraph independent set/clique \cite{KLP21}, hypergraph partitioning \cite{GD17}, hypergraph Stochastic Block Models \cite{KBG18,CZ20,ZT23,SZ25}, and the hypergraph densest $k$-subgraph problem \cite{CPSB22}. A part of the reason for this disparity could be attributed to the fact that technically hypergraph offers a challenging terrain. The familiar tools such as the linear algebraic toolkit of graphs have not been fully developed for hypergraphs.

\paragraph{Semirandom Models}
A central motivation for studying average-case models is to develop frameworks that capture the structure of typical real-world instances of the problem, and bridge the gap between average and worst-case inputs. Semirandom models capture this idea by allowing an adversary that interacts with the random instance, \say{dialing up} the difficulty to produce a spectrum of inputs that interpolate between the two extremes. This algorithm design perspective is broadly known as \say{Beyond Worst-case Analysis}, and we refer to the book, \cite{Rou21} for a comprehensive survey of numerous such frameworks.

For average-case graph problems, the work \cite{BS95} introduced the \emph{monotone adversary} model, where an adversary can add edges to the graph, except within the planted independent set, but cannot remove edges in the graph. We next discuss various semirandom models for the independent set problem in graphs, each representing a significant step towards worst-case robustness.
\begin{itemize}
    \item \textbf{Monotone Adversary.} The seemingly helpful adversary is powerful enough to break the classical spectral algorithm of \cite{AKS98}. The work \cite{FK00} demonstrated that one can still certify $k=\bigO\paren{\sqrt{n/p}}$, but having to resort to the power of the more robust techniques such as semidefinite programming (SDP). This illustrates another key benefit of semirandom models, they not only capture more realistic inputs but also foster the development of more resilient algorithm design techniques.

    \item \textbf{Feige-Kilian Model.} A particularly influential model for the planted independent set problem is the \emph{Feige-Kilian model}, due to \cite{FK01} which grants the adversary more power, in addition to adding edges, it can also delete edges in the graph outside the planted set $V\setminus S$. Recently breakthroughs in the works of \cite{CSV17,MMT20,BKS23} make substantial progress in this highly challenging setting. Notably, the works of \cite{BBKS24,GW25} study the problem in a closely related model due to \cite{CSV17}, and match the computational threshold conjectured by \cite{Ste18} in model, surprisingly with a greedy algorithm. These advances underscore how each semirandom model reveals new facets of the problem and its computational landscape. 
    
    \item \textbf{\cite{LPR25} Model.} A related semirandom model (complementary in strength to the Feige-Kilian model) was considered in \cite{LPR25} where an \emph{arbitrary adversary} may inspect the graph and freely add and remove edges across $S \times \paren{V\setminus S}$, while a \emph{monotone adversary} may further add edges within $V\setminus S$. They show approximate recovery algorithms for $k=r\cdot\tcohm(\sqrt{n/p})$ for independent set, and more generally $r$-colorable planted subgraphs in this model.
\end{itemize}

In this work, for the task of designing refutation certificates for independent sets in hypergraphs we consider the monotone adversary model. The existing spectral certificates of \cite{GKM22} are not expected to be robust in this setting. On the other hand, our Sum-of-Squares (SoS) based certificates are inherently robust. This phenomenon in hypergraphs is the direct analogue of the robustness observed for SDP-based certificates in graphs in the work \cite{FK00}.
For the problem of recovering planted independent sets (and planted $r$-colorable hypergraphs), we show approximate recovery in the \cite{LPR25} model.

\paragraph{Sum-of-Squares Hierarchy} The Sum-of-Squares/Lasserre hierarchy (also called higher-order SDP) has emerged as a powerful algorithmic workhorse, particularly in robust statistics and average-case algorithm design. Notable applications include robust method of moments \cite{KSS18}, robust linear regression \cite{KKM18,BP21}, robust clustering \cite{KSS18,BK20}, robustly learning mixture of Gaussians \cite{BDJK+22,LM23}, list-decodable learning \cite{KSS18,KKK19,RY20,BK21}, and many more. The underpinning principle in Sum-of-Squares framework is a deep connection between optimization and proof systems. It translates algorithm design into the task of certifying bounds via low degree polynomials, an approach dubbed \emph{proofs-to-algorithms} paradigm (see \cite{BS16,FKP19} for details); this framework has been instrumental in many of the success stories listed above. 

For an $\ell$-uniform hypergraph, a degree $\bigO(\ell)$ Sum-of-Squares hierarchy provides us a natural analogue of the basic SDP relaxation used for the corresponding graph problem. Classical works on hypergraph Maximum Independent Set problem, such as the works \cite{Chl07,CS08} have already utilized higher-order SDPs in designing improved approximation algorithms. More recently, the work \cite{KLP21} extends the planted clique algorithms developed for the Feige-Kilain semirandom model in the work \cite{MMT20} from graphs to hypergraphs, using higher-order SDPs.

In addition to algorithmic upper bounds, there has been significant interest in proving Sum-of-Squares (SoS) lower bounds. These include lower bounds for the planted clique problem studied in the works \cite{BHK+16,JPR+21,KPX24}, the densest subgraph problem \cite{JPRX23}, and the graph coloring problem \cite{PX25}. For random hypergraphs, the work \cite{LZ22} shows failure of low-degree polynomial method\footnote{They show this for polynomials up to degree $\bigO\paren{\log n}$.} in detecting planted cliques in random hypergraphs for regimes of $k=o\paren{\sqrt{n}}$, although a Sum-of-Squares lower bound in this regime remains open (see \cite{LZ20} for a discussion).

\paragraph{Sharper Matrix Norm Bounds.}
The asymptotic behavior of random matrices with i.i.d. entries is well-understood in classical random matrix theory literature, see \cite{Tao12}. This is complemented by a large family of non-asymptotic concentration inequalities, such as the widely popular Matrix Chernoff, Bernstein, and Khintchine inequalities \cite{Tro15}. These results apply very broadly to any general linear functions of independent random variables such as those studied in the setting of non-commutative Khintchine inequalities \cite{LPP91,Pis03}, and to matrices with polynomial entries \cite{HSS15, MP16, AMP21,RT23,TW25}. However, the resulting bounds are often not sharp and typically incur mild logarithmic factor dependencies on dimension. 

A recent line of work \cite{BBvH23, BCSvH24}, however, has successfully "shaved off" these extraneous logarithmic factors by leveraging powerful tools from free probability theory, yielding stronger, optimal matrix concentration inequalities. This approach was recently extended by \cite{BLNvH25} to a class of models termed random matrix chaos, where matrix entries are polynomials of independent random variables. They specifically analyze a \emph{chaos of combinatorial type}, which is  relevant for problems in theoretical computer science; they provide "plug-and-play" bounds for parameters that arise in such computations. We utilize this new framework, demonstrating that a natural class of matrices arising in the analysis of hypergraphs fits this model. This allows us to derive logarithmic-free bounds in our results, leading to our improved certificates.

\section{Technical Overview}
\label{sec:proof_overview}
We now explain the main ideas behind the proof of \prettyref{thm:main-theorem-informal}, we refer to \prettyref{sec:sos-main-results} for a detailed and formal treatment. The starting point is a standard polynomial system formulation of the independence number ($\alpha(H)$), the size of Maximum Independent Set (MIS) in $\ell$-uniform hypergraph. Let $\set{x_v}_{v\in V}$ be boolean variables indicating membership in the independent set. The MIS problem can then be encoded as a natural NP-hard polynomial system:
\newcommand{\pmis}{\mathcal{P}_{\text{MIS}}}
\begin{equation}
	\label{eq:pmis-overview}
	\pmis:\qquad
	\max \; \sum_{v\in V} x_v,
	\quad\text{s.t.}\quad
	x_v^2=x_v, \;\;\forall v\in V;
	\qquad
	\prod_{u\in e} x_u = 0, \;\;\forall e\in E.
\end{equation}
Any $0$-$1$ solution of $\pmis$ corresponds to an independent set, and conversely. A degree-$2\ell$ Sum-of-Squares (SoS) relaxation of $\pmis$ returns a \emph{pseudo-expectation} $\pE$  that acts like expectations on polynomials of degree at most $2\ell$, and satisfies all degree-$2\ell$ consequences of the constraints of $\pmis$. The resulting SoS optimum, formally defined  as,
\begin{align*}
    \alpha_{\text{SoS}}(H) \defeq  \max_{\substack{\pE \text{ feasible for degree-}2\ell\\ \text{SoS relaxation of }\pmis}}\pE\Brac{\sum_{v\in V}x_v},
\end{align*}
always upper bounds the true independence number.
Therefore, by SoS duality, an \emph{SoS certificate} that $\alpha_{\text{SoS}}(H) \leq k$ is equivalent to proving that for every degree-$2\ell$ pseudo-expectation $\pE$ feasible for the polynomial system \prettyref{eq:pmis-overview},  must satisfy,
\begin{align*}
	\pE\Brac{X} \leq k, \qquad \text{where} \qquad X \defeq \sum_{v\in V}x_v.
\end{align*}

\paragraph{Analysis via High-Degree Moments.}
Now how does one upper bound $\pE\Brac{X}$? The proof revolves around analyzing the $\ell^{th}$ power of the size of potential independent set given by $X$. Intuitively, if the independent set is large, the corresponding higher-degree moments would be extremely large which one can show does not happen (with high probability) in a random hypergraph. Formally, we will use the following SoS-friendly argument about moments (see  pseudo-expectation \Holder\, inequality, \prettyref{fact:pe_holder}),
\begin{align}\label{eq:pe-holder-overview}
	\pE\Brac{X} \leq \paren{\pE\Brac{X^{\ell}}}^{1/\ell}.
\end{align}
So the real work is bounding $\pE\Brac{X^{\ell}}$. One can expand $X^{\ell}$ as a sum over $\ell$-tuples of vertices and split the sum by grouping according to how many distinct vertices appear in a tuple. Because of booleanity $x_v^2=x_v$, the monomials that repeat at least one variable, effectively behave like monomials of smaller degree. Therefore one can show with a simple counting argument (formalizable within low-degree SoS) that the total contribution of all such non-all-distinct monomials is bounded,
\begin{align}\label{eq:lower-order-overview}
	\sum_{\substack{(v_1,\dots,v_\ell)\in V^\ell\\ \abs{\{v_1,\dots,v_\ell\}}<\ell}}x_{v_1}\cdots x_{v_\ell} \;\;\le\;\; C_\ell \cdot X^{\ell-1}, \qquad \paren{\text{see \prettyref{prop:lower_term_bounds}}},
\end{align}
where $C_{\ell}$ is a constant (depending only on $\ell$).
Thus, the only genuinely hard part is the all-distinct term (denoted $Y_{\ell}$).
Controlling this term is where the structure of constraints of $\pmis$, and randomness of input hypergraph plays a crucial role.

\subsection{Analysis for Even Arity Case}
We discuss the even-arity case first, and to keep the discussion simple we focus on the simplest non-trivial case of $\ell=4$. As discussed above, using $x_v^2=x_v$, the lower-order terms are bounded by $O(X^3)$ and expanding $X^4$ and separating terms, we have,
\begin{align*}
	X^4 = \sum_{\paren{a,b,c,d}\in V^4}x_a x_b x_c x_d = O(X^3) + \sum_{\substack{a,b,c,d \in V\\\abs{a,b,c,d}=4}}x_a x_b x_c x_d = O(X^3) + Y_4.
\end{align*}
\paragraph{From all-distinct monomials to quadratic forms.}
The key step in bounding the order-$\ell$ all distinct term (the quartic term here) is that one can convert what at first glance looks like a messy combinatorial sum into a single quadratic form of an appropriate matrix, thus reducing the problem to bounding the quadratic form of some matrix. For even $\ell$, there is a natural way to pair up the $\ell$ variables in all-distinct monomial $(x_ax_bx_cx_d)$, by writing as product $(x_ax_b)(x_cx_d)$. This suggests a natural $p$-biased \emph{flattening matrix}, by \say{flattening} the $4$-uniform structure to matrix indexed by pairs $\cI=\set{a,b:a<b}$, where we define the matrix $B \in \bbR^{\abs{\cI}\times \abs{\cI}}$ as, 
\begin{align*}
	B_{ab,cd} = \begin{cases}
		-\sqrt{\frac{1-p}{p}} \quad &\text{ if $a,b,c,d$ all distinct and }\set{a,b,c,d}\in E\\
		+\sqrt{\frac{p}{1-p}} \quad &\text{ if $a,b,c,d$ all distinct and }\set{a,b,c,d}\notin E\\
		\qquad 0 \quad &\text{ otherwise}
	\end{cases}
\end{align*}
Next, we define the vector $z \in \bbR^{\cI}$ where $z_{ab}=x_ax_b$. Using the hyperedge constraint, which forces $x_ax_bx_cx_d=0$ whenever $\set{a,b,c,d}\in E$, we can rewrite the all-distinct quartic term $Y_4$ as below,
\begin{align*}
	\sum_{a<b<c<d}x_ax_bx_cx_d &= \sum_{\substack{a<b<c<d\\\set{a,b,c,d}\notin E }}x_ax_bx_cx_d
	=\paren{\sqrt{\frac{1-p}{p}}\,}\sum_{a,b,c,d}B_{ab,cd} \cdot x_ax_bx_cx_d\\
	&= \frac{1}{3}\paren{\sqrt{\frac{1-p}{p}}\,}\sum_{\substack{a<b\\c<d}}z_{ab} \cdot B_{ab,cd} \cdot z_{cd} = \frac{1}{3}\Biggl(\sqrt{\frac{1-p}{p}}\,\Biggr)z^{\top} B z\numberthis\label{eq:y-4}.
\end{align*}

\paragraph{SoS analysis reduces problem to spectral norm bounds.}
Once the all-distinct term is rewritten as a quadratic form, it is much nicer to work with and we can bound it using operator norm bounds (derivable in SoS proofs, see \prettyref{fact:operator_sos_bound}),
\begin{align*}
	\sum_{a,b,c,d} z_{ab} \cdot B_{ab,cd} \cdot z_{cd} = z^{\top}B z \leq \norm{B}_2\norm{z}_2^2
\end{align*}
The next observation is that the vector $z$, defined by $z_{a,b}=x_ax_b$ has its norm tightly controlled by the size of candidate independent set. Using booleanity constraint of $\pmis$ we have,
\begin{align}\label{eq:main-difference-aow}
	\norm{z}_2^2 = \sum_{a<b}z_{ab}^2 \leq \sum_{a,b}z_{ab}^2 = \sum_{a,b}x_a^2x_b^2 = \sum_{a,b}x_ax_b  = \paren{\sum_{v}x_v}^2 = X^2
\end{align}
Importantly this step only uses the axioms of $\pmis$, and therefore the proof above is derivable in SoS. Putting together eqn.~\prettyref{eq:y-4}, eqn.~\prettyref{eq:main-difference-aow} and using $1-p \leq 1$ we say,
\begin{align*}
	\pmis \SoSp{x}{4} \set{Y_4 = \sum_{a<b<c<d}x_ax_bx_cx_d \leq \paren{\frac{\norm{B}_2}{\sqrt{p}}} \cdot X^2}.
\end{align*}
Now by SoS duality it follows that our pseudo-expectation $\pE$ also satisfies the analogous inequality,
\begin{align*}
	\pE\left[\sum_{a<b<c<d}x_ax_bx_cx_d\right] \leq \paren{\frac{\norm{B}_2}{\sqrt{p}}}\cdot \pE[X^2]
\end{align*}
Note that $\norm{B}_2$ depends only on the underlying hypergraph instance and not on the formal variables $x_v$ and it can be treated as a fixed quantity and pulled outside the pseudo-expectation. Combining this bound with earlier control on lower-order terms yields an inequality relating the quartic moment of $X$ to its lower-order moments, for a constant $\kappa_1$,
\begin{align}\label{eq:even_case_final_eqn}
	\pE[X^4] \leq \frac{\norm{B}_2}{\sqrt{p}} \cdot \pE[X^2] + \kappa_1 \cdot {\pE[X^3]}.
\end{align}
Let $t=\paren{\pE[X^4]}^{1/4}$ and using the pseudo-expectation \Holder\, inequality (see \prettyref{fact:pe_holder}),
\begin{align*}
	\pE[X] \leq \paren{\pE[X^4]}^{1/4} = t, \quad  \pE[X^2] \leq \sqrt{\pE[X^4]}=t^2,\quad  \pE[X^3] \leq \paren{\pE[X^4]}^{3/4}=t^3,
\end{align*}
Substituting these into eqn.~\prettyref{eq:even_case_final_eqn} solving the inequality for $t>0$,
\begin{align*}
	t^4 \leq \frac{\norm{B}_2}{\sqrt{p}} \cdot t^2 + \kappa_1 \cdot t^3, \text{ dividing by $t^2$ we get, } t^2\leq \frac{\norm{B}_2}{\sqrt{p}} + \kappa_1 \cdot t
\end{align*}
which is a simple quadratic inequality in $t$ that implies that $\pE[X]=t \leq \kappa_2\cdot \sqrt{\norm{B}_2}/p^{1/4}$. 
\paragraph{Bounding the spectral norm of Flattening Matrix}
Up to this point, the analysis did not rely on the randomness in the hypergraph; the bound on $\pE[X]$ holds deterministically for any instance. However, for a random $4$-uniform hypergraph, we can go a step further and obtain a high-probability bound on the operator norm $\norm{B}_2$. The matrix $B$ resembles a Wigner-type random matrix: it has zero-mean, bounded entries (bounded by roughly $1/\sqrt{p}$) with variance $1$. This suggests that non-asymptotic results from random matrix theory should apply. 
A technical subtlety arises because the entries of $B$ are not fully independent (due to symmetries such as $B[ab,cd] = B[ad,bc]$ that introduce mild dependencies). To address this, we impose a global ordering $\prec$ on the vertex set, and use that to decompose into matrices with independent blocks. The details can be found in \prettyref{prop:flattening-random-matrix-bounds} where for $4$-uniform hypergraph one obtains the bound below,
\begin{align*}
	\norm{B}_2 \lesssim_{\ell} n + \sqrt{\frac{\log n}{p}},  \quad \text{ and therefore } \quad  \pE\Brac{X}_{\ell} \lesssim \frac{\sqrt{n}}{p^{1/4}} + \frac{\paren{\log n}^{1/4}}{p^{1/2}},
\end{align*}
which matches our bounds for $\ell=4$ in \prettyref{thm:main-theorem-informal}. The extension to general even-arity case of $\ell=2q$ follows similarly, and we refer to \prettyref{sec:sos-even-analysis} for the details.

\paragraph{Comparison to \cite{AOW15}} 
We note that the work \cite{AOW15} also constructs SoS certificates for upper bounds on the size of independent set in a random hypergraph. However, their approach employs the Feige's popular XOR trick \cite{Fei02} of reducing to the problem of refuting random $k$-XOR, this was also employed earlier in \cite{CGL07}. The work \cite{AOW15} takes this one step further where they show certifiable bounds on polynomials with random coefficients, (we reproduce their proof in \prettyref{app:aow-proof-reconstruct} for completeness).

The main difference in our work from their work is that \cite{AOW15} proves worst-case bounds over all vectors with $\norm{x}_{\infty} \leq 1$ and hence their analysis can only upper bound vector norms that arise in their analysis by the ambient dimension $n$. On the other hand, in our MIS setting, the variables satisfy the booleanity constraint $x_v^2=x_v$, and hence the vector norm can scale with the size of the candidate set instead of $n$ itself; see  eqn.~\prettyref{eq:main-difference-aow} for how the analysis explicitly tracks the support size of candidate independent sets, allowing the analysis to scale with $X$ rather than defaulting to ambient dimension $n$. This viewpoint is also what underlies the improvement of \cite{GKM22}, and consequently allows the spectral certificates in \cite{GKM22}, and our SoS certificates to avoid the extra $n^{1/4}$ loss characteristic of the XOR approach.

\subsection{Analysis for Odd Arity Hypergraphs}
We now turn to the odd-arity hypergraphs and focus on the simplest, $\ell=3$ case here. As before, we begin by expanding $X^3$ and separating out the terms by number of distinct vertices.
\begin{align*}
	X^3 = \sum_{(a,b,c)\in V^3} x_a x_b x_c \;=\; O(X^2) + 6\,Y_3, \quad \text{ where } \quad Y_3 := \sum_{a<b<c} x_a x_b x_c.
\end{align*}
Again the lower-order terms are easy, and the action is in controlling $Y_3$. The key to the even-arity case was that the all-distinct term naturally became a quadratic form in degree-($\ell/2$) monomials. For $\ell=3$, there is no such evenly balanced factorization. Any naive matricization of an odd-uniform hypergraph produces a rectangular $n \times n^2$ object, and the spectral norm of such matrices typically scales with the larger dimension, resulting in operator norm bounds of $n$ rather than $\sqrt{n}$, and therefore this leads to suboptimal estimates. Indeed, if one mimics the even-arity analysis in a straightforward way, 
the resulting bounds also carry the inflated dimension as,
\begin{align*}
	\pE[X^3] \leq \kappa_2  \paren{\paren{\frac{\norm{B}_2}{\sqrt{p}}}\cdot \pE[X^2] + \pE[X^2]} \implies \pE[X] \leq  \frac{\norm{B}_2}{\sqrt{p}}.
\end{align*}
As a consequence, the naive approach would only certify a bound on the order of $k = O\!\paren{n/\!\sqrt{p}}$, which is far from the conjectured threshold. This motivates a more refined tensor-based analysis where one works directly with the inherent $\paren{2q+1}$-order tensor.

\paragraph{Tensor Based Analysis: Slices and Chaos}
Instead of trying to represent $Y_3$ itself as a single quadratic form, we can express it as a sum of quadratic forms indexed by a \say{distinguished} vertex. To do this, first fix a global ordering $\prec$ on $V$ to break symmetries, and for each vertex $i$ define a symmetric \emph{tensor slice matrix} $A_i$ of size $n \times n$. The entries of the matrix $A_i[j,k]$ is a $p$-biased indicator of the hyperedge $\set{i,j,k}$ where $i \prec j \prec k$. Using this representation, the odd all-distinct term can be written schematically, for a vector $z\in \bbR^n$ as,
\begin{align*}
	Y_3 = \paren{\sqrt{\frac{1-p}{p}}\,} \sum_{i\in V}x_i \cdot z^{\top}A_iz.
\end{align*}
Now we can square the expression above, and applying Cauchy-Schwarz (derivable in SoS),
\begin{align*}
	Y_3^2 \leq \paren{\frac{1-p}{p}}\cdot \left(\sum_{i} x_i^2\right)\left(\sum_i (z^\top A_i z)^2\right) \;=\; \paren{\frac{1-p}{p}}\cdot X\cdot \sum_i (z^\top A_i z)^2.
\end{align*}
Notice how the key idea above to keep the object as an odd $(2q+1)$-tensor and squaring converts the odd-degree polynomial into an even-degree one. But the term $\paren{z^{\top}A_iz}^2$ now lives in a tensor-square space. So to write it as a quadratic form we let $w=x\otimes x$ (index by pairs) and we have,
\begin{align*}
	\sum_i \paren{z^\top A_i z}^2 = w^{\top}\paren{\sum_i A_i \otimes A_i}w  \defeq w^{\top} \widetilde{M}w.
\end{align*}
Crucially, this sum of Kroneckered product of slice matrices $\widetilde{M}$ is a \emph{matrix polynomial of degree $2$} (a \emph{matrix chaos} of order~$2$) in the independent centered variables $\set{h_e}$ having $p$-biased distribution.

This analytic style of squaring the multilinear form for odd-arity CSPs and  applying a Cauchy-Schwarz trick dates back to early works \cite{GK01,FG01} where they used this to refute random $3$-SAT instances. This has since then been a standard step in subsequent papers \cite{Fei02,CGL07,AOW15,BM16,RRS17,OT23} studying random $\ell$-XOR (for odd $\ell$) and related problems.

\paragraph{Bounding spectral norm of sum of Kroneckered product of slices}
At this point the proof reduces to bounding the spectral norm of the large matrix $\widetilde{M}$ obtained by taking sum of Kronecker product of slices $A_i$, where $A_i$ are the slices derived from the tensor $B$. The first technical snag is that $\widetilde{M}$ can contain contributions of the form $h_e^2$ (as large as $\Theta(1/p)$) from an hyperedge appearing twice. If we directly try to bound $\widetilde{M}$, these entries  may pollute the estimate, and lead to a looser bound. To address this we follow the recipe from prior works \cite{HSS15,RRS17} which collect such terms in a matrix $D$ (\say{swapped diagonal} matrix which keeps coordinates forced by equality of unordered pairs). Now the matrix $\widetilde{M}=M+D$ and we can separate the contributions as,
\begin{align*}
    w^{\top}\widetilde{M}w = w^{\top}Dw +  w^{\top}Mw 
\end{align*}
The contribution of $D$ can be controlled directly with high probability (see \prettyref{prop:high-prob-bound-diagonal}) as,
\begin{align*}
    d_{\max} \defeq \norm{D}_2 = \max_{j,k}\sum_{i}\paren{A_i\Brac{j,k}}^2 \lesssim n + \sqrt{\frac{n\log n}{p}} + \frac{\log n}{p}.
\end{align*}
We note that for our vector $w$ and using booleanity axioms of $\pmis$ (derivable in SoS),
\begin{align*}
    \norm{w}_2^2 = \sum_{j,k}\paren{x_jx_k}^2 = \paren{\sum_j x_j}^2 = X^2
\end{align*}
Putting it together, using operator norm bounds, derivable in SoS (see \prettyref{fact:sos_cauchy_schwarz}) and using $(1-p) \leq 1$
\begin{align}\label{eq:odd-case-sos-derive-overview}
    \pmis \SoSp{x}{3} \set{ Y_3^2 \leq \frac{1}{p} \cdot X \cdot \paren{w^{\top}\widetilde{M}w} \leq  \frac{1}{p} \cdot X \cdot \paren{d_{\max}+\norm{M}_2}\norm{w}_2^2 = \frac{1}{p}\cdot \paren{d_{\max} + \norm{M}_2}\cdot X^3}.
\end{align}
So we focus on obtaining a high-probability bound on spectral norm of $M$.

Spectral bounds for matrices of this type, whose entries are polynomials in independent random variables, have been studied in the works on tensor decomposition \cite{HSS15,MP16,AMP21,RT23,TW25}. However, existing results typically assume independent matrix entries and incur logarithmic (in dimension) losses in the bounds. Even when independence assumptions does not fully hold, such as in our case, one can anyways apply the standard matrix concentration inequalities such as Matrix Bernstein, and for our $\ell=3$ case, the bound we obtain is,
\begin{align*}
	\norm{M}_2 \leq n^{3/2}\sqrt{\log n} + \frac{\paren{\log n}^2}{p}.
\end{align*}
Using bounds on lower-order terms and SoS derivable bounds eqn.~\prettyref{eq:odd-case-sos-derive-overview} for constants $\kappa_3,\kappa_4$,
\begin{align*}
    \pmis \SoSp{x}{3} \set{X^3 \leq \kappa_3\paren{\sqrt{\frac{\norm{M}_2}{p}}} + \kappa_4 \cdot X^2}
\end{align*}
By SoS duality and using pseudo-expectation \Holder\,inequality (\prettyref{fact:pe_holder}) we obtain below,
\begin{align*}
	\pE[X^3] \lesssim  \paren{\sqrt{\frac{\norm{M}_2}{p}}}\sqrt{\pE[X^3]} \implies \pE[X] \lesssim \frac{\sqrt{n}(\log n)^{1/6}}{p^{1/3}} + \frac{\log n}{p^{2/3}}
\end{align*}
The first term carries a $\sqrt{\log n}$ factor, an extraneous dimensional factor as an artifact of our concentration inequalities, but the second term  correctly captures the heavy-tail dependence on $1/p$. Therefore this matches our sparse regime certificates for the odd-arity case in \prettyref{thm:main-theorem-informal}. However in the dense case and moderate regimes, the $\sqrt{\log n}$ factor is truly extraneous. This is evidently confirmed even in the \say{weak} graph reduction folklore certificates, though their dependence on $p$ is highly sub-optimal (scaling as $1/\sqrt{p}$ instead of desired $p^{1/\ell}$).

\paragraph{Logarithmic Free Bounds on Spectral Norm of Sum of Kronecker Product of Slices.}
To eliminate the extra logarithmic factor, we leverage recent advances in high-dimensional probability.
There have been recent attempts such as the works of \cite{BBvH23,BvH24} in an effort to get rid of these logarithmic factors in matrix concentration inequalities. A very recent work \cite{BLNvH25} gives a powerful generic framework towards the same. They show that for random matrix with entries as degree-$r$ polynomials (called chaos of order $r$) of some underlying random variables, logarithmic factors for large regimes of parameters can be avoided. Their analysis uses iterative applications of strong inequalities such as the non-commutative Matrix Khintchtine and Matrix Rosenthal inequality (this is effective even in sparse regimes). There is a catch however, their variance parameter is sub-optimal (see remark A.8 in \cite{BLNvH25}), and obtaining sharp bounds in the sparse setting still remains open.

Crucially for us, they identify a subclass called \emph{combinatorial chaos matrices}, which frequently arise in theoretical computer science applications, for which they show an easy way to bound the various parameters that appear in analysis of such matrices. We can show that our matrix $M$ that has degree-$2$ polynomials of underlying random variables, is also a combinatorial chaos of order-$2$, and hence we can readily use their framework to get the desired logarithmic free dependence. Their proof shows these bounds hold in expectation, but using their ideas one can also show that they hold in high-probability, see proof of \prettyref{thm:moment-chaos-theorem} in \prettyref{sec:matrix-chaos-overview}. Putting it together,
\begin{align*}
    \norm{M}_2 \leq n^{3/2} + \frac{n\paren{\log n}^{5/2}}{p} \implies  \pE[X] \lesssim \frac{\sqrt{n}}{p^{1/3}} + \frac{n^{1/3}\paren{\log n}^{5/6}}{p^{2/3}}.
\end{align*}
These bounds are superior when $p$ is not too small, and equating the variance term in chaos bound with leading term in Bernstein bound gives the crossover threshold $p^{\star}$ in \prettyref{thm:main-theorem-informal} as,
\begin{align*}
    \frac{n\paren{\log n}^{5/2}}{p} = n^{3/2}\sqrt{\log n} \qquad \iff \qquad   p^{\star} = \frac{\paren{\log n}^2}{\sqrt{n}}.
\end{align*}

\subsection{Semirandom Robustness and Planted Recovery}
This robustness is essentially \say{for free} once we work with SoS:
adding hyperedges only adds constraints of the form $\prod_{u\in e}x_u=0$ to \prettyref{eq:pmis-overview}. Therefore the feasible region of the SoS relaxation can only shrink, and the SoS optimum $\alpha_{\text{SoS}}(H)$ can only decrease.
Consequently, any high-probability upper bound proved for $\cH_0(n,\ell,p)$ immediately extends to $\cH_1(n,\ell,p)$.

Finally, we briefly explain why strong refutation certificates imply recovery guarantees like \prettyref{cor:planted-coloring-recover}.
At a high level, refutation says: \emph{after removing the planted structure, there is no other large independent set (or $r$-colorable set) hiding in the instance}. This rules out \say{decoy} solutions that would otherwise confuse rounding algorithms. Concretely, given an SoS pseudo-distribution $\pE$, threshold rounding uses the marginals $\pE\Brac{x_v}$ as soft evidence of membership in the planted set. The refutation bound guarantees that vertices outside the planted set cannot collectively carry too much pseudo-mass
without contradicting the SoS-certified upper bound on the independence number of the residual instance.
A standard pruning step then converts this bias into a set $\widehat{S}$ of size $\approx k$ that captures a $(1-o(1))$ fraction of the planted vertices.

We defer the full rounding analysis (including the extensions to the stronger semirandom corruption model across $S\times (V\setminus S)$) to \prettyref{sec:coloring-recover}. The analysis there for recovery closely follows the analysis for $r$-colorable graphs due to \cite{LPR25}.

\subsection{Quiet Planted Distributions for Hypergraph Independent Set}
\paragraph{Low-Degree Polynomials Framework for Refutation Problems.}
We formally describe the low-degree polynomial framework for refutation problems in \prettyref{sec:ldp-refutation-framework}. For the purpose of this discussion, we give an informal overview and refer to the work \cite{KVWX23} and surveys \cite{KWB19,Wei25} for detailed exposition. The low-degree polynomial computational model considers algorithms that compute polynomial functions of the input variables. It is known that polynomials of degree $D=\Theta\paren{\log n}$ captures fairly interesting class of algorithms, including subgraph counting algorithms, spectral methods, approximate message-passing (AMP), local algorithms on sparse graphs, etc. So showing failure of degree $O(\log n)$ polynomials in \say{hard} regimes, can be seen as an evidence of hardness by ruling out concrete class of algorithms \cite{HKK+26}.

For a refutation problem, the framework in \cite{KVWX23} requires one to construct a planted distribution $\sfP$ on $\ell$-uniform hypergraphs that contains an independent set of size $k$, while at the same time no low-degree polynomial separates it from the null distribution $\sfQ$ which is random hypergraphs $\cH_0\paren{n,\ell,p}$ as per \prettyref{def:random-hypergraph}. A polynomial $f:\bbR^N \rightarrow \bbR$ is said to separate $\sfP$ and $\sfQ$ if,
\begin{align*}\numberthis\label{eq:strong-seperation-ldp}
    \sqrt{\max \set{\mathop{\Var}_{X \sim \sfP}\Brac{f},\mathop{\Var}_{X\sim \sfQ}\Brac{f}}} = o\paren{\abs{\mathop{\bbE}_{X\sim \sfP}\Brac{f} - \mathop{\bbE}_{X \sim \sfQ}\Brac{f}}},
\end{align*}
where $X$ denotes the input hypergraph $H=(V,E)$ encoded as $\set{0,1}^N$ for $N=\binom{n}{\ell}$. By dropping the variance under $\sfP$ and centering our input to obtain $Y$ so that $\bbE_{\sfQ}\Brac{f(Y)}=0$, we can relax eqn.~\prettyref{eq:strong-seperation-ldp} to obtain a more tractable quantity called \emph{low-degree likelihood ratio} or \emph{advantage},
\begin{align*}
    \Adv_{\leq D} \defeq \sup_{f\in \bbR\Brac{Y}:\deg(f) \leq D} \frac{\bbE_{\sfP}\Brac{f}}{\sqrt{\bbE_{\sfQ}\Brac{f^2}}}.
\end{align*}
To rule out separation, it suffices to show $\Adv_{\leq D}\!=\!1+o(1)$, and $\Adv_{\leq D}$ can be explicitly computed,
\begin{align*}
    \Adv^2_{\leq D} = \sum_{\abs{\alpha}=0}^D \paren{\bbE_{Y \sim \sfP}\Brac{h_{\alpha}(Y)}}^2,
\end{align*}
where $h_0,\dots,h_D$ forms an orthonormal basis for $\bbR\Brac{Y}_{\leq D}$ with respect to the null distribution $\sfQ$.

\paragraph{Failure of Canonical Planted Distributions.}
A natural candidate for the planted distribution is a generalization of the planted distribution considered for independent sets in graph in \cite{BHK+16}.
\begin{construction}\label{cons:canonical-planted-distribution-hypergraph}
For an $\ell$-uniform hypergraph the distribution $\sfP'$ can be constructed as below:
\begin{enumerate}
    \item Sample a random hypergraph $\cH_0=(V,E)$ as per \prettyref{def:random-hypergraph}.
    \item Choose a subset $S \subset V$ at random by picking each vertex with probability $\rho = 2k/n$.
    \item Remove every hyperedge $e \in E$ if $e \subseteq S$, i.e. set $X_{e}=0$.
\end{enumerate}
\end{construction}

One can show that this does give a low-degree polynomial lower bound but for regimes of $k=o\paren{\sqrt{n}/p^{1/2\ell}}$. We show this formally in \prettyref{app:analyzing-canonical-planted-distribution-hypergraphs}.

The bound is essentially tight for the canonical planted distribution, since low-degree statistics such as total hyperedge count already distinguish from the null distribution once $k = \omega\paren{\sqrt{n}/p^{1/2\ell}}$. The degree-one part of likelihood ratio corresponds to the statistic $\sum_{e\in \cE_{\ell}}Y_e$ that detects the deficit of hyperedges fully contained in the planted set $S$.

A natural idea is to remove this degree-one signal by modifying the hyperedge probability so that the planted and null models have same expected edge counts. However, this does not make the distribution \say{quiet} enough as other small subgraph statistics, such as path counts like $P_3$ can still distinguish the two distributions.

We note that this is a well-document roadblock, even for the sparse graph setting. We refer to Section 4.2 in \cite{JPR+21} which talks about failure of pseudocalibration, and lack of a \say{quiet} planted distribution. The work \cite{JPR+21} studies SoS lower bounds, and the modern approach to constructing SoS lower bounds employs the pseudocalibration framework of \cite{BHK+16}. The first step of this framework involves constructing a planted distribution that is indistinguishable from a null distribution, and then one can mechanically compute from it, a pseudomoment matrix that satisfies all problem-specific constraints. The final part is then showing that this matrix is p.s.d. However, at the scale of $k=\sqrt{n/p}$, the canonical planted distribution is distinguishable, and previous works \cite{JPR+21,KPX24} have to then work extra hard to construct SoS lower bounds. Also, to the best of our knowledge, we do not have low-degree polynomial lower bounds better then the scale of $\sqrt{n}/p^{1/4}$ in the sparse graph setting. This is also mentioned as an open problem in \cite{JPR+21,Pot22}.

\paragraph{Constructing a Quiet Planted Distribution.}
Informally, a planted distribution is called \textit{quiet} at degree $D$ if it is
supported on instances that have the planted property but it remains indistinguishable from the null distribution by degree-$D$ polynomials. In the low-degree framework this amounts to showing $\Adv^2_{\leq D}\paren{\sfP,\sfQ}=1+o(1)$. As we argued above, for our desired regimes, the canonical distribution is not quiet.
Constructing a computationally quiet planting is a problem specific task, and has been successful in other average-case problems \cite{BKW20,BBK+21,KVWX23}.

However, one may be worried that even if the true refutation threshold is $\sqrt{n}/p^{1/\ell}$, why should we expect a quiet distribution at this scale to exist? Fortunately, this is not a concern since the work \cite{KVWX23} also shows that the quiet planting approach is complete. In other words, its absence implies an existence of a low-degree refutation algorithm. So as a first step towards our goal, we should try to figure out what is the quiet planted distribution for graph case. Here, atleast we have strong evidence in form of SoS lower bounds \cite{JPR+21,KPX24} that the correct threshold is $\sqrt{n}/p^{1/2}$. We will infact explain all our main ideas next in the graph setting itself.

We note that for a candidate quiet distribution, our key computation involves showing that $\Adv_{\leq D}\paren{\sfP,\sfQ}=1+o(1)$. Let $\cE$ denote all candidate edges in a graph and fix some $\alpha \subseteq \cE$. Let $F_{\alpha}$ be the corresponding graph with these set of edges and let $V(\alpha)$ denote the vertex set. Now, for the canonical planted distribution, this is worked out in detail in \prettyref{app:analyzing-canonical-planted-distribution-hypergraphs}, but for our discussion next we recall this computation (interpreting  eqn.~\prettyref{eq:important-equation} for graphs) as,
\begin{align*}
    \Adv^2_{\leq D}\paren{\sfP,\sfQ}-1 &\lesssim C_{2,D}\cdot \sum_{\abs{\alpha}=1}^D\sum_{\abs{V(\alpha)}=2}^{2\abs{\alpha}}n^{\abs{V(\alpha)}}p^{\abs{\alpha}}\paren{\frac{2k}{n}}^{2\abs{V(\alpha)}}\\
    &\lesssim C_{2,D}\cdot \sum_{\abs{\alpha}=1}^D\sum_{\abs{V(\alpha)}=2}^{2\abs{\alpha}}\paren{\frac{k^2p}{n}}^{\abs{V(\alpha)}}\cdot p^{\abs{\alpha}-\abs{V(\alpha)}}
\end{align*}
Naively this expression above is $o(1)$ only when $k=o\paren{\sqrt{n}/p^{1/4}}$. However, as a thought experiment, if we could ignore the last term which involves factors of $p$, heuristically we would have that each term in the summation is $o_n(1)$ for $k=o(\sqrt{n/p})$, the desired scale. Since $p \leq 1$, we could genuinely ignore this term if $\abs{\alpha}-\abs{V(\alpha)} \geq 0$ holds. By a simple double counting argument on the number of edges, a sufficient condition for this to hold is that every subgraph $F$ in the summation above has $\deg(F)\geq 2$. Now this becomes our guiding principle to construct quiet distribution: If a subgraph $F_{\alpha}=\paren{V(\alpha),\alpha}$ has a leaf vertex, we want its coefficient in the summation, which in turn is the square of low-degree moment $\bbE\Brac{h_{\alpha}}^2$, and has to be set to $0$. In the graph setting, the null distribution $\sfQ = \cH_0\paren{n,2,p}$ is the \Erdos-\Renyi graph $\cG\paren{n,p}$. We let the orthonormal basis $\set{h_{\alpha}}_{\alpha \subseteq \cE}$ for $L^2(\sfQ)$ where $h_{\alpha}=\prod_{e\in \alpha}Y_e$. We will next see how to set these coefficients.

For an edge $\set{i,j}$ we consider a general form for quiet planted distribution with edge probability $q_{ab}= \Prob{X_{ij}=1|\chi_i=a,\chi_j=b}$ for $a,b \in \set{0,1}$ and with $q_{10}=q_{01}$. Recall that $\chi_i=1$ means a vertex is in the planted set $S$ and $0$ means it isn't.
For convenience we let $m_{ab}=\bbE\Brac{Y_{ij}\mid \chi_i=a,\chi_j=b}$. We can relate the two quantities as,
\begin{align*}\numberthis\label{eq:distribution-moment-relate}
    m_{ab} =\frac{q_{ab}-p}{\sqrt{p(1-p)}} \qquad \equiv \qquad q_{ab}=p+ \sqrt{p(1-p)} \cdot m_{ab}  
\end{align*}
First note that we are forced to set $q_{11}=0$ by the requirement that $S$ must be an independent set which fixes $m_{11}=-\sqrt{p/(1-p)}$. Now for a subgraph $F$, its coefficient $\nu_F$ can be computed by conditioning on $\chi$ and splitting $\alpha$ into its component edges.
\begin{align*}
    \nu_F &= \bbE_{\sfP}\Brac{\prod_{\set{i,j}\in \alpha}Y_{ij}} = \bbE_{\chi}\Brac{\bbE\Brac{\prod_{\set{i,j}\in \alpha} Y_{ij}} \middle| \chi} = \bbE_{\chi}\Brac{\prod_{\set{i,j}\in \alpha}\bbE\Brac{Y_{ij}\mid \chi}}\\
    &=\bbE_{\chi}\Brac{\prod_{\set{i,j}\in \alpha}\bbE\Brac{Y_{ij}\mid \chi_i,\chi_j}} = \bbE_{\chi}\Brac{\prod_{\set{i,j}\in \alpha}m_{\chi_i,\chi_j}}
\end{align*}
Now suppose subgraph $F$ has a leaf vertex $u$ which is connected only to $v$, then every factor in above expression other than $m_{\chi_u\chi_v}$ is independent of $\chi_u$. So we can pull this term out of the expression, and then we notice that there is a simple way that suffices to get its coefficient to zero, which is to impose that for each fixed $b\in \set{0,1}$,
\begin{align*}
    \bbE_{\chi_u}\Brac{m_{\chi_u,b}}=0.
\end{align*}

Recall that $\chi_u \sim \Ber(\rho)$ and the above gives two equations. First for $b=0$ we have,
\begin{align*}
    0 = \bbE_{\chi_u}\Brac{m_{\chi_u,0}} = \paren{1-\rho}m_{00} + \rho \cdot m_{10}
\end{align*}
Next, for the case where $b=1$, we can obtain a similar equation which gives,
\begin{align*}
    0 = \bbE_{\chi_u}\Brac{m_{\chi_u,1}}  = \paren{1-\rho}m_{01} + \rho \cdot m_{11}
\end{align*}
Now using $m_{01}=m_{10}$ and the value of $m_{11}=-\sqrt{p/1-p}$, one can solve for all three $m_{00},m_{01},m_{11}$ and by eqn.~\prettyref{eq:distribution-moment-relate}, the corresponding $q_{ab}$ values. This explicitly gives us our quiet planted distribution in the graph case, which we state next.

\begin{definition}[Quiet Planted Distribution]\label{def:graph-quiet-planted}
    For a graph $G=(V,E)$, and for a fixed $k \leq n/4$ we define the quiet planted distribution $\widetilde{\sfP}$ as follows:
    \begin{enumerate}
        \item Sample i.i.d. random variables $\chi_1,\dots,\chi_n \sim \Ber(\rho)$ for $\rho= 2k/n$. Let $S \defeq \set{i:\chi_i=1}$.
        \item Conditioned on $\chi=\paren{\chi_1,\dots,\chi_n}$, we sample edges independently as below,
        \begin{itemize}
            \item If $\chi_i=1,\chi_j=1$ no edge between $i$ and $j$.
            \item If exactly one of $\chi_i,\chi_j$ takes value $1$ we let $\pr{}{X_{ij}=1} = p/(1-\rho)$.
            \item If both $\chi_i,\chi_j=0$ we let $\pr{}{X_{ij}=1}=p \cdot \frac{1-2\rho}{(1-\rho)^2}$
        \end{itemize}
        Let ${\sfP}'$ be the resulting planted distribution generated from this process.\label{step:unconditional-quiet}
        \item Let $\cG_k$ denote the good event that $\abs{S} \geq k$ and define $\widetilde{\sfP} \defeq {\sfP'}\paren{\cdot \mid \cG_k}$.
    \end{enumerate}
\end{definition}

\paragraph{Hypergraph Quiet Planted Distribution.} One can naturally extend the quiet distribution for the graphs above to $\ell$-uniform hypergraphs. We proceed similarly by randomly sampling $S$, and then by defining a probability parameter $q_t$ for each $t \in \set{0,1,\dots,\ell}$ as,
\begin{align*}
    q_t = p\paren{1-\paren{-\frac{\rho}{1-\rho}}^{\ell-t}}.
\end{align*}
Let $\cE_{\ell}$ denote the set of all candidate edges in an $\ell$-uniform hypergraph. Then the edges are sampled as $X_e \sim \Ber\paren{q_{t_e}},\forall e \in \cE_{\ell}$ where $t_e = \abs{e \cap S}$, captures how many vertices of an edge lie in $S$. We formally describe the construction and prove its analysis in \prettyref{sec:ldp-hypergraph-quiet-planted}. Now, if we do the computation of $\Adv^2_{\leq D}\paren{\widetilde{\sfP},\sfQ}-1$ with the quiet planted distribution $\widetilde{\sfP}$, we obtain that it is indeed vanishing in the regimes of $k=o\paren{\sqrt{n}/p^{1/\ell}}$.

\section{Technical Preliminaries}
\subsection{Sum-of-Squares Preliminaries}
\label{sec:sos-prelims}
In this section, we review some basic ideas from the Sum-of-Squares (SoS) framework that will be essential for our proofs in the rest of the paper. The Sum-of-Squares hierarchy (also called Lasserre hierarchy) was developed independently in the works of \cite{Sho87,Nes00,Par00,Las01}. Our overview follows the exposition in the lecture notes \cite{BS16}, and the monograph \cite{FKP19}, to which we refer for more details.

\paragraph{Notation}
Let $\mathbf{x}=\paren{x_1,\dots,x_n}$ denote a tuple of $n$ variables (also called indeterminates), and let $\bbR[\mathbf{x}]=\bbR[x_1,x_2,\dots,x_n]$ denote the ring of polynomials in indeterminates $x_1,x_2,\dots,x_n$ with real coefficients. We extend our notation to 
use $\mathbb{R}[\mathbf{x}]_{\leq d}$ to denote the set\footnote{We note that this set of polynomials may no longer form a ring as it is not closed under multiplication.} of polynomials of degree at most $d$. Since, we will be working over the boolean hypercube, denoted $\set{0,1}^n$, our polynomials are multilinear.  We use the notation $\bbR[\mathbf{x}]\slash \set{\paren{x_i^2-x_i}}_{i\in [n]}$, to denote the quotient ring of $\bbR[\mathbf{x}]$ modulo the ideal generated by polynomials $\set{\paren{x_i^2-x_i}}_{i\in [n]}$. We will  use $\Delta$ to denote the probability simplex. 

\begin{definition}[Polynomial System]
	A polynomial system $\mathcal{P}$ is a system of of $m$ polynomial inequality constraints (also called axioms of the system) given by, 
	\begin{align*}
		\mathcal{P} = \set{p_1 \geq 0,p_2 \geq 0, \dots ,p_m \geq 0}.
	\end{align*}
\end{definition}

\begin{definition}[Sum-of-Squares Polynomial]
	A polynomial $p \in \bbR[\mathbf{x}]$ is said to be sum-of-squares polynomial if there exists polynomials $q_1,q_2,\dots,q_r \in \bbR[\mathbf{x}]$ such that we can write,
	\begin{align*}
		p= q_1^2 + q_2^2 + \dots + q_r^2.
	\end{align*}
\end{definition}

\subsubsection{Sum-of-Squares Proofs}
\begin{definition}[Degree-$d$ SoS Proof]
	Given a polynomial system $\mathcal{P}=\set{p_1\geq 0, \dots,p_m \geq 0}$, a degree-$d$ SoS proof of non-negativity of a polynomial $f$ (denoted $\mathcal{P}\SoSp{x}{d} \set{f\geq 0}$) is a list of sum-of-squares polynomials $\set{q_S}_{S\subseteq [m]}$ such that we can write,
	\begin{align*}
		f=\sum_{S\subseteq [m]}q_S\prod_{i\in S}p_i \text{ where for every $S\subseteq [m]$}, \deg\paren{q_S\prod_{i\in S}p_i} \leq d.
	\end{align*}
	We let the convex set (more precisely a convex cone) $\sos_{d}\paren{\mathcal{P}}=\set{f \, \vert \, \mathcal{P}\SoSp{x}{d} \set{f\geq 0}}$.
\end{definition}
\begin{remark}
	If $\mathcal{P} = \emptyset$ we omit it altogether and write $\SoSp{x}{d} \set{f \geq 0}$
\end{remark}

\paragraph{Inference Rules} SoS proofs follow natural rules of inference one expects from a proof system as,
\begin{itemize}
	\item \textbf{Addition Rule.} Given a polynomial system $\mathcal{P}$ and polynomials $f,g \in \bbR[\mathbf{x}]$,
	\begin{align*}\numberthis\label{eq:additiona-inference-sos}
		\frac{\mathcal{P} \SoSp{x}{d} \{ f\ge 0 , g \ge 0\}}
		{\mathcal{P} \SoSp{x}{d} \{ f+g\ge 0\}}\mper
	\end{align*}
	\item \textbf{Multiplication Rule.} Given a polynomial system $\mathcal{P}$ and polynomials $f,g \in \bbR[\mathbf{x}]$,
	\begin{align*}
		\frac{\mathcal{P} \SoSp{x}{d} \set{ f\ge 0},\mathcal{P} \SoSp{x}{d'} \set{ g \ge 0}}
		{\mathcal{P} \SoSp{x}{d+d'} \set{ fg\ge 0}}\mper
	\end{align*}
	\item \textbf{Transitivity Rule.} Given polynomial systems of constraints $\mathcal{P},\mathcal{Q},\mathcal{R}$ ,
	\begin{align*}
		\frac{\mathcal{P} \SoSp{x}{d} \mathcal{P}',~ \mathcal{P}'\SoSp{x}{d'} \mathcal{P}''}
		{\mathcal{P} \SoSp{x}{dd'} \mathcal{P}''}\mper
	\end{align*}
	\item \textbf{Substitution Rule.} Given functions $F:\bbR^n \rightarrow \bbR,G:\bbR^m \rightarrow \bbR,H: \bbR^p \rightarrow \bbR$,
	\begin{align*}
		\frac{\{F\ge 0\} \SoSp{x}{d} \{ G \ge 0\}}
		{\{F(H))\ge 0\} \SoSp{x}{d\cdot \deg(H)} \{ G(H)\ge 0\}}\mper
	\end{align*}
\end{itemize}
\subsection*{Basic Sum-of-Squares Proofs}
We recall some well-known sum-of-squares proofs that will be useful for our proofs.
\begin{fact}[Operator Norm Bounds, Fact 3.16 in \cite{BK20}]\label{fact:operator_sos_bound}
	Given a symmetric matrix $M \in \bbR^{n\times n}$ and a vector of indeterminates $\mathbf{x}\in \bbR^{n}$ it holds that,
	\begin{align*}
		\SoSp{x}{2} \set{\mathbf{x}^TM\mathbf{x} \leq \norm{M}_2\norm{\mathbf{x}}_2^2}.
	\end{align*}
\end{fact}
\begin{fact}[SoS Cauchy-Schwarz, Fact 3.17 in \cite{BK20}]
	\label{fact:sos_cauchy_schwarz}
	Let $x_i,y_i$ for $1 \leq i \leq n$ be indeterminates. Then,
	\begin{align*}
		\SoSp{x,y}{4} \set{\paren{\sum_{i=1}^nx_iy_i}^2 \leq \paren{\sum_{i=1}^nx_i^2}\paren{\sum_{i=1}^ny_i^2}}.
	\end{align*}
\end{fact}

\begin{fact}[Cancellation Within SoS, Lemma 9.3 in \cite{BK20}]
	Let $x,C \in \bbR$ be indeterminates. Then,
	\begin{align*}
		\set{x \geq 0,x^d \leq Cx^{d-1}} \SoSp{x,C}{2d}\set{x^{2d} \leq C^{2d}}
	\end{align*}
\end{fact}

\begin{fact}[Taking Roots, Lemma A.3 in \cite{KS17}]
	Let $x \in \bbR$ be an indeterminate and let $C \geq 0$ be a constant. Then,
	\begin{align*}
		\set{x^{2d} \leq C^{2d}} \SoSp{x}{2d} \set{x \leq C}.
	\end{align*}
\end{fact}

\subsubsection{Pseudo-distributions and Pseudo-expectations}
A finitely supported distribution over $\bbR^n$ can be represented by a probability mass function $\mu:\bbR^n \rightarrow \bbR_{\geq 0}$ such that $\mu(\mathbf{x})$ is the probability of a point $\mathbf{x} \in \bbR^n$ under the distribution $\mu$. Similarly for a finitely supported $\mu$, the expectation for a function $f:\bbR^n \rightarrow \bbR$ is simply,
\begin{align*}
	\E_{\mu}\Brac{f} = \sum_{\mathbf{x} \in \mathsf{supp}\paren{\mu}}f(\mathbf{x})\mu(\mathbf{x})
\end{align*}
Pseudo-distributions are generalizations of finitely supported probability distributions where one relaxes the condition that $\mu\paren{\mathbf{x}}\geq 0$ but still require it to satisfy a certain kind of low-degree non-negativity tests.
\begin{definition}[Degree-$d$ Pseudo-distribution]
	A finitely supported function $\mu:\bbR^n\rightarrow \bbR$ is a degree-$d$ pseudo-distribution if it satisfies the following,
	\begin{itemize}
		\item \textbf{Normalization}: $\sum\limits_{\mathbf{x}}\mu(\mathbf{x})=1$.
		\item \textbf{Non-Negativity of Squares}: $\sum\limits_{\mathbf{x}}\mu(\mathbf{x})f(\mathbf{x})^2 \geq 0$ for every polynomial $f$ where $\deg(f) \leq d/2$.
	\end{itemize}
\end{definition}
It is easy to see that any true distribution is a pseudo-distribution but the other direction is not true.
\begin{definition}[Pseudo-expectation]
	We define the degree-$d$ pseudo-expectation of a polynomial $f$ (with degree at most $d$) with respect to a degree-$d$ pseudo-distribution $\mu$ as the formal expectation,
	\begin{align*}
		\mathop{\pE}\limits_{\mu}[f] = \sum_{\mathbf{x}}\mu(\mathbf{x})f\paren{\mathbf{x}}
	\end{align*}
\end{definition}
Whenever clear from context we will drop the pseudo-distribution and work directly with the pseudo-expectation operator. The degree-$d$ moment tensor (captures pseudo-expectation of all monomials with degree upto $d$) of a pseudo-distribution $\mu$ is given by $\pE_{\mu}\paren{1,\mathbf{x}}^{\otimes d}$. Similar to the degree-$d$ moment tensors of a true distribution, the degree-$d$ moment tensors of a degree-$d$ pseudo-distribution also form a convex set. It can in fact be shown that,
\begin{fact}
	A finitely supported $\mu:\bbR^n \rightarrow \bbR$ with $\sum_{\mathbf{x}}\mu(\mathbf{x})=1$ is a degree-$d$ pseudo-distribution (for an even $d$) iff the formal degree-$d$ moment matrix (called pseudomoment matrix) is positive semidefinite, 
	\begin{align*}
		\mathop{\pE}\limits_{\mu}\left[\paren{\paren{1,\mathbf{x}}^{\otimes d/2}}\paren{\paren{1,\mathbf{x}}^{\otimes d/2}}^T\right] \succeq 0.
	\end{align*}
\end{fact}

\begin{definition}[Constrained Pseudo-distributions]
	Given a degree-$d$ pseudo-distribution $\mu$ and a polynomial system $\mathcal{P}=\set{p_1\geq 0,\dots,p_m\geq 0}$ we say that $\mu$ satisfies $\mathcal{P}$ at degree $d'$ (denoted $\mu \models_{d'} \mathcal{P}$) if for every sum-of-squares polynomial $q$ it holds that,
	\begin{align*}
		\mathop{\pE}\limits_\mu \left[q\prod_{i\in S}p_i\right] \geq 0, \forall S \subseteq [m] \text{ satisfying }\paren{\deg(q)+\sum_{i\in S}\max\set{\deg(p_i),d'}} \leq d
	\end{align*}
\end{definition}
We say $\mu$ satisfies $\mathcal{P}$ (denoted by $\mu \models \mathcal{P}$) if $\mu$ satisfies $\mathcal{P}$ at degree $0$. Note that for a true distribution it holds that $\mu \models \mathcal{P}$ iff $\mu$ is supported on a subset of the set of solutions to the constraints given by $\mathcal{P}$. Now if for a polynomial $f$ we have $\mu$ satisfying $\mu \models_d\set{f\geq 0}$ then for any degree $d' \geq d$ pseudo-distribution, $\pE_{\mu}[f]\geq 0$. If instead we have an equality $f=0$ then a degree $d$ pseudo-distribution $\mu$ satisfies $f=0$ is equivalent to $\pE_{\mu}\left[f.p\right]=0$ whenever $\deg(f) + \deg(p) \leq d$. 

\begin{remark}\label{rem:equality_satisfying}
	In particular, if we take the constant polynomial $p=1$ we have that for any degree at most $d$ polynomial equality $f=0$ a degree-$d$ pseudo-distribution satisfies $\pE_{\mu}[f]=0$. 
\end{remark}

\begin{remark}
	We say that a degree-$d$ pseudo-distribution $\mu$ satisfies a given polynomial system $\varepsilon$-approximately if for every SoS  polynomial $q$ it holds that,
	\begin{align*}
		\pE_{\mu}\left[q\prod_{i\in S}p_i\right] \geq -\varepsilon\norm{q}_2\prod_{i\in S}\norm{p_i}_2\quad  \forall S \subseteq [m] \text{ satisfying }{\deg(q)+\sum_{i\in S}{\deg(p_i)}} \leq d
	\end{align*}
	where for a polynomial $p$ the notation $\norm{p}_2$ denotes the Euclidean norm of the coefficient vector.
\end{remark}

\subsubsection*{Basic Facts about Pseudo-distributions and Pseudo-expectations} We recall well-known pseudo-distribution/pseudo-expectaton facts that will be useful for our proofs.

\begin{fact}[Pseudo-expectation H\"older Inequality, Fact 3.11 in \cite{BK20}]\label{fact:pe_holder}
	Let \, $f,g$ be polynomials of degree at most $d$. Fix a $t \in \mathbb{N}$, then for any degree-$dt$ pseudo-expectation the following holds,
	\begin{align*}
		\pE\left[f^{t-1}g\right] \leq \paren{\pE\left[f^{t}\right]}^{(t-1)/t}\paren{\pE[g^t]}^{1/t}.
	\end{align*}
\end{fact}

\begin{fact}[Comparisons of Norms, Fact 3.12 in \cite{BK20}]\label{fact:pe_norm} For $t,t' \in \mathbb{N}$, and a degree-$t^2$ pseudo-expectation $\pE$ over a scalar indeterminate $x$. Then it holds that,
	\begin{align*}
		\pE\left[x^{t'}\right]^{1/{t'}} \leq \pE\left[x^t\right]^{1/t}, \forall t \geq t'.
	\end{align*}
\end{fact}

\subsubsection{Duality of SoS Proofs and Pseudo-distributions/Pseudo-expectations}
Pseudo-distributions and Sum-of-Squares proofs are dual objects to each other. In our proofs we will reason about these in a composable fashion by using the soundness and completeness of sum-of-squares proof system. Low degree SoS proofs are sound and complete proof system if we take low-degree pseudo-distributions as models. Therefore, SoS proofs allow us to deduce properties of pseudo-distributions that satisfy some constraints.

\begin{lemma}[Soundness]
	Let $\mu$ be a degree-$d$ pseudo-distribution satisfying a polynomial system $\mathcal{P}$ at degree $d''$ and $f$ be a polynomial such that $\mathcal{P} \SoSp{x}{d'} \set{f \geq 0}$ then we have that $\mu \models_{d'\cdot d''+d'} \set{f \geq 0}$.
\end{lemma}
\begin{corollary}\label{cor:sos_soundness}
	For a polynomial system $\mathcal{P}$ where $\mathcal{P} \SoSp{x}{d'} \set{f\geq 0}$ and a degree-$d$ pseudo-distribution $\mu$ where $d\geq d'$ such that $\mu \models \mathcal{P}$ then we have that $\pE_{\mu}\Brac{f} \geq 0$.
\end{corollary}

The soundness is akin to weak duality and we may also hope for strong duality, also called completeness which says that every property of low-degree pseudo-distribution can be derived by low-degree sum-of-squares proof. Strong duality holds under additional assumption that the polynomial system\footnote{Also, additional axioms such as $x_i^2=x_i$ need to be Archimedean, which is indeed true over a hypercube.} is explicitly Archimedean (has constraint of the form $\sum_i x_i^2 \leq R$). Intuitively, this completeness of the SoS proof system is the reason that a true polynomial inequality derivable from $\mathcal{P}(\mathbf{x})$ has a SoS derivation (though with the caveat that it may not be possible to efficiently find it \cite{OD17,RW17}).

\begin{lemma}[Completeness]
	For an Archimedean polynomial system $\mathcal{P}$ having polynomials with degree at most $d'$, and for any degree-$d$ pseudo-distribution that satisfies $\mu \models_{d'} \mathcal{P}$, and for a polynomial constraint $\set{f \geq 0}$ where $\mu \models_{d''} \set{f \geq 0}$, then there is a degree-$d$ sum-of-squares proof showing that $\mathcal{P} \SoSp{x}{d} \set{f \geq -\varepsilon}$ for $d \geq d' \geq d''$.
\end{lemma}
We can simplify above by letting $d',d''=0$ and package this into a concise duality statement.
\begin{lemma}[Duality]
	For an Archimedean polynomial system $\mathcal{P}$, and for any degree-$d$ pseudo-distribution (for an even $d$) and 
	polynomial $f$ exactly one of the following holds,
	\begin{enumerate}
		\item For every $\varepsilon >0$, there exists a degree-$d$ SoS proof such that $\mathcal{P} \SoSp{x}{d} \set{f \geq -\varepsilon}$.
		\item There exists a degree-$d$ pseudo-distribution $\mu$ such that $\mu \models \mathcal{P}$ and $\pE_{\mu}[f]<0$.
	\end{enumerate}
\end{lemma}

\subsubsection{Optimizing over Pseudo-distributions/Pseudo-expectations}
Next, we consider the task of optimizing a function $f(\mathbf{x})$ over a system of polynomial inequalities $\mathcal{P}(\mathbf{x})$ in the boolean hypercube $\set{0,1}^n$ (expressible as a polynomial $x_i^2=x_i,\forall i \in [n]$. Let $\mathcal{K}$ be the solution set of $\mathcal{P}(\mathbf{x})\cup \set{x_i^2=x_i}_{i\in [n]}$, , and $\mathcal{E}\paren{\mathcal{K}}$ be the set of expectation functions defined on the convex hull of $\mathcal{K}$ i.e., $\mathcal{E}(\mathcal{K}) = \set{\E_{\nu}|\nu \in {\Delta(\mathcal{K})}}$. Similarly, we let the set of all degree-$d$ pseudo-expectation operators over $\mathcal{K}$ as $\mathcal{E}_{d}\paren{\mathcal{K}}$.
The polynomial optimization problem given by \prettyref{mp:poly_opt} can then be relaxed in the following fashion,
\begin{align*}
	\min_{\mathbf{x}\in \mathcal{K}}f(\mathbf{x}) =
	\min_{\E \in \mathcal{E}(\mathcal{K})}\E\Brac{f(\mathbf{x})}
	\geq \min_{\pE \in \mathcal{E}_{d}\paren{\mathcal{K}}} \pE\Brac{f(\mathbf{x})}
\end{align*}
by first optimizing for true expectations over the convex hull of feasible solutions\footnote{This doesn't change the objective value but is still intractable as specifying a distribution takes exponential space.} and then relaxing it to a pseudo-expectation as defined earlier. Since, for $d\geq 1$, any true expectations is also degree-$d$ pseudo-expectations, the program \prettyref{mp:relax} is a relaxation and commonly referred to as degree-$d$ Lasserre/Parrilo SoS Relaxation.

\begin{definition}[Lasserre/Parrilo SoS Relaxation]\label{def:sos_relaxation}
	Given a polynomial optimization problem of the form \prettyref{mp:poly_opt}, we consider a degree-$d$ SoS relaxation \prettyref{mp:relax} as, 
	\begin{tcolorbox}[left=2pt, right=2pt, top=4pt, bottom=-15pt]
		\small
		\begin{minipage}[t]{0.49\textwidth}
			\begin{MP}[Polynomial Optimization Problem]
				\label{mp:poly_opt}
				\begin{align*}
					\min \quad & f(\mathbf{x}) \\
					\text{subject to} \quad & \\
					p_i(\mathbf{x}) \geq & \, 0 \quad \,\,\forall i \in [m] \\
					x_i^2 = & \, x_i \quad \forall i \in [n]
				\end{align*}
			\end{MP}
		\end{minipage}
		\hfill
		\begin{minipage}[t]{0.49\textwidth}
			\begin{MP}[Lasserre/Parrilo SoS Relaxation]
				\label{mp:relax}
				\begin{align*}
					\min \quad & \pE[f(\mathbf{x})] \\
					\text{subject to} \quad & \\
					\pE \in & \, \mathcal{E}_{d}(\mathcal{K}) \\
					& \\ %
				\end{align*}
			\end{MP}
		\end{minipage}
	\end{tcolorbox}  
\end{definition}

\noindent
This relaxation in \prettyref{mp:relax} is tractable because the set $\mathcal{E}_{d}(\mathcal{K})$ admits a weak separation oracle (in \cite{GLS81} sense) and we have the following Sum-of-Squares (SoS) algorithm guarantees attributed to \cite{Sho87,Nes00,Par00,Las01}.

\begin{theorem}[Theorem 3.3 in \cite{MST16}]
	\label{thm:sos_tractable_solve}
	Given an optimization problem in the form given by \prettyref{mp:relax}, we can compute a degree-$d$ pseudo-expectation solution\footnote{The solution satisfies inequalities upto a small additive error.} for $\mathcal{P}(\mathbf{x})$ over the hypercube $\set{0,1}^n$ (up to arbitrary accuracy) in time $(m+n)^{O(d)}$.
\end{theorem}

\subsection{Tensor Analysis and Kronecker Product Properties}
We collect several standard definitions and properties of tensors and Kronecker products that are useful for the purpose of our analysis later.

\begin{definition}[Kronecker Product, Definition 4.2.1 in \cite{HJ91}] \label{def:kronecker-product-definition}
	Given matrices $A \in \bbR^{m \times n}$ and $B \in \bbR^{p \times q}$, the Kronecker product, denoted $A \otimes B$ is an $mp \times nq$ matrix represented in block form as,
	\begin{align*}
		A \otimes B = \begin{bmatrix}
			A_{11}B &\dots  &A_{1n}B\\
			\vdots &\ddots  &\vdots\\
			A_{m1}B &\dots &A_{mn}B
		\end{bmatrix}
	\end{align*}
	If $A \in \bbR^{n_1 \times n_1}$ and $B \in \bbR^{n_2\times n_2}$ then $A \otimes B$ is an $\paren{n_1n_2} \times \paren{n_1n_2}$ matrix where an entry of it can be indexed by a pair of pairs as,
	\begin{align*}
		A\otimes B [\paren{i,j},\paren{k,\ell}] = A[i,k] \cdot B[j,\ell]
	\end{align*}
\end{definition}

\begin{fact}[Kronecker Product Norm, Theorem 4.2.12 in \cite{HJ91}]\label{fact:kronecker-product-norm}
	Given matrices $A,B$, the spectral norm of the Kronecker product is the product of spectral norms,
	\begin{align*}
		\norm{A \otimes B}_2 = \norm{A}_2 \cdot \norm{B}_2 \text{ and hence } \norm{A \otimes A}_2 = \norm{A}_2^2
	\end{align*}
\end{fact}

\begin{fact}[Kronecker Product and Transpose,  Fact 4.2.4 in \cite{HJ91}]\label{fact:kronecker-product-transpose} For matrices $A,B$ we have that, $\paren{A \otimes B}^{\top} = A^{\top} \otimes B^{\top}$.
\end{fact}

\begin{fact}[Mixed Product Property, Lemma 4.2.10 in \cite{HJ91}]\label{fact:mixed-product-property}
	For our given matrices $A \in \bbR^{m \times n}, B \in \bbR^{p \times q},C \in \bbR^{n \times k}, D \in \bbR^{q \times r}$, then we have $\paren{A \otimes B}\paren{C \otimes D} = AC \otimes BD$
\end{fact}

\begin{claim}[Kronecker Quadratic Form]\label{claim:kronecker-quadratic-form}
	For a vector $z \in \bbR^n$ and matrix $M \in \bbR^{n\times n}$ we have,
	\begin{align*}
		\paren{z^{\top}M z}^2 = \paren{z \otimes z}^{\top}\paren{M \otimes M}\paren{z \otimes z}.
	\end{align*}
\end{claim}

\begin{proof}
	We start with the right hand side of the inequality and simplify by using  \prettyref{fact:mixed-product-property} with $A,B=M$ and $C,D=z$ and we obtain that,
	\begin{align*}
		\paren{z \otimes z}^{\top} \paren{M \otimes M}\paren{z \otimes z} = \paren{z \otimes z}^{\top} \left[ \paren{M \otimes M}\paren{z \otimes z}\right] = \paren{z \otimes z}^{\top}\left[\paren{Mz} \otimes \paren{Mz}\right]
	\end{align*}
	and using \prettyref{fact:kronecker-product-transpose} and a vector form of \prettyref{fact:mixed-product-property} with $a=b=z^{\top}$ and $c=d=Mz$ we have,
	\begin{align*}
		\paren{z \otimes z}^{\top}\left[\paren{Mz} \otimes \paren{Mz}\right] = \paren{z^{\top} \otimes z^{\top}} \left[\paren{Mz} \otimes \paren{Mz}\right] = \paren{z^{\top} M z} \cdot \paren{z^{\top} Mz} = \paren{z^{\top}Mz}^2.
	\end{align*}
\end{proof}
\vspace{-0.1ex}
\begin{fact}[Norm of Powers]\label{fact:norm-of-powers}
	For any symmetric matrix $M$ we have $\norm{M^2}_2 = \norm{M^{\top}M} = \norm{M}_2^2$ and thus inductively for any $k \in \mathbb{N}$ we have, $\norm{M^k}_2 = \norm{M}_2^k$.
\end{fact}

\subsection{\texorpdfstring{$p$}{p}-Biased Fourier Characters for Random Hypergraphs}
\label{sec:p-biased-character}
Let $V$ be the vertex set and $\cE_{\ell} \defeq \binom{V}{\ell}$ be the set of all possible $\ell$-uniform hyperedges. We can encode a hypergraph $H=(V,E)$ using edge-indicator variables as,
\begin{align*}
	Y = \paren{Y_e}_{e \in \mathcal{E}_{\ell}} \in \set{0,1}^{\abs{\mathcal{E}_{\ell}}},  \text{ where } Y_e = {\one}_{e\in E}.
\end{align*}
Under our random hypergraph model $\cH_0\paren{n,\ell,p}$, the family $\set{Y_e}_{e\in \cE_{\ell}}$ is i.i.d. $\mathsf{Ber}(p)$. For each $e \in \cE_{\ell}$ we define the scalar $p$-biased Fourier character (see Section 8.4 in \cite{OD14}) as,
\begin{align*}
	h_e = \frac{p-Y_e}{\sqrt{p(1-p)}}, \quad \text{ or equivalently } \quad h_e = \begin{cases}
		-\sqrt{\frac{1-p}{p}}, \quad & \text{if } e\in E\\
		\,\,\sqrt{\frac{p}{1-p}}, \quad & \text{if } e\notin E
	\end{cases}
\end{align*}
\begin{fact}\label{fact:p-biased-properties}
	For $h_e$ as defined above it holds that, $\E\Brac{h_e}=0$, $\E\Brac{h_e^2}=1$ and $\abs{h_e} \leq 1/\sqrt{p\paren{1-p}}$.
\end{fact}
For a set of hyperedges $\alpha \subset \cE_{\ell}$ we define,
\begin{align*}
	h_{\alpha} = \prod_{e\in \alpha}h_e \qquad \paren{\text{ where } h_{\emptyset} \defeq 1}.
\end{align*}
Then $\set{h_{\alpha}}_{\alpha \subset \cE_{\ell}}$ forms an orthonormal basis for the set of functions on $\set{0,1}^{\abs{\cE_{\ell}}}$ with respect to the product measure $\mu_p = \mathsf{Ber}(p)^{\otimes \abs{\cE_{\ell}}}$.

\subsection{Overview of Matrix Chaos Inequalities and \texorpdfstring{\cite{BLNvH25}}{BLNvH25}}
\label{sec:matrix-chaos-overview}
\begin{definition}[Matrix Chaos of order-$r$]\label{def:matrix-chaos-order-r}
	Fix integers $r,m,d \geq 1$, let $h=\set{h_u}_{u\in [m]}$ be independent, zero-mean, unit variance random variables. A matrix chaos of order $r$ is defined as,
	\begin{align*}
		Y = \sum_{u_1,\dots,u_r\in [m]}h_{u_1}\dots h_{u_r} \cdot A_{u_1,\dots,u_r} \in \bbR^{d\times d},
	\end{align*}
	where $A_{u_1,\dots,u_r} \in \bbR^{d\times d}$ are deterministic coefficient matrices.
\end{definition}
It is also convenient to view the matrix coefficients $A_{u_1,\dots ,u_r,u_{r+1},u_{r+2}}$ as an order-$\paren{r+2}$ tensor $\cA$,
\begin{align*}
    \cA_{u_1,\dots,u_r,u_{r+1},u_{r+2}} = \paren{A_{u_1,\dots,u_r}}\Brac{u_{r+1},u_{r+2}}
\end{align*}
The first $r$ indices here range from $1$ to $m$ are called \emph{chaos coordinates} and the last two indices range from $1$ to $d$ and are called \emph{matrix coordinates}.

A particularly structured and useful subclass of chaos matrices studied in \cite{BLNvH25} is \emph{chaos of combinatorial type}, where every coefficinet matrix is a tensor products of canonical basis vectors. 
\begin{definition}[Matrix Chaos of Combinatorial type, Definition 3.3 in \cite{BLNvH25}]\label{def:combinatorial_chaos}
	Fix integers $r,f \geq 1$ and positive integers $T_1,\dots,T_f$. Let $\set{h_u}_{u\in [m]}$ be independent copies of zero-mean, unit-variance scalar random variable $h$. For each $g \in [r+2]$ let $I_g$ be an index map from $[T_1] \times \dots \times [T_f] \rightarrow [m]$, and $t=(t_1,\dots,t_f)$.
	Then a \emph{matrix chaos of combinatorial type}  of order $r$ is defined as,
	\begin{align*}
		Y=\sum_{t\in [T_1] \times \dots \times [T_f]}h_{I_1(t)}\dots h_{I_r(t)} \paren{e_{I_{r+1}(t)} \otimes e_{I_{r+2}(t)}^{\top}}
	\end{align*}
	where chaos coordinates $I_1(t),\dots,I_r(t)$ are ordered subsets of $[f]$, and matrix coordinates given by $I_{r+1}(t),I_{r+2}(t)$ are also ordered subsets of $[f]$, and $e_a$ denotes the $a^{th}$ standard basis vector.
\end{definition}

\subsubsection{Matrix Flattening Parameter Computations}
For the tensor $\cA$ where $\cA_{u_1,\dots,u_r,u_{r+1},u_{r+2}}=A_{u_1,\dots,u_r}[u_{r+1},u_{r+2}]$, they consider various reshaping of tensor into flattened matrices. They consider varous flattenings, and the bounds in \cite{BLNvH25} are stated in terms of norms of certain flattenings. Here, we record a few definitions we need.

\paragraph{$\sigma$-flattenings and parameter $\sigma(\cA)$.}
The simplest flattening, $\sigma$-flattening which reshapes the coefficient tensor into matrix while preserving the original row and column structure so that,
\begin{align*}
	\cA_{[R|C]} = \sum_{\substack{u_1,\dots,u_r\in [m]\\u_{r+1},u_{r+2}\in[d]}} \paren{\mathop{\otimes}\limits_{g\in R}e_{u_g}} \otimes \paren{\mathop{\otimes}\limits_{g\in C}e_{u_g}}\cA_{u_1,\dots,u_{r+2}},
\end{align*}
for any $R,C \subseteq [r+2]$, where $m$ denotes the number of base random variables $\set{h_u}_{u\in [m]}$ and $d$ denotes the matrix dimension.
For our flattenings of interest we will always have $R=[r+2]\setminus C$. Now in a $\sigma$-flattening we always have $r+1\in R$ and $r+2\in C$ (preserve original matrix coordinates) but chaos coordinates can be arbitrary partitioned and hence leads to many possible $\sigma$-flattenings. They define a parameter $\sigma\paren{\cA}$ as the largest spectral norm of all $\sigma$-flattenings as,
\begin{align*}
	\sigma\paren{\cA} = \max_{\substack{R=[r+2]\setminus C\\r+1\in R,r+2\in C}}\norm{\cA_{[R|C]}}_2.
\end{align*}
Intuitively, the parameter $\sigma\paren{\cA}$ captures the expected leading-order behavior of chaos matrix. An important contribution of this work is that for a combinatorial chaos and fixed $R,C$ they give closed form relations for computing these spectral norms,
\begin{proposition}[Proposition 3.4 in \cite{BLNvH25}]\label{prop:bound_chaos_parameters}
	Let $Y$ be a combinatorial chaos of order-$r$ as in \prettyref{def:combinatorial_chaos} and let $R,C\subseteq [r+2]$ and $\cR=\cup_{g\in R}I_g$ and $\cC=\cup_{g\in C}I_g$ then,
	\begin{align*}
		\norm{\cA_{R|C}}_2^2 = \paren{\prod_{b\in \overline{\cR}}T_b}\paren{\prod_{b\in \overline{\cC}}T_b} .
	\end{align*}
	where $\overline{\cR}=[f]\setminus \cR$ and $T_b$ as defined in \prettyref{def:combinatorial_chaos}.
\end{proposition}

\paragraph{$\nu$-flattenings and parameter $\nu(\cA)$.}
They introduce another class of flattening called $\nu$-flattening which reshapes the tensor by moving all original matrix coordinates to the columns while chaos coordinates can still be arbitrarily partitioned and they define a parameter $\nu(\cA)$ as,
\begin{align*}
	\nu\paren{\cA} = \max_{\substack{R=[r+2]\setminus C\\r+1,r+2\in C\\R\neq \emptyset}}\norm{\cA_{[R|C]}}_2.
\end{align*}
\paragraph{Moment parameter.} They define another parameter $\alpha(h)$ for a zero-mean random variable $h$ as,
\begin{align*}
	\alpha(h) = \norm{h}_{L^{\log (d+m)}} = \paren{\E\left[\abs{h}^{\log\paren{d+m}}\right]}^{1/\paren{\log\paren{d+m}}} .  
\end{align*}

\subsubsection{Improved Bounds on Matrix Norms.}
We now state the norm bounds for matrix chaos in terms of the parameters above.
\begin{theorem}[Iterated strong Matrix Rosenthal, Theorem 2.7 in \cite{BLNvH25}]\label{thm:chaos_thm}
	For a matrix chaos $Y$ as in \prettyref{def:combinatorial_chaos} where $h$ has unit variance there is a constant $C_{\mathsf{MR}}$ (may depend on $r$) such that,
	\begin{align*}
		\E\left[\norm{Y}_2\right] \leq C_{\mathsf{MR}} \paren{\sigma\paren{\cA} +\alpha(h)^r\log \paren{d+m}^{\paren{r+3}/2}\nu\paren{\cA}}.
	\end{align*}
\end{theorem}
\prettyref{thm:chaos_thm} is stated as an expectation bound in \cite{BLNvH25}. For our applications, we also want a high-probability bound. A typical approach is  where one first proves a moment bound and applies Markov's inequality. We will use the following moment version of \prettyref{thm:chaos_thm} which is obtained by mimicking the proof of \prettyref{thm:chaos_thm}, specifically Appendix A.5 in \cite{BLNvH25} where we keep the order of moment explicit.
Towards this we define a parameter $\alpha_t(h)$ defined as,
\begin{align*}
    \alpha_t(h) = \norm{h}_{L^t} = \paren{\E\left[\abs{h}^t\right]}^{1/t}.
\end{align*}
\begin{theorem}\label{thm:moment-chaos-theorem}
	For an order-$r$ matrix chaos $Y$ (decoupled version of \prettyref{def:matrix-chaos-order-r}) given by,
	\begin{align*}
		Y = \sum_{i_1,\dots,i_r}h^{(1)}_{i_1}\dots h^{(r)}_{i_r} \cdot A_{i_1,\dots,i_r}
	\end{align*}
	where $h$ has unit variance there is a constant $C_{\mathsf{ISMR}}$(may depend on $r$) such that for every integer $t \geq 2$,
	\begin{align*}
		\paren{\E\left[\norm{Y}_2^t\right]}^{1/t} \leq C_{\mathsf{ISMR}}\paren{\sigma(\cA) + \alpha_{t+c\log(d+m)}(h)^r\paren{t+\log(d+m)}^{(r+3)/2}\cdot \nu(\cA)},
	\end{align*}
	where $h_{i_k}$ are i.i.d. copies of centered, unit-variance scalar variable $h$ and coefficient matrices $A_{i_1,\dots,i_r}\in \bbR^{d\times d}$.
\end{theorem}

\begin{proof}
	The proof mimics the argument in proof of \prettyref{thm:chaos_thm} (Appendix A.5 in \cite{BLNvH25}) but we keep the moment parameter ($t$ in our theorem statement above) as a free parameter instead of specializing to $\floor{c\log(d+m)}/2$ early as they do in their proof. 
    
    Instead we let $2s=t + \ceil{c\log (d+m)}$, where $c >0$ is a sufficiently large absolute constant (as in \cite{BLNvH25}). Next we recall the Schatten-$2s$ norm which is defined as,
    \begin{align*}
         \norm{Y}_{S_{2s}} \defeq  \paren{\Tr\paren{\abs{Y}^{2s}}}^{1/2s}, \quad \text{where} \quad \abs{Y} = \paren{Y^*Y}^{1/2}.
    \end{align*}
    Using the fact that the spectral norm $\norm{Y}_{\textrm{op}}$ is bounded by the Schatten-$2s$ we have,
	\begin{align*}
		\norm{Y}_{\textrm{op}} \leq \norm{Y}_{S_{2s}}.
	\end{align*}
    Because $t \leq 2s$ and expectation $L_k$ norms are monotone in $k$ for non-negative random variables,
	\begin{align*}
		\paren{\mathbb{E}\left[\norm{Y}_{S_{2s}}^t\right]}^{1/t} \leq \paren{\mathbb{E}[\norm{Y}_{S_{2s}}^{2s}]}^{1/2s} = \paren{\mathbb{E}[\mathsf{Tr}(\abs{Y}^{2s})]}^{1/2s}.
	\end{align*}
    Thus, for our proof it suffices to control the $2s$-trace moment as in \cite{BLNvH25}. The proof proceeds by induction on $r$, along the lines of Proof of Theorem 2.7 in \cite{BLNvH25}.
    
    \emph{Base Case} ($r=1$).
    Here $Y$ is a linear matrix sum with the form $X=\sum_{i\in [m]}h_iA_i$. the analysis in proof of Theorem A.7 in \cite{BLNvH25}, for a constant $C'_r$ yields a moment bound of the form,
	\begin{align*}\numberthis\label{eq:moment-bound-chaos-proof}
		\paren{\E\left[\norm{X}^t\right]}^{1/t} \leq C'_r\paren{\sigma_R(X) + \sigma_C(X) + \alpha_{t+c(\log(d+m))}(h)\paren{t+\log(d+m)}^2\nu(X)},
	\end{align*}
	where $\sigma_R(X) = \norm{\sum_i A_i^{\top}A_i}^{1/2}$ and $\sigma_C(X) = \norm{\sum_i A_iA_i^{\top}}^{1/2}$. We will now argue that the moment comparision machinery in the proof of Theorem A.7 works for any $2s=t + c\paren{\log(d+m)}$.

	Their proof in \cite{BLNvH25} uses Theorem 2.9 in \cite{BvH24} which compares the moments in above expressions with the Gaussian proxy $G=\sum_ig_iA_i$. Now since their proof works for any integer $t$, and we get the same expression along with the additive error terms denoted by $R_{2s}^{1/3}\sigma_{2s}^{2/3}s^{2/3}$ and $R_{2s}\cdot s^2$ terms. Next, they use Theorem 2.7 in  \cite{BBvH23} to bound the $2t$-th trace moment (again for an integer $t$) by a free model with an additive error which for us is $s^{3/4}\paren{\nu\sigma}^{1/2}$. Further the Lemma 2.5 in \cite{BBvH23} bounds the free model's operator norm by $\sigma_R+\sigma_C$. Now similarly one can bound the $R_{2s}(X)$ term by $\alpha_{2s}(h)r(X)$ and absorb the cross term $s^{3/4}\sqrt{(\sigma_R \lor \sigma_C)\nu}$ with Young's inequality by the same algebra. Putting it all together proves the base case (absorbing constants into $C_{\mathsf{ISMR}})$.
	
    \emph{Inductive step}. Condition on $h^{(1)}_{i_1},\dots,h^{(r-1)}_{i_{r-1}}$ and write as a linear chaos in $h^{(r)}_{i_r}$ as,
    \begin{align*}
        Y = \sum_{i_r \in [m]}h^{(r)}_{i_r} \cdot B_{i_r} \quad \text{where} \quad B_{i_r} \defeq \sum_{i_1,\dots,i_{r-1}}h^{(1)}_{i_1} \dots h^{(r-1)}_{i_{r-1}} \cdot A_{i_1,\dots,i_{r-1},i_r}.
    \end{align*}
    Now towards using induction, the iteration approach, Section 2.4.2 from \cite{BLNvH25} says that quantities in linear chaos (such as those in eqn.~\prettyref{eq:moment-bound-chaos-proof}) can be identified with the operator norms of \emph{intermediate flattenings} of the original coefficient tensor. In our setting, we need the following three intermediate random matrices,
    \begin{align*}
        Y_{(1)} &\defeq Y \Brac{1:\paren{r-1}\vert \set{r,r+1}\vert \set{r+2}},\\
        Y_{(2)} &\defeq Y \Brac{1:\paren{r-1}\vert \set{r+1}\vert \set{r,r+2}},\\
        Y_{(3)} &\defeq Y \Brac{1:\paren{r-1}\vert \set{r}\vert \set{r+1,r+2}},
    \end{align*}
    where the remaining chaos coordinates are $1,\dots,r-1$, the matrix row coordinate is indexed by the set after the first bar and the matrix column coordinate by the set after the second bar. Now applying eqn.~\prettyref{eq:moment-bound-chaos-proof} to the conditional linear chaos $Y=\sum_{i_r}h^{(r)}_{i_r}B_{i_r}$ for every $t \geq 2$,
    \begin{align*}
        \norm{Y}_{L_t^{h^{(r)}}} \leq C'_r \paren{\norm{Y_{(1)}} + \norm{Y_{(2)}} + \alpha_{t + c\log(d+m)}(h)\paren{t+ \log(d+m)}^2\norm{Y_{(3)}}},
    \end{align*}
    where $\norm{\cdot}_{L_t^{h_r}}$ means we take the $L_t$ norm over $h_{i_r}$, holding $h^{(1)}_1,\dots,h^{(r-1)}_{r-1}$ fixed. Now taking $L_t$ norm over $h_{h^{(1)}_{i_1}},\dots,h^{(r-1)}_{i_{r-1}}$ and repeatedly applying triangle inequality,
    \begin{align*}\numberthis\label{eq:intermediate-bound-needed}
        \norm{Y}_{L_t} \leq C'_r \paren{\norm{Y_{(1)}}_{L_t} + \norm{Y_{(2)}}_{L_t} + \alpha_{t + c\log(d+m)}(h)\paren{t+ \log(d+m)}^2\norm{Y_{(3)}}_{L_t}}.
    \end{align*}
    Both $Y_{(1)},Y_{(2)}$ are matrix chaoses of order $(r-1)$ involving $h^{(1)}_{i_1},\dots,h^{(r-1)}_{i_{r-1}}$, and with deterministic coefficient tensors given by intermediate flattenings of $\cA$, hence by induction hypothesis,
    \begin{align*}
        \norm{Y_{(k)}}_{L_t} \leq C'_{r-1}\paren{\sigma\paren{\cA_{(k)}} + \alpha_{t + c\log(d'+m)}(h)^{r-1} \paren{t + \log(d'+m)}^{(r+2)/2}\nu\paren{\cA_{(k)}}},
    \end{align*}
    where $\cA_{(k)}$ is the coefficient tensor of $Y_{(k)}$ and $d'$ is the dimension of intermediate flattenings (atmost $md$).
    By Lemma A.1 of \cite{BLNvH25}, the parameters are controlled by those for $\cA$ and we have,
    \begin{align*}
        \sigma\paren{\cA_{(1)}} \leq \sigma \paren{\cA}, \nu\paren{\cA_{(1)}} \leq \nu \paren{\cA}, \sigma\paren{\cA_{(2)}} \leq \sigma \paren{\cA}, \nu\paren{\cA_{(2)}} \leq \nu \paren{\cA},
    \end{align*}
    and $\log(d'+m) \lesssim_r \log(d+m)$, and $\alpha_{t + c\log(d'+m)}(h) \leq \alpha_{t+c'\log(d+m)}(h)$ for larger constant $c'$. Hence,
    \begin{align*}\numberthis\label{eq:intermediate-intermediate-bound}
        \norm{Y_{(1)}}_{L_t} + \norm{Y_{(2)}}_{L_t} \leq C'_r\paren{\sigma(\cA) + \alpha_{t+c\log(d+m)}(h)^{r-1}\paren{t + \log(d+m)}^{(r+2)/2}\nu\paren{\cA}}.
    \end{align*}
    Now the term $Y_{(3)}$ is also an order-$(r-1)$ chaos, but using the full induction hypothesis  will introduce a $\nu(\cdot)$ term which once multiplied by $\paren{t + \log(d+m)}^2$ factor would be too lossy. We instead proceed as in Appendix A.5 of \cite{BLNvH25}, and bound it by an iterated NCK inequality that depends only on $\sigma(\cdot)$. By a moment form of iterated NCK, similar to Appendix A.2 in \cite{BLNvH25},
    \begin{align*}\numberthis\label{eq:intermediate-bound-final}
        \norm{Y_{(3)}}_{L_t} \leq C'_r\cdot \alpha_{t + c\log(d+m)}(h)^{r-1}\paren{t+\log(d+m)}^{(r-1)/2} \nu\paren{\cA}.
    \end{align*}
    Finally by Lemma A.1 in \cite{BLNvH25}, we have $\sigma\paren{\cA_{(3)}}\leq \nu\paren{\cA}$ (since in $Y_{(3)}$, the original matrix coordinates $r+1,r+2$ are placed on same side and become $\nu$-flattenings of $\cA$).
    Now, putting it all together, plugging eqn.~\prettyref{eq:intermediate-intermediate-bound}, and eqn.~\prettyref{eq:intermediate-bound-final} in eqn.~\prettyref{eq:intermediate-bound-needed} one obtains,
    \begin{align*}
        &\qquad\qquad \norm{Y}_{L_t} \leq C'_r\paren{\sigma\paren{\cA} + \alpha_{t+c\log(d+m)}(h)^{r-1}\paren{t+\log(d+m)}^{(r+2)/2}\nu\paren{\cA}}\\
        &+C'_r\alpha_{t+c\log(d+m)}(h)\paren{t+\log(d+m)}^2\cdot \alpha_{t+c\log(d+m)}(h)^{r-1}\paren{t + \log(d+m)}^{(r-1)/2}\nu\paren{\cA}.
    \end{align*}
    The last term simplifies to $C'_r\alpha_{t+c\log(d+m)}(h)^r\paren{t+\log(d+m)}^{(r+3)/2}\nu\paren{\cA}$, and the preceding term $\alpha^{r-1}\paren{t+\log(d+m)}^{(r+2)/2}\nu\paren{\cA}$ is dominated by final term, since $\alpha_{t+c\log(d+m)}(h) \geq \alpha_2(h)$ (unit variance). Absorbing constants into $C_{\mathsf{ISMR}}$ completes the proof.
\end{proof}

\subsection{Low-Degree Polynomial Framework for Refutation Problems}
\label{sec:ldp-refutation-framework}
The low-degree polynomial framework for refutation problems was formalized in \cite{KVWX23}. We start by recalling some useful definitions and results from their framework.

\begin{definition}[Separation of Distribution and Property, Definition 2.11 in \cite{KVWX23}]
    Let $N=N_n$, and let $\sfQ_n$ be a probability distribution on $\bbR^N$, and let $\cR=\cR_n \subseteq \bbR^N$ be a property. A polynomial $f_n: \bbR^N \rightarrow \bbR$ is said to strongly separate $\sfQ_n$ from $\cR$ if,
    \begin{align*}
        f_n(X) \geq 1, \forall X \in \cR, \quad \text{and} \quad \mathop{\bbE}_{X\sim \sfQ_n}\Brac{f_n(X)}=0, \quad \mathop{\bbE}\limits_{X\sim \sfQ_n}\Brac{f_n(X)} = o_n(1).
    \end{align*}
\end{definition}

\begin{definition}[Refutation Algorithm]\label{def:refutation-algorithm}
    For an input $X \in \bbR^N$, and a property $\cR$, a refutation algorithm is a function $f:\bbR^N \rightarrow \bbR$ which outputs either $\mathsf{NO}$ or $\mathsf{MAYBE}$ such that,
    \begin{multicols}{2}
        \begin{itemize}
            \item if $X\in \cR$, outputs $\maybe$
            \item if $X \sim \sfQ$, outputs $\mathsf{NO}$ w.h.p.
        \end{itemize}
    \end{multicols}
\end{definition}

\begin{definition}[Threshold Refutation Algorithm]\label{def:threshold-refutation-algorithm}
    Let $f_n:\bbR^N \rightarrow \bbR$ be a polynomial. The threshold refutation algorithm associated with $f_n$ denoted $\sfA_f: \bbR^N \rightarrow \set{\maybe,\no}$ is defined as,
    \begin{align*}
        \sfA_f(X) = \begin{cases}
            \no &\quad f_n(X) <1\\
            \maybe &\quad f_n(X)\geq 1.
        \end{cases}
    \end{align*}
\end{definition}

\begin{proposition}[Proposition 2.13 in \cite{KVWX23}]
    Suppose $f_n$ strongly separates $\sfQ$ and $\cR$, then $\sfA_f$ as defined in \prettyref{def:threshold-refutation-algorithm} is a valid refutation algorithm.
\end{proposition}

\begin{proof}
    It is easy to see that $f(X) \geq 1$ for all $X \in \cR$ and it always outputs $\maybe$. The other guarantee follows using Markov's inequality as we notice it does not output $\no$ with probability,
    \begin{align*}
        \pr{X \sim \sfQ_n}{f_n(X)\geq 1} \leq \mathop{\bbE}_{X\sim \sfQ_n}\Brac{f_n^2(X)} = o_n(1).
    \end{align*}
\end{proof}

The strategy in \cite{KVWX23} to show computational hardness of refuting a property $\cR$ for a particular null distribution $\sfQ$ is to note that it suffices to construct a \emph{computationally quiet planted distribution} having that property $\cR$. Formally the task is to construct a planted distribution $\widetilde{\sfP}$ where,
\begin{enumerate}
    \item $\widetilde{\sfP}$ is supported on input instances with property $\cR$.
    \item It is computationally hard to distinguish $\widetilde{\sfP}$ from $\sfQ$.
\end{enumerate}

\begin{definition}[Advantage of a Test]
    Given two distributions $\sfP$ and $\sfQ$, a polynomial $f_n: \bbR^N \rightarrow \bbR$, and a degree parameter $D=D_n$ the advantage of a degree-$D$ test denoted $\Adv_{\leq D}$ is defined as,
    \begin{align*}
        \Adv_{\leq D}(\sfP,\sfQ) = \sup_{f:\deg(f) \leq D} \frac{\E_{\sfP}\Brac{f}}{\sqrt{\E_{\sfQ}\Brac{f^2}}}.
    \end{align*}
\end{definition}

\begin{proposition}[Proposition 2.14 in \cite{KVWX23}]
    Suppose for an infinite subsequence of values of $n$ there is a distribution $\sfP=\sfP_n$ that is supported on $\cR$ and for some $D=D_n$ we have $\Adv_{\leq D}(\sfP,\sfQ)=1+o(1)$, then no degree-$D$ polynomial separates $\sfQ$ from $\cR$.
\end{proposition}

Let $\cE_{\ell} = \binom{V}{\ell}$ and $N=\abs{\cE_{\ell}} = \binom{n}{\ell}$. We encode the hypergraph by $X=\paren{X_e}_{e\in \cE_{\ell}}\in \set{0,1}^N$, where $X_e=1$ indicates that the candidate hyperedge $e$ is present. Let $\cR(k)$ be the property that there is an independent set of size at least $k$,
\begin{align*}
    \cR(k) \defeq \set{X \in \set{0,1}^N: \mathsf{MIS}(X) \geq k}.
\end{align*}
Fix a parameter $k$. Towards our quiet planted distribution $\widetilde{\sfP}$, we need a distribution that is supported on $\cR(k)$ i.e. has independent set of size $k$ and at the same time is low-degree indistinguishable from the null distribution $\sfQ=\cH_0\paren{n,\ell,p}$ (as defined in \prettyref{def:random-hypergraph}) upto degree $D$. We start by normalizing the input,
\begin{align*}
    Y_e \defeq \frac{X_e-p}{\sqrt{p(1-p)}}.
\end{align*}
Under the null distribution $\sfQ=\cH_0\paren{n,\ell,p}$, the family $\set{Y_e}_{e\in \cE_{\ell}}$ is independent with $\bbE_{\sfQ}\Brac{Y_e}=0$ and $\bbE_{\sfQ}\Brac{Y_e^2}=1$. For $\alpha \subseteq \cE_{\ell}$ define the monomial $h_{\alpha}=\prod_{e\in \alpha}Y_e$. Then, note that the collection $\set{h_{\alpha}}_{\alpha \subseteq \cE_{\ell}}$ is an orthonormal basis for $L^2\paren{\sfQ}$. One can verify that indeed for $\alpha,\beta \subseteq \cE_{\ell}$,
\begin{align*}
    \bbE_{\sfQ}\Brac{h_{\alpha}\cdot h_{\beta}} = \bbE_{\sfQ}\Brac{\prod_{e\in \alpha}Y_e\cdot \prod_{e\in \beta}Y_e}.
\end{align*}
If $\alpha \neq \beta$, then some $e \in \alpha \triangle \beta$ appears once and by independence $\bbE_{\sfQ}\Brac{Y_e}=0$. If $\alpha = \beta$ we have,
\begin{align*}
    \bbE_{\sfQ}\Brac{h^2_{\alpha}} = \prod_{e\in \alpha}\bbE_{\sfQ}\Brac{Y_e^2}=1.
\end{align*}
Hence, $\bbE_{\sfQ}\Brac{h_{\alpha}h_{\beta}}=\one_{\alpha=\beta}$ and hence they form an orthonormal basis over $L^2\paren{\sfQ}$.
\begin{proposition}[Proposition 8.3 in \cite{Wei25}]\label{prop:advantage-computation}
    Given an orthonormal basis w.r.t the null distribution $\sfQ$ and a planted distribution ${\sfP}$ on hypergraphs, for every $D \geq 0$ we have,
    \begin{align*}
        \Adv_{\leq D}\paren{\sfP,\sfQ}^2 = 1 + \sum_{1 \leq \abs{\alpha} \leq D}\paren{\bbE_{\sfP}\Brac{h_{\alpha}}}^2.
    \end{align*}
\end{proposition}

\begin{remark}
    The quantity $\Adv_{\leq D}$ is also denoted low-degree likelihood ratio $\norm{L^{\leq D}}$ is the projection of likelihood ratio onto subspace of degree-$D$ polynomials \cite{Hop18}; where the likelihood ratio is  $L=d\sfP/d\sfQ$ having $\norm{L}=\sqrt{\bbE_{\sfQ}\Brac{L(Y)^2}}$ and it holds that $\norm{L^{\leq D}-1}^2 = \Adv^2_{\leq D}-1$.
\end{remark}

Note that the monomial $h_{\alpha}=\prod\limits_{e\in \alpha}Y_e$ has degree $\abs{\alpha}$ in the hyperedge variables. Hence, a degree-$D$ projection of the likelihood ratio only contains terms with $\abs{\alpha} \leq D$.

\section{Sum-of-Squares Certificates in Semirandom Hypergraphs}
\label{sec:sos-main-results}
In this section, we show low-degree Sum-of-Squares (SoS) proofs certifying useful bounds on the size of an independence number of $\ell$-uniform hypergraphs.

We start by proving upper bounds on independence number, the size of Maximum Independent Set (MIS) for a random hypergraph model $\cH_0\paren{n,\ell,p}$ (see \prettyref{def:random-hypergraph}) and later argue that the certificates are inherently robust, and naturally extend to semirandom hypergraphs, see \prettyref{def:semirandom-hypergraph}. 

Concretely, given an $\ell$-uniform hypergraph $H=(V,E)$ on $n$ vertices we introduce decision variables $\set{x_v: v \in V}$ (also called indeterminates) to indicate if $v$ belongs to an independent set. We formulate the Maximum Independent Set (MIS) problem as the polynomial system using the constraints (axioms of the system) as below.

\newcommand{\hypermissdp}{\textup{\hyperref[sdp:hypergraph_mis_poly]{$\mathcal{P}_{\textsf{MIS}}$}}}
\begin{namedSDP}{$\mathcal{P}_{\textsf{MIS}}$ (Polynomial System for $\ell$-Uniform Hypergraph MIS)}
\label{sdp:hypergraph_mis_poly}
\begin{align*}
    \max \,&\sum_{v\in V}\,{x_v}\\
    \qquad\text{subject to } \qquad\qquad\qquad &\\
    \prod_{v\in e}x_v&=0 && \forall e \in E\numberthis \label{eq:edge_constraint_mis}\\
    x_{v}^2 &= x_v && \forall v \in  V \numberthis \label{eq:indicator_constraint_mis}
    \end{align*}
\end{namedSDP}

We consider a degree-$2\ell$ SoS relaxation of \hypermissdp, and obtain a pseudo-expectation operator $\pE:\bbR[x]_{\leq 2\ell} \rightarrow \bbR$ satisfying all degree-$2\ell$ consequences of the axioms of \hypermissdp. Since this is a relaxation, the true independence number $\alpha(H)$ can only be smaller and it always produces correct certificates,
\begin{align*}
    k \defeq \sum_{v\in V}\pE\Brac{x_v}.
\end{align*}

Our goal will be to formally prove \prettyref{thm:main-theorem-informal}, first by showing that with high probability over a random hypergraph, $H \sim \cH_0(n,\ell,p)$ the result holds.

\begin{theorem}\label{thm:main-theorem-formal}
    Given an instance of a random hypergraph $H=(V,E)$ generated from the random hypergraph model $\cH_0\paren{n,\ell,p}$, and a $2\ell$-degree pseudo-expectation $\pE$ for the Sum-of-Squares relaxation of \hypermissdp, that there exists a constant $c_{\ell}$ (depending only on $\ell$) such that with high probability (over the randomness of the input), it follows that for,
    \begin{itemize}
        \item \textbf{Even Arity Case.} For the case of $\ell=2q$ we have,
            \begin{align*}
                \sum_{v\in V}\pE[x_v] \leq c_{\ell} \cdot \paren{\frac{\sqrt{n}}{p^{1/\ell}} + \frac{(\log n)^{1/\ell}}{p^{2/\ell}}}.
        \end{align*}

        \item \textbf{Odd Arity Case.} For the case of $\ell=2q+1$ we have,
            \begin{itemize}
                \item \textbf{Regime I}: If $p \geq p^{\star} \defeq \paren{\log n}^2/\sqrt{n}$ for $\ell=2q+1$ we have that,
                    \begin{align*}
                        \sum_{v\in V}\pE[x_v] \leq c_{\ell} \cdot \paren{\frac{\sqrt{n}}{p^{1/\ell}} + \frac{n^{q/(2q+1)}\cdot (\log n)^{5/2\ell}}{p^{2/\ell}}}.
                    \end{align*}

                \item \textbf{Regime II}: Where $p \leq p^{\star}=\paren{\log n}^2/\sqrt{n}$ we have that,
                    \begin{align*}
                        \sum_{v\in V}\pE[x_v] \leq c_{\ell} \cdot \paren{\frac{\sqrt{n}\paren{\log n}^{1/2\ell}}{p^{1/\ell}} + \frac{(\log n)^{3/\ell}}{p^{2/\ell}}}.
                    \end{align*}
            \end{itemize}
    \end{itemize}
\end{theorem}
We will prove this theorem separately for even and odd values of $\ell$, a common theme in hypergraph and Constraint satisfaction problems. In \prettyref{sec:sos-even-analysis} we prove the even arity case and \prettyref{sec:sos-odd-analysis} we prove the odd arity case. 

Finally to prove the result in the semirandom model $\cH_1\paren{n,\ell,p}$ we simply note that the monotone adversary can only add edges in the hypergraph instance. 
\begin{lemma}\label{lem:monotone-adversary}
    Let $H=(V,E)$ be a hypergraph instance and $H'=(V,E')$ be another hypergraph instance where $E \subseteq E'$. Let $\alpha_{\text{SoS}}(H)$ be the degree-$2\ell$ SoS optimum of $\hypermissdp(H)$. Then,
    \begin{align*}
        \alpha_{\text{SoS}}(H') \leq \alpha_{\text{SoS}}(H).
    \end{align*}
\end{lemma}
\begin{proof}
    Feasible pseudo-expectations for $H'$ must satisfy all constraints for the hypergraph $H$, and additionally more constraints from $H'$. Therefore, we obtain above the desired conclusion as the feasible set of solution can only shrink.
\end{proof}
Since a hypergraph $H'\sim \cH_1\paren{n,\ell,p}$ is obtained by adding edges to a random $H \sim \cH_0\paren{n,\ell,p}$, using \prettyref{lem:monotone-adversary} it follows that the bound proven for a random $H$ continues to hold. Therefore the proof of \prettyref{thm:main-theorem-formal} also holds in semirandom hypergraph model $\cH_1\paren{n,\ell,p}$.

\section{Independence Number Certificates in Even Arity Hypergraphs}
\label{sec:sos-even-analysis}
We state the main result we show in this section, for $2q$-uniform random hypergraphs.

\begin{theorem}\label{thm:even-case-bounds}
	Given an instance of a random hypergraph $H=(V,E)$ generated from $\cH_0\paren{n,2q,p}$, and a degree-$4q$ pseudo-expectation $\pE$ for the SoS relaxation of \hypermissdp, it follows that with high probability (over the randomness of the input) for a constant $c_{2q}$ we have,
	\begin{align*}
		\sum_{v\in V}\pE\Brac{x_v} \leq c_{2q} \cdot \paren{\frac{\sqrt{n}}{p^{1/2q}} + \frac{\log^{1/2q}(n)}{p^{1/q}}}.
	\end{align*}
\end{theorem}

\subsection{Sum-of-Squares Analysis for Random Hypergraphs}
We let $X\defeq \sum_{v\in V}x_v$. The proof proceeds by considering $X^{2q}$ and expanding it to a sum over ordered $2q$-tuples, followed by splitting into the \say{lower-order terms} (tuples with less than $2q$ distinct vertices) and the \say{all distinct term}. Then we proceed to prove SoS certifiable bounds of degree at most $2q$ on each of these, and use weak duality of SoS proofs to obtain a similar bounds for a degree-$4q$ pseudo-expectation moments. Finally, we  use well-known pseudo-expectation version of the usual inequalities from analysis to obtain the final bound.

We start with a multinomial expansion of $X^{2q}$ as sum over the ordered $2q$-tuples,
\begin{align*}
	X^{2q} = \paren{\sum_{v\in V}x_v}^{2q} = \sum_{\paren{v_1,v_2,\dots,v_{2q}}\in V^{2q}}x_{v_1}x_{v_2}\dots x_{v_{2q}}
\end{align*}
We partition these ordered tuples according to the number of distinct vertices in $\paren{v_1,v_2,\dots,v_{2q}}$ and we have the formal identity,
\begin{align} \label{eq:poly_expand}
	X^{2q} = \underbrace{\sum_{\substack{\paren{v_1,v_2,\dots,v_{2q}}\in V^{2q}\\\abs{\set{v_1,v_2,\dots,v_{2q}}}=2q}}x_{v_1}x_{v_2}\dots x_{v_{2q}}}_{\text{all-distinct term}} + \underbrace{\sum_{\substack{\paren{v_1,v_2,\dots,v_{2q}}\in V^{2q}\\\abs{\set{v_1,v_2,\dots,v_{2q}}}<2q}}x_{v_1}x_{v_2}\dots x_{v_{2q}}}_{\text{lower-order terms}}
\end{align}
Note that the above is just an exact identity in the polynomial ring and hence SoS at any-degree immediately proves it. Now, at first we will bound the lower-order terms in eqn.~\prettyref{eq:poly_expand}.

\begin{proposition}\label{prop:lower_term_bounds}
	Let $H=(V,E)$ be a $f$-uniform hypergraph, and consider the polynomial system \hypermissdp, then there exists a constant $C_f$ such that,
	\begin{align*}
		\hypermissdp \SoSp{x}{2f} \set{\sum_{\substack{\paren{v_1,v_2,\dots,v_{f}}\in V^{f}\\\abs{\set{v_1,v_2,\dots,v_{f}}}<f}}x_{v_1}x_{v_2}\dots x_{v_{f}} \leq C_f\paren{\sum_{v\in V}x_v}^{f-1}}.
	\end{align*}
\end{proposition}
\begin{proof}
	We start by simplifying the lower-order terms. For each $1 \leq d \leq f-1$ define,
	\begin{align*}
		T_d=\sum_{\substack{\paren{v_1,v_2,\dots,v_{f}}\in V^{f}\\\abs{\set{v_1,v_2,\dots,v_{f}}}=d}}x_{v_1}x_{v_2}\dots x_{v_{f}}
	\end{align*}
	and our new goal is to establish that $T_d \leq C_{f,d} \cdot X^{f-1}$ for some constant $C_{f,d}>0$. Towards this we consider $X^d$ and expanding again into all-distinct terms and terms with repetition,
	\begin{align*}
		X^d = \paren{\sum_{v\in V}x_v}^d= \sum_{\substack{\paren{v_1,v_2,\dots,v_{d}}\in V^{d}\\\abs{\set{v_1,v_2,\dots,v_{d}}}=d}}x_{v_1}x_{v_2}\dots x_{v_{d}} + \sum_{\substack{\paren{v_1,v_2,\dots,v_{d}}\in V^{d}\\\abs{\set{v_1,v_2,\dots,v_{d}}}<d}}x_{v_1}x_{v_2}\dots x_{v_{d}}
	\end{align*}
	Rearranging the terms above and using the constraint (\ref{eq:indicator_constraint_mis}) we obtain that,
	\begin{align*}
		X^d - \sum_{\substack{\paren{v_1,v_2,\dots,v_{d}}\in V^{d}\\\abs{\set{v_1,v_2,\dots,v_{d}}}=d}}x_{v_1}x_{v_2}\dots x_{v_{d}} &= \sum_{\substack{\paren{v_1,v_2,\dots,v_{d}}\in V^{d}\\\abs{\set{v_1,v_2,\dots,v_{d}}}<d}}x_{v_1}x_{v_2}\dots x_{v_{d}}\\
		&= \sum_{\substack{\paren{v_1,v_2,\dots,v_{d}}\in V^{d}\\\abs{\set{v_1,v_2,\dots,v_{d}}}<d}}\paren{x_{v_1}x_{v_2}\dots x_{v_{d}}}^2
	\end{align*}
	and since the polynomial ${x_{v_1}x_{v_2}\dots x_{v_d}}$ has degree $d$ it follows that,
	\begin{align*}
		\hypermissdp \SoSp{x}{2d} \set{X^d \geq \sum_{\substack{\paren{v_1,v_2,\dots,v_{d}}\in V^{d}\\\abs{\set{v_1,v_2,\dots,v_{d}}}=d}}x_{v_1}x_{v_2}\dots x_{v_{d}}}
	\end{align*}
	Now note that if we count unordered $d$-tuples we have the equality,
	\begin{align*}
		\sum_{\substack{\paren{v_1,v_2,\dots,v_{d}}\in V^{d}\\\abs{\set{v_1,v_2,\dots,v_{d}}}=d}}x_{v_1}x_{v_2}\dots x_{v_{d}} = d! \sum_{\substack{W \subset V\\\abs{W}=d}}\prod_{v\in W}x_v        
	\end{align*}
	Since it is an equality, if holds in SoS of any degree and it follows that,
	\begin{align}\label{eq:bound_t_d_term_first_half}
		\hypermissdp \SoSp{x}{2d} \set{X^d \geq d! \sum_{\substack{W \subset V\\\abs{W}=d}}\prod_{v\in W}x_v}
	\end{align}
	On the other hand, for $T_d$ we count ordered $f$-tuples with exactly $d$ distinct entries. It is well-known that for a fixed unordered set $W \subset V$ of size $\abs{W}=d$, the number of ordered $f$-tuples whose set of distinct entries equal $W$ is classically $d!\St(f,d)$ where $\St$ denotes the Stirling number of second kind. Hence the following polynomial identity holds for $T_d$,
	\begin{align*}
		\sum_{\substack{\paren{v_1,v_2,\dots,v_{f}}\in V^{f}\\\abs{\set{v_1,v_2,\dots,v_{f}}}=d}}x_{v_1}x_{v_2}\dots x_{v_{f}} = d!\,\St(f,d) \sum_{\substack{W\subset V\\\abs{W}=d}}\prod_{v\in W}x_v.
	\end{align*}
	Since it is an equality SoS also proves it, and using eqn.~\prettyref{eq:bound_t_d_term_first_half} it follows that for any $d\in [f-1]$,
	\begin{align*}
		\hypermissdp \SoSp{x}{2f} \set{T_d \leq \St(f,d) \cdot X^d}
	\end{align*}
	\begin{claim}\label{claim:monotonicity_bounds}
		For each integer $1 \leq d \leq f-1$, for the polynomial system \hypermissdp, it holds that,
		\begin{align*}
			\hypermissdp \SoSp{x}{2f} \set{{X}^d \leq {X}^{f-1}}.
		\end{align*}
	\end{claim} 
	\begin{proof}
		The proof starts by noting that for any $t\in [d,f-2]$ we can write a telescoping sum,
		\begin{align*}
			X^{f-1}-X^d  = \sum_{t=d}^{f-2}\paren{X^{t+1}-X^t}
		\end{align*}
		By the addition inference rule of SoS proofs (see \prettyref{eq:additiona-inference-sos}) it suffices to show that for every integer value of $t\in [d,f-2]$ the following holds,
		\begin{align*}
			\hypermissdp \SoSp{x}{2f} \set{X^{t+1} - X^t \geq 0}
		\end{align*}
		We focus on a fixed $t \in [d,f-2]$ and expand the polynomial $X^{t+1}-X^t$ and we have,
		\begin{align*}
			X^{t+1}-X^t = \sum_{\paren{w_1,w_2,\dots,w_{t+1}}\in V^{t+1}}x_{w_1}x_{w_2}\dots x_{w_{t+1}} - \sum_{\paren{w_1,w_2,\dots,w_t}\in V^t} x_{w_1}x_{w_2}\dots x_{w_t}
		\end{align*}
		We let an ordered $t$-tuple $\paren{w_1,w_2,\dots,w_t}$ denoted $W$ and let $x_W=\prod_{w\in W}x_w$ and we have,
		\begin{align*}
			X^{t+1}-X^t &= \sum_{\paren{w_1,w_2,\dots,w_t}\in V^t} \sum_{v \in V}x_{w_1}x_{w_2}\dots x_{w_t}x_v - \sum_{\paren{w_1,w_2,\dots,w_t}\in V^t}x_{w_1}x_{w_2}\dots x_{w_t}\\
			&=\sum_{\paren{w_1,w_2,\dots,w_t}\in V^t}x_{w_1}x_{w_2}\dots x_{w_t}\paren{\sum_{v\in V}x_v-1}= \sum_{W\in V^t}x_W \paren{X-1}
		\end{align*}
		It now remains to show that for each fixed ordered $t$-tuple $W$ the polynomial $x_W\paren{X-1}$ is SoS under \hypermissdp. Now in the expansion of $x_W\paren{X-1}$, first consider the terms where we have $v \in \set{w_1,w_2,\dots,w_t}$ and here we have $x_Wx_v=x_W$ using $x_v^2=x_v$, and we have exactly $t$ such $v$'s. Now if $v\notin \set{w_1,w_2,\dots,w_t}$ it is a distinct monomial of degree $t+1$ and hence,
		\begin{align*}
			x_W\paren{\sum_{v\in V}x_v} = t \cdot x_W  + \sum_{\substack{v\in V\\v\notin \set{w_1,w_2,\dots,w_t}}}x_W \cdot x_v
		\end{align*}
		Now using this we have that the term $x_W\paren{X-1}$ can be written as,
		\begin{align*}
			x_W\paren{X-1} &= \paren{t-1}x_W + \sum_{\substack{v\in V\\v\notin \set{w_1,w_2,\dots,w_t}}}x_Wx_v\\
			&= \paren{t-1}\paren{\prod_{w\in W}x_w^2} + \sum_{\substack{v\in V\\v\notin \set{w_1,w_2,\dots,w_t}}}\paren{x_v\prod_{w\in W}x_w}^2
		\end{align*}
		and we are done noting that $t\geq 1$ and the degree of polynomial is at most $t+1\leq f-1$.
	\end{proof}
	Using \prettyref{claim:monotonicity_bounds} and the addition inference rule of SoS proofs it follows,
	\begin{align*}
		\hypermissdp \SoSp{x}{2f} \set{\sum_{d=1}^{f-1}T_d \leq \paren{\sum_{d=1}^{f-1}\St(f,d)}X^{f-1}}.
	\end{align*}
	Now letting $C_f\defeq \sum_{d<f}\St(f,d)$ we obtain the desired result.
\end{proof}

Next we show how to bound the higher-order all-distinct term in eqn.~\prettyref{eq:poly_expand}. Towards this, we make a definition which matricizes a $2q$-uniform hypergraph.

\begin{definition}[$p$-Biased Hypergraph Flattening Matrix]\label{def:signed_flattening_matrix}
	Given a $2q$-uniform hypergraph $H=(V,E)$, we define the $p$-biased Fourier hypergraph flattening matrix (using the $p$-biased Fourier character $h_e$ defined in \prettyref{sec:p-biased-character}) denoted $B^{(2q)}:\abs{\cI_q} \times \abs{\cI_q} \rightarrow\mathbb{R}$ where $\cI_q \defeq \binom{V}{q}$ as,
	\begin{align*}
		B_{I,J}^{(2q)}=
		\begin{cases}
			h_{I \cup J} \quad &\text{if }I,J \in \cI, \abs{I \cup J}=2q\\
			0 \quad &\text{otherwise}
		\end{cases}
	\end{align*}
\end{definition}
We next show useful bounds on the spectral norm of matrix $B^{(2q)}$ for a random hypergraph.

\begin{proposition}\label{prop:flattening-random-matrix-bounds}
	For a random $2q$-uniform hypergraph $H=(V,E)$ where $H\sim \cH_0\paren{n,2q,p}$ it holds with high probability (over the randomness of input) for a constant $\kappa_{2q}$ such that,
	\begin{align*}
		\norm{B^{(2q)}}_2 \leq \kappa_{2q} \cdot \paren{n^{q/2} + \sqrt{\frac{\log n}{p}}}.
	\end{align*}
\end{proposition}
\noindent
We refer to \prettyref{sec:deviation_matrix_bounds} for a formal proof of \prettyref{prop:flattening-random-matrix-bounds}.

\begin{lemma}\label{lem:even_higher_order_term_bound}
	For a $2q$-uniform random hypergraph $H=(V,E)$ where $H \sim \cH_0(n,2q,p)$, and the polynomial system \hypermissdp, it follows that with high probability (over the randomness of the input),
	\begin{align*}
		\hypermissdp \SoSp{x}{2q} \set{\sum_{\substack{\paren{v_1,v_2,\dots,v_{2q}}\in V^{2q}\\\abs{\set{v_1,v_2,\dots,v_{2q}}}=2q}}x_{v_1}x_{v_2}\dots x_{v_{2q}} \leq \gamma_{2q}\cdot \paren{\frac{n^{q/2}}{\sqrt{p}} + \frac{\sqrt{\log n}}{p}} \paren{\sum_{v\in V}x_v}^q},
	\end{align*}
	for a constant $\gamma_{2q}\defeq \kappa_{2q}\cdot q!$ where $\kappa_{2q}$ as defined in \prettyref{prop:flattening-random-matrix-bounds}.
\end{lemma}
\begin{proof}
	Recall that $X=\sum_{v\in V}x_v$, and for each ordered $q$-tuple we denote $I=\paren{i_1,i_2,\dots,i_q}$, and we denote $x_I = \prod_{i\in I}x_i$ and  we have that,
	\begin{align*}
		\sum_{\substack{\paren{i_1,i_2,\dots,i_{2q}}\in V^{2q}\\\abs{\set{i_1,i_2,\dots,i_{2q}}}=2q}}x_{i_1}x_{i_2}\dots x_{i_{2q}} = \paren{q!}^2 \cdot \sum_{\substack{I,J\in \cI\\\abs{I\cup J}=2q}}x_Ix_J
	\end{align*}
	Now using the constraint (\ref{eq:edge_constraint_mis}) of \hypermissdp, if $e=I \cup J \in E$ we have $x_Ix_J=0$, hence we can write the higher-order term (weighted with matrix $B^{(2q)}$) as simply a sum over the non-edges,
	\begin{align*}
		\sum_{\substack{I,J\in \cI\\\abs{I\cup J}=2q}}B^{(2q)}_{I,J}\cdot x_Ix_J &= \sum_{\substack{I,J\in \cI\\\abs{I\cup J}=2q\\I\cup J \notin E}}B^{(2q)}_{I,J} \cdot x_Ix_J + \sum_{\substack{I,J\in \cI\\\abs{I\cup J}=2q\\I\cup J \in E}}B^{(2q)}_{I,J} \cdot x_Ix_J = \sum_{\substack{I,J\in \cI\\\abs{I\cup J}=2q\\I\cup J \notin E}}B^{(2q)}_{I,J} \cdot x_Ix_J\\
	\end{align*}
	Again note that the constraint (\ref{eq:edge_constraint_mis}) is an equality constraint and hence SoS at any degree proves it. We define the vector ${z}\in \bbR^{\abs{\cI}}$ as $z(I)=x_I,\forall I \in \cI$ and observe that the quadratic form,
	\begin{align*}
		{z}^{\top}B^{(2q)}{z} &= \sum_{I,J \in \cI}B^{(2q)}_{I,J} \cdot x_Ix_J= \sum_{\substack{I,J\in \cI\\\abs{I\cup J}=2q\\I\cup J \notin E}}B^{(2q)}_{I,J} \cdot x_Ix_J = \Bigg(\sqrt{\frac{p}{1-p}}\Bigg)\sum_{\substack{I,J\in \cI\\\abs{I\cup J}=2q\\I\cup J \notin E}}x_Ix_J 
	\end{align*}
	where again we used the constraint (\ref{eq:edge_constraint_mis}) and hence we have the SoS provable polynomial identity,
	\begin{align*}
		\sum_{\substack{I,J\in \cI\\\abs{I\cup J}=2q}}x_Ix_J &= \Bigg(\sqrt{\frac{1-p}{p}}\,\Bigg)\cdot {z}^{\top}B^{(2q)}{z}\\
		&\leq \Biggl(\sqrt{\frac{1-p}{p}}\,\Biggr)\norm{B^{(2q)}}_2\norm{{z}}_2^2 && \paren{\text{Using \prettyref{fact:operator_sos_bound}}}\\
		&\leq \kappa_{2q}\cdot \paren{\frac{n^{q/2}}{\sqrt{p}} + \frac{\sqrt{\log n}}{p}} \norm{{z}}_2^2 && \paren{\text{Using \prettyref{prop:flattening-random-matrix-bounds}}}.
	\end{align*}
	Putting everything together it follows that with high probability (over the randomness of input),
	\begin{align*}
		\hypermissdp \SoSp{x}{2q} \set{\sum_{\substack{\paren{v_1,v_2,\dots,v_{2q}}\in V^{2q}\\\abs{\set{v_1,v_2,\dots,v_{2q}}}=2q}}x_{v_1}x_{v_2}\dots x_{v_{2q}} \leq \paren{q!}^2 \cdot \kappa_{2q} \paren{\frac{n^{q/2}}{\sqrt{p}} + \frac{\sqrt{\log n}}{p}} \paren{\sum_{I\in \cI}x_I^2}}
	\end{align*}
	Now finally we use the SoS degree-$2q$ provable inequality in eqn.~\prettyref{eq:bound_t_d_term_first_half} and we have that,
	\begin{align*}
		\sum_{I \in \cI}x_I \leq \frac{1}{q!}\sum_{\paren{v_1,\dots,v_q} \in V^q}x_{v_1}\dots x_{v_q} = \frac{X^q}{q!}.
	\end{align*}
	Now using the axioms of \hypermissdp, it follows that $x_I=x_I^2$ and hence it follows that,
	\begin{align*}
		\hypermissdp \SoSp{x}{2q} \set{\sum_{\substack{\paren{v_1,v_2,\dots,v_{2q}}\in V^{2q}\\\abs{\set{v_1,v_2,\dots,v_{2q}}}=2q}}x_{v_1}x_{v_2}\dots x_{v_{2q}} \leq \paren{{\kappa_{2q}}\cdot {q!}} \paren{\frac{n^{q/2}}{\sqrt{p}} + \frac{\sqrt{\log n}}{p}} \paren{\sum_{v\in V}x_v}^q}.
	\end{align*}
\end{proof}

\noindent
Now we have all the ingredients to complete the proof of \prettyref{thm:even-case-bounds}.

\begin{proof}[Proof of \prettyref{thm:even-case-bounds}]
	Using the bounds from \prettyref{prop:lower_term_bounds} by setting $f=2q$ and using \prettyref{lem:even_higher_order_term_bound} we have that,
	\begin{align*}
		\hypermissdp \SoSp{x}{4q} \set{X^{2q} \leq \gamma_{2q}\cdot  X^q + C_q \cdot X^{2q-1}}
	\end{align*}
	By using \prettyref{cor:sos_soundness} it follows that for a degree-$4q$ pseudo-expectation $\pE$ it holds that,
	\begin{align}\label{eq:pe_solve_eqn}
		\pE[X^{2q}] \leq \gamma_{2q} \cdot \paren{\frac{n^{q/2}}{\sqrt{p}} + \sqrt{\frac{\log n}{p^2}}} \pE[X^q] + C_q \cdot \pE\left[X^{2q-1}\right].
	\end{align}
	Now let $M=\pE\left[X^{2q}\right]$ and let $u=M^{1/2q}$ so that $\pE[X^{2q}]=u^{2q}$.
	Now we use \prettyref{fact:pe_holder} by letting $f(x)=\sum_{v\in V}x_v,g(x)=1,t=2q$ (note the degree is $dt=2q$ as $f$ has degree $1$) and we obtain,
	\begin{align*}
		\pE[X^{2q-1}] \leq \paren{\pE[X^{2q}]}^{(2q-1)/2q} = u^{2q-1}
	\end{align*}
	Similarly we let $f(x)=1,g(x)=\paren{\sum_{v\in V}x_v}^q,t=2$ in \prettyref{fact:pe_holder} (note that the degree is $dt=2q$ as $g$ has degree $q$) and hence we have that,
	\begin{align*}
		\pE[X^q] \leq \paren{\pE\left[X^{2q}\right]}^{1/2} = u^q
	\end{align*}
	Now using the above bounds in eqn.~\prettyref{eq:pe_solve_eqn}, we obtain that,
	\begin{align*}
		u^{2q} \leq \gamma_{2q} \cdot \paren{\frac{n^{q/2}}{\sqrt{p}} + \sqrt{\frac{\log n}{p^2}}} \cdot u^q + C_q \cdot u^{2q-1}
	\end{align*}
	Now we consider two cases, where either $u \leq 2C_q$ and otherwise for $u \geq 2C_q$,
	\begin{align*}
		\frac{u^{2q}}{2} \leq \gamma_{2q} \cdot \paren{\frac{n^{q/2}}{\sqrt{p}} + \sqrt{\frac{\log n}{p^2}}} \cdot u^q \implies u \leq \paren{2\gamma_{2q} \cdot \paren{\frac{n^{q/2}}{\sqrt{p}} + \sqrt{\frac{\log n}{p^2}}}}^{1/q}
	\end{align*}
	Therefore using above, it follows that in either case we have,
	\begin{align*}
		\pE[X] \leq \max \set{2C_q, \paren{2\gamma_{2q}}^{1/q}\cdot \Biggl(\frac{\sqrt{n}}{p^{1/2q}} + \frac{\log^{1/2q}(n)}{p^{1/q}}\Biggr)} = c_{2q} \cdot \Biggl(\frac{\sqrt{n}}{p^{1/2q}} + \frac{\log^{1/2q}(n)}{p^{1/q}}\Biggr)
	\end{align*}
	for a large enough $n$ and by letting $c_{2q} = \paren{2\cdot C_q + \paren{2\cdot \gamma_{2q}}^{1/q}}$.
\end{proof}

\subsection{Bounding Spectral Norm for \texorpdfstring{$p$}{p}-Biased Hypergraph Flattened Matrix}
\label{sec:deviation_matrix_bounds}
For a $2q$-uniform hypergraph $H=(V,E)$, we recall the set $\cI=\set{I\subset V: \abs{I}=q}$ for the set of $q$-subsets of vertices, and we define $N=\abs{\cI}=\binom{n}{q}$ . We recall the $p$-biased flattening matrix from \prettyref{def:signed_flattening_matrix}, where for $I,J \in\cI$ the matrix is given as below, 
\begin{align*}
	B^{(2q)}_{I,J} = \begin{cases}
		+\sqrt{\frac{p}{1-p}} \quad\text{if $\abs{I \cup J}=2q$ and }I \cup J\notin E \vspace{1ex}\\
		-\sqrt{\frac{1-p}{p}}\quad\text{if $\abs{I \cup J}=2q$ and }I \cup J\in E \vspace{1ex}\\        \quad\,\,\,\,\,0\quad\quad\text{otherwise}
	\end{cases}
\end{align*}
It is easy to verify that $B^{(2q)}$ is centered as $\E\left[B^{(2q)}_{I,J}\right]=0$. One can also check that if $\abs{I\cup J}<2q$,then the corresponding matrix entry has $0$ variance and otherwise we can compute as,
\begin{align*}
	\Var\paren{B^{(2q)}_{I,J}} = \E\left[\paren{B^{(2q)}_{I,J}}^2\right] = \paren{1-p}\paren{\frac{p}{1-p}} + p \paren{\frac{1-p}{p}} = 1.
\end{align*}
We also note that the entries of the matrix $B^{(2q)}$ are bounded as,
\begin{align*}
	\abs{B^{(2q)}_{I,J}} \leq \max\set{\sqrt{\frac{p}{1-p}},\sqrt{\frac{1-p}{p}}} \leq \frac{1}{\sqrt{p(1-p)}}.
\end{align*}
We recall results in random matrix theory \cite{BvH16} that gives sharp non-asymptotic bounds for such matrices.
\begin{theorem}[Corollary 3.12, Remark 3.13 in \cite{BvH16}]\label{thm:random-matrix-wigner-bounds}
	Let $M$ be an $n \times n$ matrix whose entries $M_{ij}$ are independent random variables . Then for any $0 \leq \varepsilon \leq 1/2$ there exists a universal constant $c_{\varepsilon}$ such that,
	\begin{align*}
		\Pr{\norm{M} \geq \paren{1+\varepsilon}2\sqrt{2} \cdot \sigma + t} \leq n \cdot \exp\paren{-\frac{t^2}{c_{\varepsilon}\cdot \sigma^2_{\star}}},
	\end{align*}
	for any $t \geq 0$ where $\sigma = \max_{i}\sqrt{\sum_{j}\E\Brac{M^2_{ij}}}$ and $\sigma_{\star} = \max_{i,j}\norm{M_{ij}}_{\infty}$.
\end{theorem}
Next we would like to apply \prettyref{thm:random-matrix-wigner-bounds} to the $p$-biased matrix $B^{(2q)}$ from \prettyref{def:signed_flattening_matrix} by setting $M=B^{(2q)}$. However the entries of $B^{(2q)}$ are not independent, due to presence of symmetries. To address this we write the matrix $B^{(2q)}$ as sum of $R$ independent matrix (for a bounded $R$) as below.

\begin{construction}[Independent Block Decomposition]\label{cons:independent-block-matrix-decompose}
	For a given $2q$-uniform hypergraph $H=(V,E)$, fix an arbitrary global ordering $\mathfrak{S}$ of the vertices as,
	\begin{align*}
		v_1 \prec v_2 \prec \dots \prec v_n.
	\end{align*}
	For every hyperedge $e=\set{v_{i_1},v_{i_2},\dots,v_{i_{2q}}}\in E$ we list its $\binom{2q}{q}$ subsets of size $q$ in the order given by $\mathfrak{S}$ and form all ordered pairs of distinct subsets in the same order. Enumerating these pairs globally over all the edges we have that there are,
	\begin{align*}
		R \defeq \binom{2q}{q} \paren{\binom{2q}{q}-1} \leq 2^{2q}\cdot 2^{2q} = 2^{4q}
	\end{align*}
	many disjoint index sets given by $\Gamma_1,\Gamma_2,\dots,\Gamma_R \subset \cI \times \cI$. For each $r \in [R]$ we define the matrix $M^{(r)}=B_{r}^{(2q)}$ as below,
	\begin{align*}
		B^{(2q)}_r[I,J] = \begin{cases}
			B^{(2q)}[I,J] \quad &\text{if }(I,J) \in \Gamma_r\\
			\quad\, 0 \qquad\, &\text{otherwise}
		\end{cases}
	\end{align*}
\end{construction}

\begin{proof}[Proof of \prettyref{prop:flattening-random-matrix-bounds}]
	It is easy to note that from \prettyref{cons:independent-block-matrix-decompose} we have $B^{(2q)}=\sum_{r=1}^R B^{(2q)}_r$ where each $B^{(2q)}_r$ has independent entries and are centered $0$ mean random variables with variance at most $1$. Now fix any $r\in [R]$ and we apply \prettyref{thm:random-matrix-wigner-bounds} with $N_q=\binom{n}{q} \leq n^q/q!$ and we note that we can thus bound $\sigma^2_{\star}$ as below,
	\begin{align*}
		\sigma = \max_i\sqrt{\sum_j\E\Brac{M_{ij}^2}} \leq \sqrt{\binom{n}{q}} \leq \frac{n^{q/2}}{\sqrt{q!}}, \quad \text{and} \quad \sigma^2_{\star} \leq \frac{1}{p(1-p)}.
	\end{align*}
	We analyze this for choice of $t=\sqrt{(q+1)c_{\varepsilon}\log n/p(1-p)}$ and by \prettyref{thm:random-matrix-wigner-bounds}, for a fix $r \in [R]$,
	\begin{align*}
		\Pr{\norm{B^{(2q)}_r}_2 \geq \paren{1+\varepsilon}2\sqrt{2}\cdot \sigma + t } \leq N_q \cdot \exp\paren{-\frac{t^2p(1-p)}{c_{\varepsilon}}} \leq \paren{\frac{n^q}{q!}}n^{-(q+1)} \leq \frac{1}{n\cdot q!}
	\end{align*}
	and now by a union bound over all $R$ matrices we obtain that for any $r \in [R]$ we have,
	\begin{align}\label{eq:failure-prob-bound}
		\Pr{\norm{B^{(2q)}  }_2 \geq \paren{1+\varepsilon}2\sqrt{2}\cdot \sigma + t } \leq R \cdot \paren{\frac{1}{n\cdot q!}} \leq \frac{\paren{2^{4q}/q!}}{n} = \frac{C'_{q}}{n} 
	\end{align}
	where $C'_{q} \defeq 2^{4q}/q!$. Now we plug values of $t,\sigma$ and applying triangle inequality for norms,
	\begin{align}\label{eq:triangle-inequality-decompose-bound}
		\norm{B^{(2q)}}_2 &= \norm{\sum_{r=1}^RB^{(2q)}_r}_2 \leq \sum_{r=1}^R \norm{B^{(2q)}_r}_2 \leq R\paren{(1+\varepsilon)2\sqrt{2}\cdot \sigma + t}\\
		&\leq R\paren{(1+\varepsilon)2\sqrt{2}\cdot \paren{\frac{n^{q/2}}{\sqrt{q!}}} + \sqrt{\frac{(q+1)c_{\varepsilon}\log n}{p(1-p)}}}.
	\end{align}
	Now let $C_{\text{first}} = R(1+\varepsilon)2\sqrt{2}/\sqrt{q!}$ and $C_{\text{second}} = R\sqrt{\paren{q+1}c_{\varepsilon}}$ and using eqn.~\prettyref{eq:failure-prob-bound} and eqn.~\prettyref{eq:triangle-inequality-decompose-bound},
	\begin{align*}
		\Pr{\norm{B^{(2q}}_2 \geq C_{\text{first}}\cdot n^{q/2} + C_{\text{second}}\cdot \sqrt{\frac{\log n}{p(1-p)}}\,} \leq \frac{C'_q}{n}. 
	\end{align*}
	Now for $p \geq 1/4$ we have that $\sqrt{\log n/p(1-p)} \leq 2\sqrt{\log n/p} \leq n^{q/2}$, and for $p \leq 1/4$ it follows,
	\begin{align*}
		C_{\text{second}}\sqrt{\frac{\log n}{p(1-p)}} = \paren{\frac{C_{\text{second}}}{\sqrt{1-p}}}\sqrt{\frac{\log n}{p}} \leq \paren{\frac{2C_{\text{second}}}{\sqrt{3}}}\sqrt{\frac{\log n}{p(1-p)}} \leq 2\cdot C_{\text{second}}\sqrt{\frac{\log n}{p(1-p)}}
	\end{align*}
	Therefore it follows that with high probability (over the randomness of input),
	\begin{align*}
		\norm{B^{(2q)}}_2 \leq \paren{C_{\text{first}} + C_{\text{second}}}\cdot n^{q/2} + 2\cdot C_{\text{second}}\cdot \sqrt{\frac{\log n}{p}} \leq \kappa_{2q}\paren{n^{q/2} + \sqrt{\frac{\log n}{p}}}
	\end{align*}
	where $\kappa_{2q} = \max\set{C_{\text{first}} + C_{\text{second}},\, 2\cdot C_{\text{second}}}$.
\end{proof}

\section{Independence Number Certificates in Odd Arity Hypergraphs}
\label{sec:sos-odd-analysis}
Our main result in this section is to show a bound on pseudo-expectation mass for an odd-arity uniform random hypergraphs. Unlike the even case, matricizing an $\paren{2q+1}$-uniform hypergraph naively gives matrix of size $n^{q+1} \times n^{q+1}$ (thus we work in an inflated dimension of $2q+2$) and the operator norm bounds here yields sub-optimal results. Instead, we directly work with the order-$\paren{2q+1}$  tensor and consider its squared multilinear form, this simplifies to analyzing the matrix $M=\sum_{i\in V}A_i \otimes A_i$ where $A_i$ are slice matrices of size $n^q \times n^q$, one for each vertex. We will then proceed to prove sharper bounds on the spectral norm of $M$ using the \cite{BLNvH25} framework which allows us to bound the pseudo-expectation mass, matching the bounds for the even case in the dimension parameter $n$.

\begin{theorem}\label{thm:odd-case-bounds}
	Given an instance of a random hypergraph $H=(V,E)$ generated from random model $\cH_0\paren{n,2q+1,p}$, and a $\paren{4q+2}$-degree pseudo-expectation $\pE$ for the SoS relaxation of \hypermissdp, it follows that there exists a constant $c_{2q+1}$ such that with high probability (over the randomness of the input) the following holds,
	\begin{itemize}
		\item \textbf{Regime I}: If $p \geq p^{\star} \defeq \paren{\log n}^2/\sqrt{n}$ we have that,
		\begin{align*}
			\sum_{v\in V}\pE[x_v] \leq c_{2q+1} \cdot \paren{\frac{\sqrt{n}}{p^{1/(2q+1)}} + \frac{n^{q/(2q+1)}\cdot (\log n)^{5/(4q+2)}}{p^{2/(2q+1)}}}.
		\end{align*}
		
		\item \textbf{Regime II}: Where $p \leq p^{\star}=\paren{\log n}^2/\sqrt{n}$ we have that,
		\begin{align*}
			\sum_{v\in V}\pE[x_v] \leq c_{2q+1} \cdot \paren{\frac{\sqrt{n}\paren{\log n}^{1/(4q+2)}}{p^{1/(2q+1)}} + \frac{(\log n)^{2/(2q+1)}}{p^{3/(2q+1)}}}.
		\end{align*}
	\end{itemize}
\end{theorem}

\subsection{Sum-of-Squares Analysis for Random Hypergraphs}
\label{sec:sos-odd-analysis-details}
The proof begins similar to the even case, as we let $X\defeq \sum_{v\in V}x_v$ and proceeds by expanding $X^{2q+1}$ into ordered $\paren{2q+1}$-tuples, and further splitting into the \say{lower-order terms} (tuples with less than $\paren{2q+1}$ distinct vertices) and the \say{all distinct term}. Then we prove SoS certifiable bounds of degree at most $\paren{2q+1}$ on each of these, and use weak duality of SoS proofs to obtain a similar bounds for a degree-$\paren{4q+2}$ pseudo-expectation moments.
We start with a multinomial expansion of $X^{2q+1}$ over the ordered $\paren{2q+1}$-tuples which gives,
\begin{align*}
	X^{2q+1} = \paren{\sum_{v\in V}x_v}^{2q+1} = \sum_{\paren{v_0,v_1,\dots,v_{2q}}\in V^{2q+1}}x_{v_0}x_{v_1}\dots x_{v_{2q}}
\end{align*}
We partition these ordered tuples according to the number of distinct vertices in $\paren{v_0,v_1,\dots,v_{2q}}$ and we have the formal identity,
\begin{align} \label{eq:poly_expand_odd}
	X^{2q+1} = \underbrace{\sum_{\substack{\paren{v_0,v_1,\dots,v_{2q}}\in V^{2q+1}\\\abs{\set{v_0,v_1,\dots,v_{2q}}}=2q+1}}x_{v_0}x_{v_1}\dots x_{v_{2q}}}_{\text{all-distinct term}} + \underbrace{\sum_{\substack{\paren{v_0,v_1,\dots,v_{2q}}\in V^{2q+1}\\\abs{\set{v_0,v_1,\dots,v_{2q}}}<2q+1}}x_{v_0}x_{v_1}\dots x_{v_{2q}}}_{\text{lower-order terms}}
\end{align}
Note that the above is just an exact equality in the polynomial ring and hence SoS at any degree immediately proves it. Now, for the lower-order terms in eqn.~\prettyref{eq:poly_expand_odd} using \prettyref{prop:lower_term_bounds} we have,
\begin{align}\label{eq:lower_order_odd_bounds}
	\hypermissdp \SoSp{x}{2q+1} \set{\sum_{\substack{\paren{v_0,v_1,\dots,v_{2q}}\in V^{2q+1}\\\abs{\set{v_0,v_1,\dots,v_{2q}}}<2q+1}}x_{v_0}x_{v_1}\dots x_{v_{2q}} \leq C_{2q+1}\paren{\sum_{v\in V}x_v}^{2q}}.
\end{align}
Next, we proceed to the technically more challenging higher-order all-distinct term for which we introduced a signed tensor of order-$\paren{2q+1}$ as below,
\begin{definition}[$p$-Biased Hypergraph Tensor]\label{def:scaled_signed_tensor}
	For a given $\paren{2q+1}$-uniform hypergraph $H=(V,E)$, we define the $p$-biased symmetric hypergraph tensor (using the $p$-biased Fourier character $h_e$ defined in \prettyref{sec:p-biased-character}) denoted $B^{(2q+1)}$ as,
	\begin{align*}
		B_{j_0,j_1,\dots,j_{2q}}^{(2q+1)}=
		\begin{cases}
			h_{j_0,j_1,\dots,j_{2q}} \quad &\text{if } \abs{j_0,j_1,\dots,j_{2q}}=2q+1\\
			0 \quad &\text{otherwise}.
		\end{cases}
	\end{align*}
\end{definition}
We recall the global ordering $\mathfrak S$ on the vertices of the hypergraph, using which an unordered hyperedge $e=\set{j_0,j_1,\dots,j_{2q}}$ can be written in increasing order as,
\begin{align*}
	j_0 \prec j_1 \prec \dots \prec j_{2q.}
\end{align*}
We use the global ordering $\mathfrak S$, and symmetric tensor $B^{(2q+1)}$ to define an oriented tensor $T$ as,
\begin{align*}
	T_{j_0,j_1,\dots,j_{2q}} = \begin{cases}
		h_{j_0,j_1,\dots,j_{2q}} \quad & \text{if }j_0 \prec j_1 \prec \dots \prec j_{2q}\\
		0 \quad & \text{otherwise}.
	\end{cases}
\end{align*}
We also recall we let $\cI_q$ be the set of ordered $q$-tuples written in increasing order as below,
\begin{align*}
	\cI_q = \set{J = \paren{j_1,j_2,\dots,j_q} \in V^q : j_1 \prec j_2 \prec \dots \prec j_q}.
\end{align*}

\begin{definition}[Matrix Slice of Hypergraph Tensor]\label{def:matrix-slice}
	For a global ordering $\mathfrak S$ on the vertices, and for a fixed vertex $i\in V$, the matrix slice $A_i$ of the tensor $T$ is defined by first defining an upper-triangular slice matrix $\widetilde{A}_i$, indexed by pair of $q$-tuples $J=\paren{j_1,\dots,j_q} \in \cI_q$ and $K=\paren{k_1,\dots,k_q} \in \cI_q$ as follows,
	\begin{align*}
		\widetilde{A}_i\Brac{J,K} = \begin{cases}
			T_{i,j_1,\dots,j_q,k_1,\dots,k_q} & \text{if }i \prec j_1 \prec \dots j_q \prec k_1 \prec \dots \prec k_q\\
			0 &\text{otherwise}.
		\end{cases}
	\end{align*}
	Thus we define the symmetric matrix slice matrix $A_i$ as below,
	\begin{align*}
		A_i  = \widetilde{A}_i + \paren{\widetilde{A}_i}^{\top}
	\end{align*}
\end{definition}
\begin{remark}
	Although the purpose of oriented tensor is to remove redundancies due to symmetries involved in the natural tensor, it will be convenient for us to define our slice matrix to be symmetric, where every hyperedge appears twice.
\end{remark}
We define a vector $z \in \bbR^{\cI_q}$ for $J=\paren{j_1,\dots,j_q} \in \cI_q$ as follows,
\begin{align*}
	z_J = \prod_{t=1}^qx_{j_t}, \qquad \text{and we have that} \qquad z^{\top}A_iz = 2 \cdot z^{\top}\widetilde{A}_iz
\end{align*}
\begin{claim}\label{claim:booleanity-higher-order-claim}
	For $\paren{2q+1}$-uniform random hypergraph $H=(V,E)$ where $H \sim \cH_0\paren{n,2q+1,p}$, and the polynomial system \hypermissdp, and for our vector $z$ and any $J \in \cI_q$,
	\begin{align*}
		\hypermissdp \SoSp{x}{2q} \set{z_J^2=z_J}.
	\end{align*}
\end{claim}
\begin{proof}
	By our definition of $z_J$ we can expand it out and using our axiom \prettyref{eq:indicator_constraint_mis} we have,
	\begin{align*}
		z_J^2 = \paren{\prod_{t=1}^qx_{j_t}}^2= \prod_{t=1}^qx_{j_t}^2 = \prod_{t=1}^qx_{j_t} = z_J
	\end{align*}
	Since we only used equalities in the polynomial ring, SoS at any degree proves it.
\end{proof}

\begin{lemma}\label{lem:all-distinct-term-intermediate-bound}
	For $\paren{2q+1}$-uniform random hypergraph $H=(V,E)$ where $H \sim \cH_0\paren{n,2q+1,p}$, and the polynomial system \hypermissdp, we have
	\begin{align*}
		\hypermissdp \SoSp{x}{2q+1} \set{\sum_{\substack{\paren{v_0,\dots,v_{2q}}\in V^{2q+1}\\\abs{\set{v_0,\dots,v_{2q}}}=2q+1}}B^{(2q+1)}_{v_0,\dots,v_{2q}}\cdot x_{v_0}\dots x_{v_{2q}} =\frac{\paren{2q+1}!}{2}\cdot \sum_{i\in V}x_i\paren{z^{\top}A_iz}}.
	\end{align*}
\end{lemma}
\begin{proof}
	Since every unordered edge in $B^{(2q+1)}$ has $\paren{2q+1}!$ possible permutations we obtain the identity below,
	\begin{align*}
		\sum_{\substack{\paren{v_0,v_1,\dots,v_{2q}} \in V^{2q+1}\\\text{all distinct}}} B^{(2q+1)}_{v_0,v_1,\dots,v_{2q}}\cdot x_{v_0}x_{v_1}\dots x_{v_{2q}} = \paren{2q+1}! \sum_{v_0 \prec v_1 \dots \prec v_{2q}}T_{v_0,v_1,\dots,v_{2q}} \cdot x_{v_0}x_{v_1}\dots x_{v_{2q}}
	\end{align*}
	Fix a sorted $\paren{2q+1}$-tuple as $v_0 \prec v_1 \prec \dots \prec v_{2q}$ and let us denote $i=v_0$, $J= \paren{v_1,\dots,v_q}$ and $K=\paren{v_{q+1},\dots,v_{2q}}$. Now using vector $z$ and matrix x we can write the all-distinct term as,
	\begin{align*}
		\sum_{v_0 \prec \dots \prec v_{2q}}T_{v_0,\dots,v_{2q}} \cdot x_{v_0} \dots x_{v_{2q}} &= \sum_{i\in V}\sum_{J,K \in \cI_q}T_{i,J,K}\cdot x_iz_Jz_K = \sum_{i\in V}x_i \sum_{J,K \in \cI_q}\widetilde{A}_i[J,K]z_Jz_K\\
		&=\sum_{i\in V}x_i \paren{z^{\top}\widetilde{A}_iz} = \sum_{i\in V}x_i\paren{\frac{z^{\top}A_iz}{2}}.
	\end{align*}
	Since we only used identities in our proof, we have SoS at any degree also proves it.
\end{proof}

\begin{definition}[Kroneckered Sum of Slices Matrix]\label{def:kroncekered-matrix-define}
	For a $\paren{2q+1}$-uniform hypergraph $H=(V,E)$ with matrix slices $\set{A_i}_{i\in V}$ as defined in \prettyref{def:matrix-slice}, we define a matrix $\widetilde{M}$ as the sum of Kroneckered product of slices as below,
	\begin{align*}
		\widetilde{M} = \sum_{i\in V} \paren{A_i \otimes A_i} \in \bbR^{\paren{\cI_q \times \cI_q} \times \paren{\cI_q\times \cI_q}},
	\end{align*}
\end{definition}
It will be convenient to also define a Kroneckered vector $w \defeq z\otimes z \in \bbR^{\cI_q \times \cI_q}$.

\begin{lemma}\label{lem:all-distinct-term-intermediate-2-bound}
	For $\paren{2q+1}$-uniform random hypergraph $H=(V,E)$ where $H \sim \cH_0\paren{n,2q+1,p}$, and the polynomial system \hypermissdp, we have that,
	\begin{align*}
		\hypermissdp \SoSp{x}{4q+2} \set{\paren{\sum_{\substack{\paren{v_0,\dots,v_{2q}}\in V^{2q+1}\\\abs{\set{v_0,\dots,v_{2q}}}=2q+1}}x_{v_0}\dots x_{v_{2q}}}^2 \leq \paren{\frac{\paren{2q+1}!}{2}}^2 \paren{\frac{1-p}{p}} \paren{\sum_{i\in V}x_i}\paren{w^{\top}\widetilde{M}w}}.
	\end{align*}
\end{lemma}
\begin{proof}
	Similar to the even-arity case we start by the SoS provable equality constraint \prettyref{eq:edge_constraint_mis} to obtain,
	\begin{align*}
		\sum_{\substack{\paren{v_0,\dots,v_{2q}}\in V^{2q+1}\\\abs{\set{v_0,\dots,v_{2q}}}=2q+1}}x_{v_0}\dots x_{v_{2q}} &= 
		\sum_{\substack{\paren{v_0,\dots,v_{2q}}\in V^{2q+1}\\\abs{\set{v_0,\dots,v_{2q}}}=2q+1\\\set{v_0,\dots,v_{2q}}\notin E}}x_{v_0}\dots x_{j_{2q}} + \sum_{\substack{\paren{v_0,\dots,v_{2q}}\in V^{2q+1}\\\abs{\set{v_0,\dots,v_{2q}}}=2q+1\\\set{v_0,\dots,v_{2q}}\in E}}x_{v_0}\dots x_{v_{2q}}\\
		&=\sum_{\substack{\paren{v_0\dots,v_{2q}}\in V^{2q+1}\\\abs{\set{v_0,\dots,v_{2q}}}=2q+1\\\set{v_0,\dots,v_{2q}}\notin E}}x_{v_0}\dots x_{v_{2q}}
	\end{align*}
	Now again using constraint \prettyref{eq:edge_constraint_mis} and the definition of $B^{(2q+1)}$ we obtain the equality below,
	\begin{align*}
		\sum_{\substack{\paren{v_0,\dots,v_{2q}}\in V^{2q+1}\\\abs{\set{v_0,\dots,v_{2q}}}=2q+1}}x_{v_0}\dots x_{v_{2q}} = \Biggl(\sqrt{\frac{1-p}{p}}\,\Biggr)\sum_{\substack{\paren{v_0,\dots,v_{2q}}\in V^{2q+1}}} B^{(2q+1)}_{v_0,\dots,v_{2q}} \cdot x_{v_0}\dots x_{v_{2q}}
	\end{align*}
	Towards bounding this term, we note that it is LHS in \prettyref{lem:all-distinct-term-intermediate-bound} (pre-multiplied by factor of $\sqrt{(1-p)/p}$), and hence we continue by bounding the expression in RHS of \prettyref{lem:all-distinct-term-intermediate-bound}. We do so by squaring the expression and applying the Cauchy-Schwarz inequality (note that it is certifiable in degree-$(4q+2)$ SoS, see \prettyref{fact:sos_cauchy_schwarz}) as,
	\begin{align*}
		\paren{\sum_{i\in V}x_i\paren{z^{\top}A_iz}}^2 \leq \paren{\sum_{i\in V}x_i^2}\paren{\sum_{i\in V} \paren{z^{\top}A_iz}^2} = \paren{\sum_{i\in V}x_i}\paren{\sum_{i\in V} \paren{z^{\top}A_iz}^2}
	\end{align*}
	where the equality above follows from the SoS-provable constraint \prettyref{eq:indicator_constraint_mis}. Now with $\widetilde{M}$ as defined in \prettyref{def:kroncekered-matrix-define} and $w=z\otimes z$, and using \prettyref{claim:kronecker-quadratic-form} we have that,
	\begin{align*}
		\sum_i \paren{z^{\top}A_iz}^2 = \paren{z\otimes z}^{\top}\paren{\sum_i \paren{A_i \otimes A_i}}\paren{z \otimes z} = w^{\top} \widetilde{M} w.
	\end{align*}
	Now putting everything together we obtain the desired inequality (in degree-$(4q+2)$ SoS),
	\begin{align*}
		\paren{\sum_{\substack{\paren{v_0,\dots,v_{2q}}\in V^{2q+1}\\\abs{\set{v_0,\dots,v_{2q}}}=2q+1}}x_{v_0}\dots x_{v_{2q}}}^2 \leq \paren{\frac{(1-p)\paren{(2q+1)!}^2}{2p}}\paren{\sum_{i\in V}x_i}\paren{w^{\top}\widetilde{M}w}.
	\end{align*}
\end{proof}
Though the entries in the matrix $A_i$ (coming from random variable $h_e$) are centered, the matrix $A_i \otimes A_i$ can contain $h_e^2$ type squared terms which can be very large ($\approx 1/p$), and these must be separately removed (see eqn. 3.7  in \cite{RRS17}).
\begin{definition}[Centered Kroneckered Sum of Slice Matrix]\label{def:centered-kroneckered-matrix-define}
	We define the unordered pair notation where $\set{J,K} = \set{K,J}$ (compares the $q$-tuples as whole in set fashion).
	Now for $J,J',K,K' \in \cI_q$ , the swapped diagonal matrix $D$ is defined as,
	\begin{align*}
		D\Brac{\paren{J,J'},\paren{K,K'}} \defeq \widetilde{M}\Brac{\paren{J,J'},\paren{K,K'}} \cdot \one_{\set{J,K}=\set{J',K'}}
	\end{align*}
	and the centered Kroneckered sum of slice matrix is then defined as $M = \widetilde{M}-D$.
\end{definition}
Now we will bound the corresponding quadratic form that appears in RHS of \prettyref{lem:all-distinct-term-intermediate-2-bound} as,
\begin{align*}
	w^{\top}\widetilde{M}w = w^{\top}Dw + w^{\top}Mw.
\end{align*}
For every unordered pair $\set{J,K} \in \cI_q$ we define the (swapped) diagonal weight as,
\begin{align}\label{eq:diagonal-weights}
	d_{J,K} \defeq D\Brac{\paren{J,J},\paren{K,K}} = \sum_{i\in V}A_i^2\Brac{J,K}, \qquad \text{and} \qquad  d_{\max} \defeq \max_{J,K \in \cI_q}d_{J,K}.
\end{align}
\begin{lemma}\label{lem:all-distinct-term-intermediate-1-bound}
	For a $\paren{2q+1}$-uniform hypergraph $H=(V,E)$ where $H \sim \cH_0\paren{n,2q+1,p}$, with the polynomial system \hypermissdp, and matrix $D$ as defined in \prettyref{def:centered-kroneckered-matrix-define} it holds that,
	\begin{align*}
		\hypermissdp \SoSp{x}{4q+2} \set{w^{\top}Dw \leq d_{\max}\cdot \paren{\frac{1}{(q!)^2}}\paren{\sum_{i\in V}x_i}^{2q}}.
	\end{align*}
\end{lemma}
\begin{proof}
	Since $D$ only keeps coordinates where $\set{J,K} = \set{J',K'}$, we have that the matrix $D$ is block-diagonal with $1\times 1$ block of canonical $\set{J,K}$ coordinates. Hence for $w = z\otimes z$ we have,
	\begin{align*}
		w^{\top}Dw = \sum_{{J,K}}d_{J,K}\paren{z_Jz_K}^2 = \sum_{J,K}d_{J,K}\cdot z_Jz_K \leq d_{\max}\sum_{J,K}z_Jz_K,
	\end{align*}
	where we used \prettyref{claim:booleanity-higher-order-claim} for middle equality and the inequality holds at degree-$2q$ by simply noting that for any $J,K \in \cI_q$ the coefficient $d_{\max}-d_{J,K} \geq 0$ and hence,
	\begin{align*}
		\paren{d_{\max}-d_{J,K}}z_Jz_K = \paren{d_{\max}-d_{J,K}}\cdot \paren{z_Jz_K}^2.   
	\end{align*}
	Next we recall the SoS degree-$2q$ provable inequality below and obtain,
	\begin{align}\label{eq:set-tuple-relate-higher-order}
		\sum_{J \in \cI_q}z_J = \sum_{v_1 \prec \dots\prec v_{q}}x_{v_1}\dots x_{v_q} \leq \frac{1}{q!}\sum_{\paren{v_1,\dots,v_q} \in V^q}x_{v_1}\dots x_{v_q} = \frac{X^q}{q!}.
	\end{align}
	Now we can finish the proof by using the inequality above as,
	\begin{align*}
		\sum_{J,K}z_Jz_K = \paren{\sum_J z_J}^2 \leq \paren{\frac{X^{2q}}{(q!)^2}}.
	\end{align*}
\end{proof}

\begin{proposition}\label{prop:high-prob-bound-diagonal}
	For a $\paren{2q+1}$-uniform hypergraph $H=(V,E)$ where $H \sim \cH_0\paren{n,2q+1,p}$, matrix $D$ as defined in \prettyref{def:centered-kroneckered-matrix-define} and $d_{\max}$ as defined in eqn.~\prettyref{eq:diagonal-weights} we have that with high probability (over the randomness of the input), there is a constant $C'_{q,D}$ such that,
	\begin{align*}
		d_{\max} \leq n + C_{q,D}\paren{\sqrt{\frac{n\log n}{p}} + \frac{\log n}{p}}.
	\end{align*}
\end{proposition}
\begin{proof}
	Let $S_{J,K}$ denote the scalar random variable for the $d_{J,K}$ which we bound as,
	\begin{align*}
		S_{J,K} = \sum_{i=1}^{\min\paren{J \cup K}}h^2_{\set{i}\cup J \cup K} \leq \sum_{i=1}^{n-1}\E\left[h^2_{\set{i}\cup J \cup K}\right] + \sum_{i=1}^{n-1}\paren{h^2_{\set{i}\cup J \cup K}-\E\left[h^2_{\set{i}\cup J \cup K}\right]} 
	\end{align*}
	Now using $\Ex{h^2_{\set{i}\cup J \cup K}}=1$ we have that for any $J,K$ we can rewrite $S_{J,K}$ as,
	\begin{align*}
		S_{J,K} = \paren{n-1} + \sum_{i=1}^{n-1}\paren{h^2_{\set{i} \cup J \cup K}-1} \text{ and } \abs{S_{J,K}} \leq \paren{n-1} + \abs{\sum_{i=1}^{n-1}\paren{h^2_{\set i \cup J \cup K}-1}}.
	\end{align*}
	Now to bound the expression we consider another random variable $W_i \defeq h^2_{\set i \cup J \cup K}-1$ and note,
	\begin{align*}
		W_i = \begin{cases}
			\frac{1-2p}{p} &\quad \text{with probability }p\\
			\frac{2p-1}{1-p} &\quad \text{with probability }1-p
		\end{cases}
	\end{align*}
	Note that $\Ex{W_i}=0$ and for $p \leq 1/2$, we have that $K = \abs{W_i} \leq 1/p$ and it's variance as,
	\begin{align*}
		\Var(W_i) &= \Var(h^2_{\set i \cup J \cup K}-1) = \Var(h^2_{\set i \cup J \cup K}) = \paren{\Ex{h^4_{\set i \cup J \cup K}}} - \paren{\Ex{h^2_{\set i \cup J \cup K}}}^2\\
		&=p\paren{\frac{1-p}{p}}^2 + \paren{1-p}\paren{\frac{p}{1-p}}^2 - 1 \leq \frac{1}{p}
	\end{align*}
	Now towards bounding $\abs{\sum_i W_i}$ note that $\sigma^2\paren{\sum_i W_i} \leq n/p$ and using Bernstein's inequality,
	\begin{align*}
		\Pr{\abs{\sum_i W_i } >t} \leq 2\cdot \exp\paren{-\frac{t^2/2}{\sigma^2 + Kt/3}} \leq 2 \cdot \exp\paren{-\frac{t^2/2}{\frac{n}{p} + \frac{t}{3p}}}
	\end{align*}
	Let $\cE_{J,K}$ be the event that the variable $S_{J,K}$ deviates from it's mean by more than $t$ and by union bound over all $q$-tuples $J,K \in \cI_q$ we have,
	\begin{align*}
		\Pr{\exists (J,K): S_{J,K}-\Ex{S_{J,K}}>t} \leq \sum_{J,K}\Pr{\cE_{J,K}} \leq n^{2q} \cdot 2\exp \paren{-\frac{t^2/2}{\frac{n}{p} + \frac{t}{3p}}}.
	\end{align*}
	Hence for $t \geq 6q\paren{\sqrt{n\log n/p} + \log n/p}$ it follows that,
	\begin{align*}
		\frac{t^2/2}{\frac{n}{p} + \frac{t}{3p}} \geq 3q\log n,
	\end{align*}
	and hence with high probability (over randomness of $H$) we have that,
	\begin{align*}
		d_{\max} \leq n + 6q\cdot\sqrt{\frac{n\log n}{p}} + 6q\cdot \frac{\log n}{p}.
	\end{align*}
	This finishes the proof by letting $C'_{q,D}=6q$.
\end{proof}
\begin{lemma}\label{lem:all-distinct-term-intermediate-3-bound}
	For a $\paren{2q+1}$-uniform hypergraph $H=(V,E)$ where $H \sim \cH_0\paren{n,2q+1,p}$, with the polynomial system \hypermissdp, and matrix $M$ as defined in \prettyref{def:centered-kroneckered-matrix-define} it holds that,
	\begin{align*}
		\hypermissdp \SoSp{x}{4q+2} \set{w^{\top}Mw \leq \norm{M}_2\cdot \paren{\frac{1}{(q!)^2}}\paren{\sum_{i\in V}x_i}^{2q}}.
	\end{align*}
\end{lemma}
\begin{proof}
	Using SoS provable operator norm bounds (see \prettyref{fact:operator_sos_bound}) we have,
	\begin{align*}
		w^{\top}Mw \leq \norm{M}_2 \norm{w}_2^2.
	\end{align*}
	Now to finish the proof we use \prettyref{claim:booleanity-higher-order-claim} and eqn.~\prettyref{eq:set-tuple-relate-higher-order} to obtain the SoS provable inequality,
	\begin{align*}
		\norm{w}_2^2 = \norm{z\otimes z}^2  = \sum_{J,K \in \cI_q}\paren{z_Jz_K}^2 = \paren{\sum_{J \in \cI_q}z_{J}}\paren{\sum_{K \in \cI_q}z_K} \leq \paren{\frac{X^q}{q!}}^2.
	\end{align*}
\end{proof}

\begin{lemma}\label{lem:all-distinct-term-intermediate-4-bound}
	For $\paren{2q+1}$-uniform random hypergraph $H=(V,E)$ where $H \sim \cH_0\paren{n,2q+1,p}$, and the polynomial system \hypermissdp, we have that,
	\begin{align*}
		\hypermissdp \SoSp{x}{4q+2} \set{\paren{\sum_{\substack{\paren{v_0,\dots,v_{2q}}\in V^{2q+1}\\\abs{\set{v_0,\dots,v_{2q}}}=2q+1}}x_{v_0}\dots x_{v_{2q}}}^2 \leq \beta_q\paren{\frac{1-p}{p}}\paren{d_{\max}+\norm{M}_2} \paren{\sum_{i\in V}x_i}^{2q+1}},
	\end{align*}
	where $\beta_q \defeq \paren{(2q+1)!}^2/4\paren{q!}^2$.
\end{lemma}
\begin{proof}
	We start with \prettyref{lem:all-distinct-term-intermediate-2-bound} and then using \prettyref{lem:all-distinct-term-intermediate-1-bound} and \prettyref{lem:all-distinct-term-intermediate-3-bound} to bound $w^{\top}\widetilde{M}w$,
	\begin{align*}
		\paren{\sum_{\substack{\paren{v_0,\dots,v_{2q}}\in V^{2q+1}\\\abs{\set{v_0,\dots,v_{2q}}}=2q+1}}x_{v_0}\dots x_{v_{2q}}}^2 &\leq \paren{\frac{(2q+1)!}{2}}^2\paren{\frac{1-p}{p}}\cdot X \cdot \paren{d_{\max} + \norm{M}_2}\cdot \frac{1}{\paren{q!}^2}\cdot X^{2q}\\
		&=\beta_q\paren{\frac{1-p}{p}}\paren{d_{\max}+\norm{M}_2} \paren{\sum_{i\in V}x_i}^{2q+1}.
	\end{align*}
\end{proof}

\begin{lemma}\label{prop:best-of-both-world-tensor-bound}
	For a $\paren{2q+1}$-uniform hypergraph $H=(V,E)$ where $H \sim \cH_0\paren{n,2q+1,p}$, and for the matrix $M$ as defined in \prettyref{def:centered-kroneckered-matrix-define} and $d_{\max}$ as defined in eqn.~\prettyref{eq:diagonal-weights}, it holds that with high probability (over the randomness of the input), there is a constant $\kappa_q$ such that,
	\begin{itemize}
		\item For $p \geq p^{\star} \defeq \paren{\log n}^2/\sqrt{n}$ we have that,
		\begin{align*}
			d_{\max} + \norm{M}_2 \leq \kappa_{q} \cdot \paren{n^{q + \frac{1}{2}}+ \frac{n^q\paren{\log n}^{5/2}}{p}}.
		\end{align*}
		\item For $p \leq p^{\star} \defeq \paren{\log n}^2/\sqrt{n}$ we have that,
		\begin{align*}
			d_{\max}+\norm{M}_2 \leq \kappa_q\cdot \paren{n^{q + \frac{1}{2}}\sqrt{\log n}+ \frac{q \paren{\log n}^{2}}{p}}.
		\end{align*}
	\end{itemize}
\end{lemma}

\begin{proof}
	The details will appear in \prettyref{sec:operator-norm-odd-bounds}. We will first show in \prettyref{prop:centered-random-tensor-bound-bernstein} in \prettyref{sec:matrix-bernstein-based-bound}, using Matrix Bernstein inequality that regardless of $p$ there is a constant $\kappa^{\mathsf{Ber}}_{q}$ such that,
	\begin{align*}
		\norm{M}_2 \leq \kappa_{q}^{\mathsf{Ber}} \cdot \paren{n^{q + \frac{1}{2}}\sqrt{\log n}+ \frac{q \paren{\log n}^{2}}{p}}.
	\end{align*}
	Then in \prettyref{prop:centered-random-tensor-bound-chaos} in \prettyref{sec:matrix-chaos-based-bound} we show using Matrix Chaos inequalities that there is a constant $\kappa^{\mathsf{chaos}}_{q}$ such that,
	\begin{align*}
		\norm{M}_2 \leq \kappa_{q}^{\mathsf{chaos}} \cdot \paren{n^{q + \frac{1}{2}}+ \frac{n^q\paren{\log n}^{5/2}}{p}}.
	\end{align*}
	Next we note from the chaos bound above that the main term in the Bernstein bound term always dominates the variance terms for our regimes of interest $p \gg 1/n^{(2q+1)/2}$, and the leading chaos bound term is always better than the leading term in Bernstein bound. Therefore to get the best of both bounds we should set our switching threshold $p^{\star}$ when the variance term from chaos bound exceeds the leading term in Bernstein bound; this happens when,
	\begin{align*}
		\frac{n^q\paren{\log n}^{5/2}}{p} \geq n^{q+\frac{1}{2}}\sqrt{\log n} \implies p \leq \frac{\paren{\log n}^2}{\sqrt{n}}.
	\end{align*}
	Finally note that in either case the bounds obtained in \prettyref{prop:high-prob-bound-diagonal} for $d_{\max} \ll \norm{M}_2$ and hence using $\kappa_q \geq 2\cdot \max\set{\kappa_q^{\mathsf{Ber}},\kappa_q^{\mathsf{chaos}}}$ finishes the proof.
\end{proof}

\begin{proof}[Proof of \prettyref{thm:odd-case-bounds}]
	We recall we denote $X=\sum_{v\in V} x_v$ and we let $U = \pE\Brac{X^{2q+1}}$ and $u=U^{1/(2q+1)}$ and let $Y=\sum_{v_0,\dots,v_{2q}\in V^{2q+1}}x_{v_0}\dots x_{v_{2q}}$ with $\abs{v_0,\dots,v_{2q}}=2q+1$ (all distinct term). Using bounds form eqn.~\prettyref{eq:lower_order_odd_bounds} and \prettyref{prop:lower_term_bounds} with $f=2q+1$ it follows that,
	\begin{align*}
		\hypermissdp \SoSp{x}{4q+2} \set{X^{2q+1} \leq Y + C_{2q+1}X^{2q}}
	\end{align*}
	By \prettyref{cor:sos_soundness} there exists an degree-$\paren{4q+2}$ pseudo-expectation such that,
	\begin{align}\label{eq:pe-bound-eqn}
		U = \pE\Brac{X^{2q+1}} \leq \pE\Brac{Y} + C_{2q+1}\cdot \pE\Brac{X^{2q}}
	\end{align}
	Using pseudo-expectation Cauchy-Schwarz (\prettyref{fact:pe_holder} with $t=2$) and \prettyref{lem:all-distinct-term-intermediate-4-bound}, and use $1-p \leq 1$,
	\begin{align*}
		\pE\Brac{Y} \leq \sqrt{\pE\Brac{Y^2}} \leq\sqrt{\frac{\beta_q}{p}\paren{d_{\max} + \norm{M}_2}}\cdot \sqrt{U}.
	\end{align*}
	Now we apply \prettyref{fact:pe_holder} with $f(x)= X = \sum_{v\in V}x_v,g(x)=1,t=2q+1$ (note that $dt=2q+1$ as degree of $f$ is $1$) and we have,
	\begin{align*}
		\pE\Brac{X^{2q}} \leq \paren{\pE\Brac{X^{2q+1}}}^{2q/2q+1} = u^{2q}
	\end{align*}
	Now putting everything together in eqn.~\prettyref{eq:pe-bound-eqn} we have an inequality in $u$ as,
	\begin{align*}
		u^{2q+1} \leq \paren{\sqrt{\frac{\beta_q}{p}\paren{d_{\max} + \norm{M}_2}}}\cdot u^{q + \frac{1}{2}} + C_{2q+1}\cdot u^{2q}
	\end{align*}
	Now either $u \leq 2\cdot C_{2q+1}$ holds or otherwise we have that $C_{2q+1}\cdot u^{2q}\leq u^{2q+1}/2$ and we have,
	\begin{align*}
		\frac{u^{2q+1}}{2} \leq \paren{\sqrt{\frac{\beta_q}{p}\paren{d_{\max} + \norm{M}_2}}}\cdot u^{q + \frac{1}{2}}.
	\end{align*}
	We simplify the expression here and solve for $u$ which gives,
	\begin{align*}
		u \leq \paren{4\cdot\beta_q}^{1/(2q+1)}\paren{\frac{d_{\max} + \norm{M}_2}{p}}^{1/(2q+1)}
	\end{align*}
	Now again using \prettyref{fact:pe_holder} with $f(x)=\sum_{v\in V}x_v,g(x)=1,t=2q+1$ (note that degree $dt=2q+1$),
	\begin{align*}
		\pE[X] = \pE[f\cdot g] \leq \paren{\pE[f^t]}^{1/t} = \paren{\pE[X^{2q+1}]}^{1/(2q+1)}=u.
	\end{align*}
	Therefore using above, it follows that in either case we have that,
	\begin{align*}
		\pE\Brac{X} \leq \max\set{2C_{2q+1},\paren{4\cdot \beta_q}^{1/(2q+1)}\paren{\frac{d_{\max} + \norm{M}_2}{p}}^{1/(2q+1)}}
	\end{align*}
\end{proof}

\subsection{Bounding the Operator Norm for Centered Kroneckered Sum of Slices}
\label{sec:operator-norm-odd-bounds}
We prove two competing bounds on $\norm{M}_2$ in this section. Firstly in \prettyref{sec:matrix-bernstein-based-bound} we use the traditional Matrix concentration inequalities and obtain a bound which has extraneous logarithmic factors in $n$ (especially for $p=\tcohm(1)$) but has a sharp dependence on $p$. Later in \prettyref{sec:matrix-chaos-based-bound} we show improved bounds with sharp dependence on $n$ using the matrix chaos inequalities from \cite{BLNvH25}.

\subsubsection{Bounds Using Matrix Bernstein Inequality}\label{sec:matrix-bernstein-based-bound}
\begin{proposition}\label{prop:centered-random-tensor-bound-bernstein}
	For a $\paren{2q+1}$-uniform hypergraph $H=(V,E)$ where $H \sim \cH_0\paren{n,2q+1,p}$, and for the matrix $M$ as defined in \prettyref{def:centered-kroneckered-matrix-define}, it holds that with high probability (over the randomness of the input), there is a constant $\kappa^{\mathsf{Ber}}_{q}$ such that,
	\begin{align*}
		\norm{M}_2 \leq \kappa_{q}^{\mathsf{Ber}} \cdot \paren{n^{q + \frac{1}{2}}\sqrt{\log n}+ \frac{\paren{\log n}^{3}}{p}}.
	\end{align*}
\end{proposition}

\begin{proof}
	For every $i \in V$, we can define the (uncentered) slice contribution matrix $\widetilde{M}_i$ as,
	\begin{align*}
		\widetilde{M}_i \defeq A_i \otimes A_i,
	\end{align*}
	and the a (swapped) diagonal slice matrix $D_i$ where we have,
	\begin{align*}
		D_i\Brac{\paren{J,J'},\paren{K,K'}} \defeq \widetilde{M}_i\Brac{\paren{J,J'},\paren{K,K'}}\cdot \one_{\set{J,K} = \set{J',K'}},
	\end{align*}
	and the centered slice contribution matrix $M_i = \widetilde{M}_i-D_i$. Now using our notation we have,
	\begin{align*}
		\widetilde{M} = \sum_{i\in V}\widetilde{M}_i, \quad D=\sum_{i\in V}D_i, \quad M=\sum_{i\in V}M_i.
	\end{align*}
	Next, note that $M_i$ is symmetric by construction, since $A_i$ is symmetric as defined in \prettyref{def:matrix-slice} and $D_i$ is symmetric by construction. Note that by convention in \prettyref{def:matrix-slice}, every hyperedge contributes to exactly one slice and therefore the random variables appearing in different $A_i$'s (and hence different $M_i$'s) are functions of disjoint sets of edge indicators $Y_e$. Hence $\set{M_i}_{i\in V}$ are mutually independent. Also note that our matrices are zero-mean $\Ex{M_i}=0$. To see this we fix $\paren{J,J'},\paren{K,K'}$, and if $\set{J,K}=\set{J',K'}$ then $M_i\Brac{\paren{J,J'},\paren{K,K'}}=0$ deterministically (as we subtracted it from $D_i$) and otherwise,
	\begin{align*}
		M_i\Brac{\paren{J,K},\paren{J',K'}} = A_i\Brac{J,K}\cdot A_i\Brac{J',K'}
	\end{align*}
	Now each non-zero $A_i\Brac{\cdot,\cdot}$ corresponds to some $p$-biased $h_e$ and for $\set{J,K} \neq \set{J',K'}$, these correspond to two distinct hyperedges $e \neq e'$ which are independent and using \prettyref{fact:p-biased-properties} since $\Ex{h_e}=0$,
	\begin{align*}
		\E\Brac{M_i\Brac{\paren{J,J'},\paren{K,K'}}} = \Ex{h_e}\cdot \Ex{h_{e'}}=0.
	\end{align*}
	Now since $\Ex{M_i}=0$ entrywise, it follows that $\Ex{M_i}=0$ as a matrix. Hence $M=\sum_i M_i$ is a sum of independent, zero-mean, self-adjoint matrices and we can use Matrix Bernstein inequality.
	\begin{fact}[Matrix Bernstein Inequality, Theorem 6.6.1 in \cite{Tro15}]\label{fact: matrix-bernstein}
		Let $\set{M_i}$ be a set of independent, zero-mean self-adjoint matrices of dimension $d$. Let $L=\max_i \norm{M_i}$ and let,
		\begin{align*}
			\nu \defeq \norm{\sum_{i\in V}\Ex{M_i^2}}.
		\end{align*}
		Then for all $t \geq 0$ it holds that,
		\begin{align*}
			\Pr{\norm{\sum_{i\in V} M_i} \geq t} \leq d\cdot \exp\paren{\frac{-t^2/2}{\nu + Lt/3}}.
		\end{align*}
	\end{fact}
	In our setting, we have $d=n^{2q}$, and we will show that $\nu \approx n^{2q+1}$ and now we know a bound on $L$ that holds with high probability. The simple deterministic bound on $M_i$'s that bounds $L \leq n^q/p$ is too crude for our purpose and the corresponding bound has extraneous $n^q$ factors in the second term in bound claimed in \prettyref{prop:centered-random-tensor-bound-bernstein}. However we note that if we fix an $i\in V$ the hypergraph corresponding to $A_i$ is a $2q$-uniform hypergraph sampled from $\cH_0(n-1,2q,p)$ and hence using \prettyref{prop:flattening-random-matrix-bounds}, for a constant $\kappa_{A,q}$ with high probability (over the randomness of input),
	\begin{align*}
		\norm{A_i}_2 \leq \kappa_{A,q}\paren{n^{q/2} + \sqrt{\frac{\log n}{p}}}.
	\end{align*}
	However,the vanilla version of Matrix Bernstein in \prettyref{fact: matrix-bernstein} requires that the bound on $\norm{M_i}$ hold almost surely. Luckily, there is a version of Matrix Bernstein (stated below) which we can  use.
	\begin{theorem}[Proposition 2 in \cite{Kolt11}]\label{thm:general-matrix-bernstein}
		Let $\set M_i$ be a set of independent, zero-mean, self-adjoint matrices of dimension $d$. Let $L=\max_i \norm{\norm{M_i}_2}_{\psi_1}$ (sub-exponential Orlicz norm). Then there is a constant $C^{\mathsf{Ber}}$ such that for any $t \geq 0$, we have that with probability at least $1-e^{-t}$,
		\begin{align*}
			\norm{\sum_{i=1}^n M_i}_2 \leq  C^{\mathsf{Ber}}\paren{\sqrt{\nu(t+\log(2d))} + L\cdot \log\paren{\frac{L\sqrt{n}}{\sqrt{\nu}}}\cdot \paren{t + \log(2d)}}.
		\end{align*}
	\end{theorem}
	Note that our matrices have dimension $d=n^{2q}$.
	To compute the variance parameter $\nu$ we note that,
	\begin{align*}
		\Ex{M_i^2\Brac{\alpha,\alpha}} = \Ex{\sum_{\beta}M_i\Brac{\alpha,\beta}\cdot M_i\Brac{\beta,\alpha}} = \Ex{\sum_{\beta}\paren{M_i\Brac{\alpha,\beta}}^2} = \sum_{\beta}\E\Brac{\paren{M_i\Brac{\alpha,\beta}}^2}.
	\end{align*}
	Now whenever $M_i\Brac{\alpha,\beta} \neq 0$ it has to be a product of two distinct $p$-biased characters $h_e,h_{e'}$ (since we already removed the $e=e'$ in the swapped diagonal $D$ matrix). Therefore since $h_e,h_{e'}$ are independent and by \prettyref{fact:p-biased-properties}  we have $\Ex{\paren{h_eh_{e'}}^2}= \Ex{h_e^2}\Ex{h^2_{e'}}=1$ and hence,
	\begin{align*}
		\Ex{M_i^2\Brac{\alpha,\alpha}} =\set{\abs{\beta: M_i\Brac{\alpha,\beta} \neq 0}} \leq \abs{\beta: \beta \in \cI_q \times \cI_q} =d \leq n^{2q}.
	\end{align*}
	Next, we consider the off-diagonal entries  in $\Ex{M_i^2\Brac{\alpha,\gamma}}$ for $\alpha \neq \gamma$. Again we have that,
	\begin{align*}
		\Ex{\paren{M_i\Brac{\alpha,\gamma}}^2} = \sum_{\beta}\E\Brac{M_i\Brac{\alpha,\beta}\cdot M_i\Brac{\beta,\gamma}}
	\end{align*}
	Now each non-zero entry is a product of four $p$-biased characters corresponding to hyperedges in the slice, and the product is non-zero iff every underlying edge-character appears an even number of times, but this forces $\alpha=\gamma$ and hence every term is $0$. Thus $\Ex{M_i^2}$ is a diagonal matrix with every diagonal entry bounded by $n^{2q}$ and hence $\Ex{M_i^2} \preceq n^{2q}I$. Using triangle inequality for norms it follows that,
	\begin{align*}
		\nu = \norm{\sum_{i\in V}\Ex{M_i^2}} \leq \sum_{i\in V}\Ex{\norm{M_i^2}} \leq n \cdot n^{2q} \leq n^{2q+1}.
	\end{align*}
	Now to bound the $L$ we will use triangle inequality on Orlicz norms along with \prettyref{fact:kronecker-product-norm},
	\begin{align*}
		\norm{\norm{M_i}}_{\psi_1} \leq \norm{\norm{A_i\otimes A_i}}_{\psi_1} + \norm{\norm{D_i}}_{\psi_1} \leq \norm{\norm{A_i}}^2_{\psi_1} + \norm{\norm{D_i}}_{\psi_1}.
	\end{align*}
	We will first deterministically bound this for $D_i$ matrix (for $p \leq 1/2$) as below,
	\begin{align*}
		\norm{D_i}_2 = \max_{J,K}A_i^2\Brac{J,K} \leq \abs{h_e}^2 \leq \paren{\frac{1}{\sqrt{p(1-p)}}}^2 \leq \paren{\frac{2}{\sqrt{p}}}^2 \leq \frac{4}{p}.
	\end{align*}
	Now to bound the sub-exponential norm, we examine the sub-guassian norm of $A_i$ (denoted $\psi_2$).
	\begin{fact}[Proposition 2.5.2 in \cite{Ver18}]\label{fact:sub-gaussian=define}
		For non-negative random variable $X$ satisfying,
		\begin{align*}
			\Pr{X \geq a+b\sqrt{t}} \leq e^{-t}, \quad \forall t\geq 0,
		\end{align*}
		there is a constant $C^{\mathsf{SG}}$ such that,
		\begin{align*}
			\norm{X}_{\psi_2} \leq C^{\mathsf{SG}}\paren{a+b}.
		\end{align*}
	\end{fact}
	Since the hypergraph induced on link graph of $i$ is $2q$-uniform we can use \prettyref{thm:random-matrix-wigner-bounds} and analyze $A_i$ using \prettyref{cons:independent-block-matrix-decompose} and a union bound, where $\sigma \leq n^{q/2}$ and $\sigma_{\star} \leq 2/\sqrt{p}$ and for $u\geq 0$,
	\begin{align*}
		\Pr{\norm{A_i} \geq Cn^{q/2} +u} \leq n^q\cdot \exp\paren{-\frac{u^2}{C'/p}} = n^q \cdot \exp\paren{-cpu^2}.
	\end{align*}
	Now let $u=\sqrt{(t + q\log n)/p}$ and it follows that,
	\begin{align*}
		\Pr{\norm{A_i} \geq C''\paren{n^{q/2} + \sqrt{\frac{q\log n}{p}}}} \leq e^{-t}.
	\end{align*}
	Now using \prettyref{fact:sub-gaussian=define} it follows that setting constant $C^{\mathsf{SG}}=2^q\cdot C''$ we have,
	\begin{align*}
		\norm{\norm{A_i}}_{\psi_1} \leq C^{\mathsf{SG}}\paren{n^{q/2} + \sqrt{\frac{q\log n}{p}}}.
	\end{align*}
	\begin{fact}[Lemma 2.7.6 in \cite{Ver18}]
		A random variable $X$ is sub-gaussian if and only if $X^2$ is sub-exponential. Moreover the norms are related as,
		\begin{align*}
			\norm{X^2}_{\psi_1} = \norm{X}^2_{\psi_2}.
		\end{align*}
	\end{fact}
	Using this connection between sub-guassian and sub-exponential norm, we consider $\norm{\norm{M_i}}_{\psi_1}$ as,
	\begin{align*}
		\norm{\norm{M_i}}_{\psi_1} &= \norm{\norm{A_i^2}}_{\psi_1} = \norm{\norm{A_i}}^2_{\psi_2} \leq \paren{C^{\mathsf{SG}}\paren{n^{q/2} + \sqrt{\frac{q\log n}{p}}}}^2\\
		&=\paren{C^{\mathsf{SG}}}^2\paren{n^q + 2n^{q/2}\sqrt{\frac{q\log n}{p}} + \frac{q\log n}{p}} \leq C^{\mathsf{SE}}\paren{n^q + \frac{q\log n}{p}}.
	\end{align*}
	where $C^{\mathsf{SE}} = 2C^{\mathsf{SG}}$ and we use AM-GM inequality in the last step. Now putting it all together with $C^{\mathsf{Orl}} = 2 C^{\mathsf{SE}}$ it follows that,
	\begin{align*}
		L \leq C^{\mathsf{Orl}}\paren{n^q + \frac{q\log n}{p}}.
	\end{align*}
	Now using \prettyref{thm:general-matrix-bernstein} by setting $t=4 \log n$ we obtain that for large enough constant $C$,
	\begin{align*}
		\norm{\sum_{i=1}^nM_i}_2 \leq C\paren{n^{q + \frac{1}{2}}\sqrt{\log n} + n^q\cdot q\cdot \log^2n + q^2 \cdot \frac{\log^3n}{p}}
	\end{align*}
	Now for a large enough $n$ and $\kappa_q^{\mathsf{Ber}}=3q^2C$ we obtain the desired result that,
	\begin{align*}
		\norm{M}_2 \leq \kappa_{q}^{\mathsf{Ber}} \cdot \paren{n^{q + \frac{1}{2}}\sqrt{\log n}+ \frac{\paren{\log n}^{2}}{p}}.
	\end{align*}
\end{proof}

\subsubsection{Bounds Using Matrix Chaos Inequalities}\label{sec:matrix-chaos-based-bound}
To get logarithmic free bounds (in dimension parameter $n$) we apply the recent progress on sharp matrix chaos inequalities \cite{BLNvH25}. First we state our main result.

\begin{proposition}\label{prop:centered-random-tensor-bound-chaos}
	For a $\paren{2q+1}$-uniform hypergraph $H=(V,E)$ where $H \sim \cH_0\paren{n,2q+1,p}$, and for the matrix $M$ as defined in \prettyref{def:centered-kroneckered-matrix-define}, it holds that with high probability (over the randomness of the input), there is a constant $\kappa^{\mathsf{chaos}}_{q}$ such that,
	\begin{align*}
		\norm{M}_2 \leq \kappa_{q}^{\mathsf{chaos}} \cdot \paren{n^{q + \frac{1}{2}}+ \frac{n^q\paren{\log n}^{5/2}}{p}}.
	\end{align*}
\end{proposition}

\begin{proof}
	We apply the matrix chaos framework to bound the matrix $M$ which is a sum of Kronecker product of matrix slices. The first step is to reformulate our matrix in a form comparable with the matrix chaos as defined in \prettyref{def:combinatorial_chaos} which is a random sum over deterministic matrices. The matrix $M$ is indexed by a pair of $q$-tuples where we let row index be $\paren{J,J'}$ and column index be $\paren{K,K'}$ and an entry of the matrix $M$ is then given by (using \prettyref{def:kronecker-product-definition}),
	\begin{align*}
		M\left[\paren{J,J'},\paren{K,K'}\right] &= \sum_{i\in V}\paren{A_i\otimes A_i}\left[\paren{J,J'},\paren{K,K'}\right] \cdot \one_{\set{J,K}\neq \set{J',K'}}\\
		&=\sum_{i\in V}A_i\Brac{J,K}\cdot A_i\Brac{J',K'\cdot \one_{\set{J,K}\neq \set{J',K'}}}=\sum_{i\in V} \zeta_{i,J,K} \cdot \zeta_{i,J',K'}.\cdot \one_{\set{J,K}\neq \set{J',K'}}
	\end{align*}
	We note that $\zeta_e$ is the same $p$-biased Fourier character we called $h_e$ earlier (to avoid confusion with our notation in \prettyref{sec:matrix-chaos-overview} we call it $\zeta$ for the purpose of this proof).
	Thus our matrix $M$ is a chaos of order-$2$, where we set $r=2$ since each entry is sum of product of two random variables. This is actually a combinatorial chaos as we show by a mapping to \prettyref{def:combinatorial_chaos}. In that definition, the chaos is a sum over a vector $t$ and the random variables are $h$. For the rest of the proof we say that role of the generic random variable $h$ is played by our hyperedge variables $\zeta$. The summation vector $t$ represents the collection of all vertices used to index these two random variables. The first variable $\zeta_{i,J,K}$ is indexed by $2q+1$ vertices in $\set{i}\cup J \cup K$, and the second $\zeta_{i,J',K'}$ by $\set{i}\cup J'\cup K'$. Since they share the vertex $i$, the total number of distinct vertices in the most general case is $f=4q+1$. The vector $t$ is a tuple of these $f$ vertices, $s=\paren{i,j_1,\dots,j_q,k_1,\dots,k_q,j'_1,\dots,j'_q,k'_1,\dots,k'_q}$ where each vertex is drawn from $V=[n]$, we sum over $s \in [n]^f$, which means $T_1,\dots,T_f=n$. The index maps $I_g(s)$ (from \prettyref{def:combinatorial_chaos}) then select vertices from the long tuple vector $s$ to build chaos and matrix coordinates. We have $r=2,r+1=3$ and $r+2=4$, and the index maps are,
	\begin{align*}
		I_1(s) &= \set{i,j_1,\dots,j_q,k_1,\dots,k_q}\\
		I_2(s) &= \set{i,j'_1,\dots,j'_q,k'_1,\dots,k'_q}\\
		I_3(s) &= \set{j_1,\dots,j_q,j'_1,\dots,j'_q} && \paren{\text{row index}}\\
		I_4(s) &= \set{k_1,\dots,k_q,k'_1,\dots,k'_q} && \paren{\text{column index}}. \numberthis \label{eq:index-sets}
	\end{align*}
	Now we proceed to check our matrix is indeed a combinatorial chaos by using maps in eqn.~\prettyref{eq:index-sets},
	\begin{align*}
		Y&=\sum_{s\in [n]^{f}}h_{I_1(s)}\cdot h_{I_2(s)}\paren{e_{I_3(s)}\otimes e_{I_4(s)}^{\top}}\\
		&=\sum_{\paren{i',\widehat{J},\widehat{K},\widehat{J'},\widehat{K'}}\in [n]^f}\zeta_{i',\widehat{J},\widehat{K}}\cdot \zeta_{i',\widehat{J'},\widehat{K'}}\paren{e_{\paren{\widehat{J},\widehat{J'}}}\otimes e_{\paren{\widehat{J'},\widehat{K'}}}^{\top}}
	\end{align*}
	and we verify that $Y[\paren{J,J'},\paren{K,K'}] = \sum_{i'\in V} \zeta_{i',J,K}\cdot \zeta_{i',J',K'}$. This is exactly the definition of matrix chaos $M[\paren{J,J'},\paren{K,K'}]$ with the summation index $i$ relabeled as $i'$ and hence $Y=M$.
	Next we consider $\cA$, the underlying order-$4$ coefficient tensor for the chaos matrix and proceed to parameter calculations for $\cA$ by applying \prettyref{prop:bound_chaos_parameters}. 
	We start by first computing the $\sigma$-flattening parameter of $\cA$ of our chaos matrix. Recall that here we are required to have that $3\in R,4\in C$ and consequently we consider all four possible flattenings and compute the spectral norms of $\norm{\cA_{[R|C]}}$, the spectral norm of flattened coefficient tensor $\cA$ reshaped  into matrix where rows are indexed by $R$ and columns by $C$.
	\begin{itemize}
		\item Where all chaos coordinates are in $R$ and hence $R=\set{1,2,3},C=\set{4}$ which gives $\cR=I_1\cup I_2\cup I_3=\set{i,j_1,\dots,j_q,j'_1,\dots,j'_q,k_1,\dots,k_q,k'_1,\dots,k'_q}$ and we have that $\cC=I_4=\set{k_1,\dots,k_q,k'_1,\dots,k'_q}$ and hence $\overline{\cR}=\emptyset$ and $\overline{\cC}=\set{i,j_1,\dots,j_q,j'_1,\dots,j'_q}$ which gives,
		\begin{align*}
			\norm{\cA_{[1,2,3|4]}} = \sqrt{n^0}\cdot \sqrt{n^{2q+1}} = {n^{(2q+1)/2}}.
		\end{align*}
		
		\item Where chaos coordinates lie in both $R=\set{1,3},C=\set{2,4}$ which gives $\cR=I_1\cup I_3=\set{i,j_1,\dots,j_q,j'_1,\dots,j'_q,k_1,\dots,k_q}$ and $\cC=I_2\cup I_4=\set{i,j'_1,\dots,j'_q,k_1,\dots,k_q,k'_1,\dots,k'_q}$ and hence $\overline{\cR}=\set{k'_1,\dots,k'_q}$ and $\overline{\cC}=\set{j_1,\dots,j_q}$ which gives,
		\begin{align*}
			\norm{\cA_{[1,3|2,4]}} = \sqrt{n^q}\cdot \sqrt{n^q} = {n}^q.
		\end{align*}
		
		\item Where again chaos coordinates lie in both $R=\set{2,3},C=\set{1,4}$ which gives $\cR=I_2\cup I_3=\set{i,j_1,\dots,j_q,j'_1,\dots,j'_q,k'_1,\dots,k'_q}$ and $\cC=I_1\cup I_4=\set{i,j_1,\dots,j_q,k_1,\dots,k_q,k'_1,\dots,k'_q}$ and hence $\overline{\cR}=\set{k_1,\dots,k_q}$ and $\overline{\cC}=\set{j'_1,\dots,j'_q}$ which gives,
		\begin{align*}
			\norm{\cA_{[2,3|1,4]}} = \sqrt{n^q} \cdot \sqrt{n^q} = {n}^q.
		\end{align*}
		
		\item Where all chaos coordinates are in $C$ and hence $R=\set{3},C=\set{1,2,4}$ which gives $\cR=I_3=\set{j_1,\dots,j_q,j'_1,\dots,j'_q}$ with $\overline{\cR}=\set{i,k_1,\dots,k_q,k'_1,\dots,k'_q}$ and $\cC=I_1\cup I_2 \cup I_4=\set{i,j_1,\dots,j_q,j'_1,\dots,j'_q,k_1,\dots,k_q,k'_1,\dots,k'_q}$ and hence $\overline{\cC}=\emptyset$ and we obtain,
		\begin{align*}
			\norm{\cA_{[3|1,2,4]}} = \sqrt{n^{2q+1}} \cdot \sqrt{n^0} = {n^{(2q+1)/2}}
		\end{align*}
	\end{itemize}
	and now using the definition of the parameter $\sigma\paren{\cA}$ and above calculation we have that,
	\begin{align*}
		\sigma\paren{\cA} = \max\set{\norm{\cA_{[1,2,3|4]}},\norm{\cA_{[1,3|2,4]}},\norm{\cA_{[2,3|1,4]}}, \norm{\cA_{[3|1,2,4]}}} = n^{(2q+1)/2}.
	\end{align*}.
	
	Next we move to the $\nu$-flattening matrices defined from $\cA$ of our chaos matrix. The $\nu$-flattening requires $3,4\subseteq C$ and $R\neq \emptyset$ and there are three possibilities,
	\begin{itemize}
		\item Where all chaos coordinates are in $R$ and hence $R=\set{1,2},C=\set{3,4}$ which gives $\cR=I_1\cup I_2=\set{i,j_1,\dots,j_q,j'_1,\dots,j'_q,k_1,\dots,k_q,k'_1,\dots,k'_q}$ and we have that $\cC=I_3\cup I_4=\set{j_1,\dots,j_q,j'_1,\dots,j'_q,k_1,\dots,k_q,k'_1,\dots,k'_q}$ and hence $\overline{\cR}=\emptyset$ and $\overline{\cC}=\set{i}$ which gives,
		\begin{align*}
			\norm{\cA_{[1,2|3,4]}} = \sqrt{n^0}\cdot \sqrt{n^1} = \sqrt{n}.
		\end{align*}

		\item Where $R=\set{1},C=\set{2,3,4}$ which gives $\cR=I_1=\set{i,j_1,\dots,j_q,k_1,\dots,k_q}$ and $\cC=I_2\cup I_3 \cup I_4=\set{i,j_1,\dots,j_q,j'_1,\dots,j'_q,k_1,\dots,k_q,k'_1,\dots,k'_q}$ (so then $\overline{\cC}=\emptyset$) and we have $\overline{\cR}=\set{j'_1,\dots,j'_q,k'_1,\dots,k'_q}$ which gives,
		\begin{align*}
			\norm{\cA_{[1|2,3,4]}} = \sqrt{n^{2q}}\cdot \sqrt{n^0} = {n}^q.
		\end{align*}
		
		\item Where $R=\set{2},C=\set{1,3,4}$ which gives $\cR=I_1=\set{i,j'_1,\dots,j'_q,k'_1,\dots,k'_q}$ and $\cC=I_1\cup I_3 \cup I_4=\set{i,j_1,\dots,j_q,j'_1,\dots,j'_q,k_1,\dots,k_q,k'_1,\dots,k'_q}$ (so then $\overline{\cC}=\emptyset$) and we have $\overline{\cR}=\set{j_1,\dots,j_q,k_1,\dots,k_q,}$ which gives,
		\begin{align*}
			\norm{\cA_{[2|1,3,4]}} = \sqrt{n^{2q}}\cdot \sqrt{n^0} = {n}^q
		\end{align*}
	\end{itemize}
	and now using definition of parameter $\nu(\cA)$ and calculation above we have that,
	\begin{align*}
		\nu\paren{\cA} = \max\set{\norm{\cA_{[1,2|3,4]}},\norm{\cA_{[1|2,3,4]}},\norm{\cA_{[2|1,3,4]}}} = n^q
	\end{align*}
	Now since $h$ is a standardized centered Bernoulli random variable (see \prettyref{def:scaled_signed_tensor}) we can compute $\alpha(h)$. We let $d'=\log\paren{d+m}$ and assume $d' \geq 2$ and then,
	\begin{align*}
		\alpha(h)^2 &= \paren{\E\left[\abs{h}^{d'}\right]}^{2/d'} = \paren{\paren{1-p}\paren{\frac{p}{1-p}}^{d'/2} + p\paren{\frac{1-p}{p}}^{d'/2}}^{2/d'}\\
		&=\paren{\paren{1-p}^{1-d'/2}p^{d'/2} + p^{1-d'/2}\paren{1-p}^{d'/2}}^{2/d'}\\
		&\leq \paren{p^{d'/2} + p^{1-d'/2}}^{2/d'} \leq \paren{2\cdot p^{1-d'/2}}^{2/d'} = \frac{2^{2/d'}\cdot p^{2/d'}}{p} \leq \frac{2}{p}
	\end{align*}
	Now it follows using \prettyref{thm:chaos_thm} that with $d=\bigO\paren{n^{2q}}$ and $m=\binom{n}{2q+1}=\bigO\paren{n^{2q+1}}$,
	\begin{align*}
		\E\left[\norm{M}_2\right] \leq C_{\mathsf{MR}}\paren{n^{\paren{2q+1}/2} + \log\paren{n^{2q}+n^{2q+1}}^{5/2}\frac{2n^q}{p}}
	\end{align*}
	where $C_{\mathsf{MR}}$ is the constant from \cite{BLNvH25}.
	Now for high-probability bound we let $\phi(t)=C_{\mathsf{ISMR}}\paren{\sigma(\cA) + \alpha_{t+c\log(d+m)}(h)^r\paren{t+\log(d+m)}^{(r+3)/2}\nu(\cA)}$ and using \prettyref{thm:moment-chaos-theorem} we have that $ \paren{\E\left[\norm{Y}^t\right]}^{1/t} \leq \phi(t)$. Now for any $\lambda>0$ using Markov's inequality we obtain that,
	\begin{align*}
		\Pr{\norm{M} \geq \lambda \cdot \phi(t)} = \Pr{\norm{M}^t \geq \lambda^t \cdot \phi(t)^t}  \leq \frac{\E\left[\norm{M}^t\right]}{\lambda^t \cdot \phi(t)^t} \leq \lambda^{-t}
	\end{align*}
	Now we let $\lambda=e$ and the above application of Markov's ineqality gives us for $C_{\text{hp}}= e \cdot C_{\text{ISMR}}$,
	\begin{align*}
		\Pr{\norm{M} \geq C_{\text{hp}}\paren{\sigma(\cA) + \alpha_{t+c\log(d+m)}(h)^r\paren{t + \log(d+m)}^{(r+3)/2}\cdot \nu(\cA)}} \leq e^{-t}
	\end{align*}
	For a $\paren{2q+1}$-uniform hypergraph we already computed that $\sigma(\cA)=n^{(2q+1)/2}$ and $\nu(A)=n^q$ and we have $r=2,d=\bigO\paren{n^{2q}},m=\bigO\paren{n^{2q+1}}$, so $\log(d+m) = (2q+1)\bigO(\log n)$. Now we set $t=10 \log n$ and the term $t+c\log\paren{d+m}=\bigO\paren{\log n}$ and we can bound,
	\begin{align*}
		\alpha_{t+c\log(d+m)}(h)^{2}\paren{t+\log(d+m)^{5/2}} &\leq \paren{2/p}\paren{10\log n + c\paren{2q+1}\log n}^{5/2} \\
		&= \paren{2/p}\paren{10+c\paren{2q+1}}^{5/2}\paren{\log n}^{5/2}
	\end{align*}
	where $c_q=2\cdot\paren{10+c\paren{2q+1}}^{5/2},$ and we obtain that with probability $1-1/n^{O(1)}$,
	\begin{align*}
		\norm{M} \leq C_{\text{hp}}\paren{n^{(2q+1)/2} + \frac{c_q \cdot n^q\cdot \paren{\log n}^{5/2}}{p}}.
	\end{align*}
	We let $\kappa^{\mathsf{chaos}}_{q} = C_{\text{hp}} \cdot c_q$ and we have that with high probability (over the randomness of the input),
	\begin{align*}
		\norm{M}_2 \leq \kappa^{\mathsf{chaos}}_{q} \cdot \paren{n^{(2q+1)/2} + \frac{n^q \cdot (\log n)^{5/2}}{p}}.
	\end{align*}
\end{proof}

\section{Recovering Planted \texorpdfstring{$r$}{r}-Colorable Subhypergraphs}
\label{sec:coloring-recover}
In this section we prove our main result for recovering a $r$-colorable subhypergraphs planted in semirandom hypergraphs, a formal version of \prettyref{cor:planted-coloring-recover}. We consider a semirandom model for planting a $r$-colorable hypergraph along lines of \cite{LPR25} as below,

\begin{definition} [Arbitrary Planted-in-Semirandom Hypergraph Model]
	\label{def:semirandom_colorable_hypergraph}
	Given $n,k,r,\ell \in \mathbb{N}$ and $p\in [0,1]$, the distribution $\cH_2(n,k,r,\ell,p)$ defines a semirandom planted $r$-colorable $\ell$-uniform hypergraph and an instance from the semirandom model $H \sim \mathcal{H}_2(n,k,r,\ell,p)$ is constructed as follows,
	\begin{enumerate}
		\item Let $V=[n]$ and fix an arbitrary subset $S \subset V$ such that $\abs{S}=k$.
		\item Add hyperedges arbitrarily within $S$ such that the induced hypergraph is $r$-colorable, with color classes $S_1,\dots,S_2,\dots,S_r$ where $\abs{S_i}=k_i,i\in [r]$ and $k=\sum\limits_{i=1}^rk_i$. 
		\item For each $\ell$-tuple $\set{i_1,i_2,\dots,i_{\ell}} \subset \paren{V\setminus S}$ include the corresponding hyperedge independently at random with probability $p$.
		\item An adversary may inspect the hypergraph constructed so far and add an arbitrary set of hyperedges between the vertices in $S$ and vertices in $V\setminus S$.
		\item A monotone adversary may additionally add (but not delete) hyperedges in the hypergraph induced on the vertices in $V\setminus S$.
	\end{enumerate}
\end{definition}
The goal is that given a hypergraph $H$, output a set $\widehat{S} \subset V$ that contains almost all of $S$ and has size close to $k$. The main idea, similar to \cite{LPR25} is that since each class is an independent set, and hence one can apply the $r=1$ refutation bound to each color class separately. This is possible as the SoS objective (see formulation below) is a linear function over the vertex set.
 
\subsection{\texorpdfstring{$r$}{r}-Colorable Subhypergraph Recovery Algorithm}
We consider a hypergraph $H=(V,E)$ where $H\sim \cH_2\paren{n,k,r,\ell,p}$ and a polynomial optimization formulation (denoted $\mathcal{P}_{\textsf{$r$-Col}}$) with decision variables $\set{x_v^{(i)}:i\in [r],v\in V}$ that indicate if a vertex $v$ belongs to the color class $S_i$ given by \prettyref{def:semirandom_colorable_hypergraph}.

\newcommand{\hypercolorsdp}{\textup{\hyperref[sdp:hypergraph_color_poly]{$\mathcal{P}_{\textsf{$r$-Col}}$}}}
\begin{namedSDP}{$\mathcal{P}_{\textsf{$r$-Col}}$ (Polynomial System for $\ell$-Uniform $r$-Colorable Subhypergraphs)}
	\label{sdp:hypergraph_color_poly}
	\begin{align*}
		\max \,&\sum_{v\in V}\sum_{i\in [r]}\,{x_v^{(i)}}\\
		\qquad\text{subject to } \qquad\qquad\qquad &\\
		\prod_{u\in e}x_u^{(i)}&=0 && \forall e \in E,\forall i\in [r] \numberthis \label{eq:edge_constraint_color}\\
		x_v^{(i)}x_v^{(j)}&=0 &&\forall i,j\in [r],i\neq j \numberthis\label{eq:choose_color_class_constraint}\\
		\paren{x_{u}^{(i)}}^2 &= x_u^{(i)} && \forall u \in  V \numberthis \label{eq:indicator_constraint_color}
	\end{align*}
\end{namedSDP}
\begin{construction}[Intended Solution]\label{cons:intent_solution}
	For a given $\ell$-uniform hypergraph $H=(V,E)$ where we have $H \sim \cH_2(n,k,r,\ell,p)$, the intended solution to \hypercolorsdp\, is,
	\begin{align*}
		x_v^{(i)} = 
		\begin{cases}
			1 \quad \text{if }v\in S_i\\
			0 \quad \text{otherwise}.
		\end{cases}
	\end{align*}  
\end{construction}
We present our \sosalg\, next that does threshold rounding on degree-$(2l+1)$ pseudoexpectation $\pE$ obtained by rounding a degree-$2\ell$ SoS relaxation of \hypercolorsdp,
\paragraph{SoS Recovery Algorithm:}
\label{par:sdp_recovery_algorithm}
For an instance of $\ell$-uniform hypergraph $H=(V,E)$, and for choice of parameter $\delta>0$,
\begin{enumerate}
	\item Compute degree-$2\ell$ SoS relaxation of the polynomial system \hypercolorsdp.\\
	Let $\set{\pE\left[\mathbf{x}^{(i)}\right],i\in[r]}$ be the optimal solution to the polynomial system.
	
	\item Return $\widehat{S}=\set{v\in V: \sum\limits_{i\in [r]}\pE\left[x_v^{(i)}\right] \geq 1-\delta}$.
\end{enumerate}

\begin{theorem}[Formal Version of \prettyref{cor:planted-coloring-recover}]\label{thm:main_result_formal}
	For any choice of $\varepsilon>0$, and an $\ell$-uniform hypergraph $H=(V,E)$ generated from the semirandom model $H\sim \cH_2\paren{n,k,r,\ell,p}$ and for regimes of $k\geq \paren{c_{p,\ell}/ \varepsilon^2}\paren{r\sqrt{n}}$, and for $p\geq (\log n)^{5/2}/\sqrt{n}$, where $c_{p,\ell}$ as defined in eqn.~\prettyref{eq:final_constant}, the \sosalg\, running in time $n^{\bigO\paren{\ell}}$, outputs a set $\widehat{S}$ such that with high probability (over the randomness of input) we have,
	\begin{multicols}{2}
		\begin{itemize}
			\item $\abs{\widehat{S}} \leq \paren{1+ \varepsilon}k$
			\item $\abs{\widehat{S} \cap S} \geq \paren{1-\varepsilon}k$.
		\end{itemize}
	\end{multicols}
\end{theorem}
\noindent
\begin{remark}
    Recovery in \prettyref{thm:main_result_formal} is possible for a broader range of parameters $\paren{n,k,\ell,p}$. However, for clarity, and to highlight the sharp dependence on $n$, we restrict to the regime where the leading-order term, $\sqrt{n}/p^{1/\ell}$ governs our refutation certificate bounds.
\end{remark}
Hence, for the purpose of recovery guarantees, we are only interested in parameter regimes where the dominant term in the refutation bound is $\sqrt{n}/p^{1/\ell}$, regardless of whether $\ell$ is even or odd. For even $\ell$, this behavior always holds for $p\gg 1/n^{\ell/2}$. For odd $\ell$, it requires that,
\begin{align*}
    \frac{\sqrt{n}}{p^{1/\ell}} \geq \frac{n^{(\ell-1)/2\ell}\paren{\log n}^{5/2\ell}}{p^{2/\ell}} \implies p \geq \frac{(\log n)^{5/2}}{\sqrt{n}}.
\end{align*}
This explains lower bound on $p$ in \prettyref{thm:main_result_formal}, and for such regimes of $p$, the certificate bound is,
\begin{align*}\numberthis\label{eq:final_constant}
    c_{\ell} \cdot \frac{\sqrt{n}}{p^{1/\ell}} \defeq c_{p,\ell} \cdot \sqrt{n}.
\end{align*}
\noindent
Thus putting together results in \prettyref{sec:sos-even-analysis} and \prettyref{sec:sos-odd-analysis}  we obtain a uniform result.    
\begin{corollary}\label{cor:random_hypergraph_mis_bounds}
	For a random hypergraph $H=(V,E)$ generated from $\cH_1\paren{n,\ell,p}$, and a degree-$2\ell$ pseudoexpectation $\pE$ for the SoS relaxation of \hypermissdp\, and for regimes of $p \geq (\log n)^{5/2}/\sqrt{n}$  it follows that with with high probability (over the randomness of the input),
	\begin{align*}
		\sum_{v\in V}\pE[x_v] \leq c_{p,\ell}\cdot \sqrt{n}.
	\end{align*}
\end{corollary}

\subsection{Proof by \texorpdfstring{\cite{LPR25}}{LPR25} Analysis}
\begin{lemma}{\label{lem:bound_pe_coloring}}
	For a feasible solution to the degree-$2\ell$ SoS relaxation of \hypercolorsdp\, it follows that,
	\begin{align*}
		\sum_{i\in[r]}\pE\left[x_v^{(i)}\right] \leq 1 ,\forall v \in V.
	\end{align*}
\end{lemma}
\begin{proof}
	One can show that for the polynomial system \hypercolorsdp, it holds that,
	\begin{align*}
		\SoSp{x}{1}\paren{1-\paren{\sum_{i\in[r]}x_v^{(i)}}} \geq 0.
	\end{align*}
	This follows by first using the axiom (\ref{eq:choose_color_class_constraint}) and rewriting the above as,
	\begin{align*}
		1-\sum_{i\in[r]}x_v^{(i)} = 1-2\sum_{i\in[r]}x_v^{(i)} + \sum_{i\in [r]}x_v^{(i)} + 2\sum_{\substack{i,j\in[r]\\i\neq j}}x_v^{(i)}x_v^{(j)}
	\end{align*}
	Next using the constraint (\ref{eq:indicator_constraint_color}) it follows that,
	\begin{align*}
		1-\sum_{i\in[r]}x_v^{(i)} &= 1- 2\sum_{i\in[r]}x_v^{(i)} + \sum_{i\in [r]}x_v^{(i)} + 2\sum_{\substack{i,j\in[r]\\c\neq c'}}x_v^{(i)}x_v^{(j)}\\
		&= 1- 2\sum_{c\in[r]}x_v^{(i)} + \sum_{i\in [r]}\paren{x_v^{(i)}}^2 + 2\sum_{\substack{i,j\in[r]\\i\neq j}}x_v^{(i)}x_v^{(j)}\\
		&=\paren{1-\paren{\sum_{i\in [r]}x_v^{(i)}}}^2
	\end{align*}
	Now using \prettyref{cor:sos_soundness} one obtains the desired conclusion that $\sum_{i\in[r]}\pE\left[x_v^{(i)}\right] \leq 1$.
\end{proof}

\begin{proposition}[Proposition 17 in \cite{LPR25}]
	\label{prop:recover_large_set}
	For a optimal solution $\set{\pE\left[\mathbf{x}^{(i)}\right],i\in[r]}$ to the degree-$2\ell$ SoS relaxation of the polynomial  system \hypercolorsdp\,  such that,
	\begin{align*}
		\sum_{v\in S}\sum_{i\in [r]}\pE\left[x_v^{(i)}\right] \geq (1-\eta)k,
	\end{align*}
	it follows that $\abs{\widehat{S} \cap S} \geq k\paren{1-\frac{\eta}{\delta}}$.
\end{proposition}

\begin{proof}
	Let $Z$ be a uniform random variable that takes values in the set $\set{\sum_{i\in[r]}\pE\left[x_v^{(c_i)}\right]}_{v \in S}$ with probability $1/k$ each. Then using above one obtains that
	\begin{align*}
		\E[Z]=\sum_{v \in S}\frac{1}{k}\sum_{i\in [r]}\pE\left[x_v^{(i)}\right]= \frac{1}{k}\sum_{v \in S}\sum_{i \in [r]}\pE\left[x_v^{(i)}\right] \geq {1-\eta}.
	\end{align*}
	Using \prettyref{lem:bound_pe_coloring} one notes that $\sum_{i\in [r]}\pE[x_v^{(i)}] \leq 1,\forall v \in V$. Therefore, it also follows that $Z \leq 1$. Hence applying Markov's inequality on the random variable $1-Z$.
	\begin{align*}
		\Pr{Z \geq 1-\delta} &= 
		1-\Pr{Z\leq 1-\delta}
		= 1-\Pr{1-Z \geq \delta}\\
		&\geq 1- \frac{\E[1-Z]}{\delta} \geq 1- \frac{\eta}{\delta}.
	\end{align*}
	Since $Z$ is a uniform distribution over $\set{\sum_{i\in[r]}\pE\left[x_v^{(i)}\right]}_{v \in S}$ we have that,
	\begin{align*}
		\abs{\widehat{S}\cap S}= \abs{\set{v \in S:\sum_{i\in[r]}\pE\left[x_v^{(i)}\right] \geq 1-\delta}} \geq k\paren{1-\frac{\eta}{\delta}}.
	\end{align*} 
\end{proof}

\begin{lemma}\label{lem:sos_coloring_bound}
	Let $\set{\pE\left[\mathbf{x}^{(i)}\right],i\in[r]}$ be a feasible solution to degree-$2\ell$ SoS relaxation of the polynomial system \hypercolorsdp\, for a semirandom hypergraph $H=(V,E)$ generated from $\cH_1(n,\ell,p)$, for regimes of $p \geq (\log n)^{5/2}/\sqrt{n}$. Then with high probability (over the randomness of the input) it holds that,
	\begin{align*}
		\sum_{i\in[r]}\sum_{v\in V}\pE\left[x_v^{(i)}\right] \leq {{c_{p.\ell}}} \cdot\paren{r\sqrt{n}}.
	\end{align*}
	where $c_{p,\ell}$ as defined in eqn.~\prettyref{eq:final_constant}.
\end{lemma}
\begin{proof}
	The proof simply follows from the fact that using the constraints of the polynomial system \hypercolorsdp\, for a fixed $i\in [r]$ it follows from \prettyref{cor:random_hypergraph_mis_bounds} that,
	\begin{align*}
		\sum_{v\in V}\pE\left[x_v^{(i)}\right] \leq \paren{{c_{p,\ell}}}\sqrt{n}
	\end{align*}
	Hence by a union bound one obtains that with high probability (over the randomness of input),
	\begin{align*}
		\sum_{i\in[r]}\sum_{v\in V}\pE\left[x_v^{(i)}\right] \leq {{c_{p,\ell}}}\cdot {r\sqrt{n}}.
	\end{align*}
\end{proof}

Now we can finish the proof of \prettyref{thm:main_result_formal}, which essentially follows the proof of Theorem 9 in \cite{LPR25} for graphs ($\ell=2$). We reproduce the argument here for $\ell$-uniform hypergraphs.

\begin{proof}[Proof of \prettyref{thm:main_result_formal}].
	We start by noting that the intended feasible solution in \prettyref{cons:intent_solution} (easy to see satisfies constraints \prettyref{eq:edge_constraint_color}, \prettyref{eq:choose_color_class_constraint}, and \prettyref{eq:indicator_constraint_color} gives an objective value $k$. Thus we have,
	\begin{align}\label{eq:large_mass_planted}
		k \leq \sum_{i\in [r]}\sum_{v\in V}\pE\left[x_v^{(i)}\right] = \sum_{i\in [r]}\sum_{v\in S}\pE\left[x_v^{(i)}\right] + \sum_{i\in [r]}\sum_{v\in V\setminus S}\pE\left[x_v^{(i)}\right]
	\end{align}
	We note that the algorithm and analysis never uses the edges in the subgraph $S \times \paren{V\setminus S}$ and hence the same proof also works for arbitrary edges between vertices in $S$ and vertices in $V\setminus S$ as in \prettyref{def:semirandom_colorable_hypergraph}. Further for the monotone edges considered in the model given by \prettyref{def:semirandom_colorable_hypergraph} we note that the adding edges only increases the number of constraints and hence any bounds obtained for original SDP remain valid for new optimal solution $\set{\pE\left[{x'}_{\!\!v}^{(i)}\right],i\in [c]}$ and it follows,
	\begin{align*}\numberthis\label{eq:coloring-application-rest-mass}
		\sum_{i\in [r]}\sum_{v\in V\setminus S}\pE\left[{x'}_{\!\!v}^{(i)}\right] \leq\sum_{i\in [r]}\sum_{v\in V\setminus S}\pE\left[x_v^{(i)}\right].
	\end{align*}
	The construction in \prettyref{cons:intent_solution} is still a feasible solution and gives objective value $k$,
	Now one considers a degree-$2\ell$ SoS relaxation of \hypercolorsdp, but only on the subhypergraph induced on $V\setminus S$. Note that $\set{\pE\left[\mathbf{x}^{(i)}\right],i\in [r]}_{v\in V\setminus S}$ is still a feasible solution to the polynomial system on the hypergraph induced by $V\setminus S$ as the constraints \prettyref{eq:edge_constraint_color}, \prettyref{eq:choose_color_class_constraint}, and \prettyref{eq:indicator_constraint_color} continue to hold. Therefore using \prettyref{lem:sos_coloring_bound} in eqn.~\prettyref{eq:large_mass_planted} it follows that,
	\begin{align*}
		\sum_{i\in [r]}\sum_{v\in S}\pE\left[x_v^{(i)}\right] \geq k-{{c_{p,l}}}\cdot{r\sqrt{n}} \geq \paren{1-\eta}k.
	\end{align*}
	Finally using \prettyref{prop:recover_large_set} for $\eta=\varepsilon^2/3$ and $\delta=\varepsilon/3$ shows that $\abs{\widehat{S} \cap S} \geq \paren{1-\varepsilon}k$. From our bounds in eqn~\prettyref{eq:large_mass_planted} and eqn~\prettyref{eq:coloring-application-rest-mass} we have that,
    \begin{align*}
        \sum_{i\in \Brac{r}}\sum_{v\in V\setminus S}\pE\Brac{x_v^{(i)}} \leq c_{p,\ell}\cdot r\sqrt{n} \leq \eta k,
    \end{align*}
    and by \prettyref{lem:bound_pe_coloring} we have for every $v\in V$ that $\sum\limits_{i\in \Brac{r}}\pE\Brac{x_v^{(i)}} \leq 1$. So the total pseudo-mass is atmost $\paren{1+\eta}k$ and since each $v\in \widehat{S}$ has mass atleast $1-\delta$ we obtain,
    \begin{align*}
        \abs{\widehat{S}}\paren{1-\delta} \leq \paren{1+\eta}k.
    \end{align*}
    For our choice of parameters $\delta=\varepsilon^2/3$ and $\eta=\varepsilon/3$ and using $1/(1-x) \leq 1+x$ we have,
    \begin{align*}
        \abs{\widehat{S}} \leq \paren{1+\eta}k \cdot \paren{1+\delta} = \paren{1+\frac{\varepsilon^2}{3}}\paren{1+\frac{\varepsilon}{3}}k \leq \paren{1+\frac{\varepsilon^3}{9} + \frac{\varepsilon^2}{3} + \frac{\varepsilon}{3}}k \leq \paren{1+\varepsilon}k.
    \end{align*}
\end{proof}

\section{Low-Degree Polynomial Lower Bounds for Random Hypergraphs}
\label{sec:ldp-hypergraph-quiet-planted}

We now extend the quiet distribution in \prettyref{def:graph-quiet-planted} from the graph setting to $\ell$-uniform hypergraphs, and formally prove our main result for it.

\begin{definition}[Quiet Planted Distribution]\label{def:hypergraph-quiet-planted}
    Start with an empty hypergraph $H=(V,E)$, and for a given fixed $k \leq n/4$ we construct the quiet planted distribution $\widetilde{\sfP}$ as follows:
    \begin{enumerate}
        \item Sample i.i.d. random variables $\chi_1,\dots,\chi_n \sim \Ber(\rho)$ for $\rho= 2k/n$. Let $S \defeq \set{i:\chi_i=1}$.
        \item  Define a probability parameter for $t=0,1,\dots,\ell$ as,
        \begin{align*}
            q_t = p\paren{1- \paren{-\frac{\rho}{1-\rho}}^{\ell-t}}
        \end{align*}
        Condition on $\chi=\paren{\chi_1,\dots,\chi_n}$, sample edges as $X_e \sim \Ber\paren{q_{t_e}},\forall e \in \cE_{\ell}$, where $t_e = \abs{e \cap S}$.\\
        Let ${\sfP}'$ be the resulting planted distribution generated at this stage of the process.\label{step:unconditional-quiet-hypergraph}
        \item Let $\cG_k$ denote the good event that $\abs{S} \geq k$ and define $\widetilde{\sfP} \defeq {\sfP'}\paren{\cdot \mid \cG_k}$.
    \end{enumerate}
\end{definition}

\begin{theorem} \label{thm:quiet-planting-hypergraph-theorem}
    Fix an integer $\ell \geq 2$, and let $\sfQ_n = \cH_0\paren{n,\ell,p}$, where $0 <p \leq 1/2$. Let $D=D_n \geq 1$, and let $\varepsilon \in (0,1)$. Then, for all sufficiently large $n$, there exists $\kappa_{\ell,1},\kappa_{\ell,2},C_{\ell}>0$ depending only on $\ell$, and there exists a quiet planted distribution $\widetilde{\sfP}$ supported on $\cR(k)$ such that $\Adv^2_{\leq D}\paren{\widetilde{\sfP},\sfQ} \leq 1+\varepsilon$, for the regimes of $n,k,\ell,p,D,\varepsilon$ satisfying,
    \begin{align*}
        k \leq \paren{\kappa_{\ell,1}\cdot \varepsilon^{1/(2(\ell+1))}\cdot 2^{-C_{\ell}\cdot D^{1-\frac{1}{\ell}}\log(2D)}} \frac{\sqrt{n}}{p^{1/\ell}},\quad \text{and} \quad k \geq \kappa_{\ell,2}\paren{D\log\paren{\frac{en^{\ell}}{p}} + \log\paren{\frac{16D}{\varepsilon}}}.
    \end{align*}
\end{theorem}

\begin{proof}
    We will work with the planted distribution ${\sfP'}$ generated at end of \prettyref{step:unconditional-quiet-hypergraph} in \prettyref{def:hypergraph-quiet-planted}, and the resulting conditional distribution $\widetilde{\sfP}$. We start with the following useful result.
    
    \begin{claim}\label{claim:quick-claim-hypergraph}
        For every candidate hyperedge $e \in \cE_{\ell}$ and the planted distribution ${\sfP'}$ we have that,
        \begin{align*}
            \bbE_{\sfP'}\Brac{Y_{e} \mid \chi} = \lambda\prod_{i\in e}\paren{\chi_i-\rho}, \quad \text{where} \quad \lambda \defeq -\sqrt{\frac{p}{1-p}}\cdot\paren{\frac{1}{1-\rho}}^\ell.
        \end{align*}
    \end{claim}

    \begin{proof}
        Fix $e\in \cE_{\ell}$ and let $t=t_e = \abs{e \cap S}$. Then first we note that by definition of $q_t$ we have,
    \begin{align*}
        \bbE\Brac{Y_e \mid \chi} = \frac{q_t-p}{\sqrt{p(1-p)}} = - \sqrt{\frac{p}{1-p}}\paren{-\frac{\rho}{1-\rho}}^{\ell-t}.
    \end{align*}
    On the other hand we can write the term,
    \begin{align*}
        \prod_{i\in e}\paren{\chi_i - \rho} = \paren{1-\rho}^t\paren{-\rho}^{\ell-t},
    \end{align*}
    and verify that the expression on RHS is actually the same value by noting that,
    \begin{align*}
        \lambda \prod_{i\in e}\paren{\chi_i - \rho}= - \sqrt{\frac{p}{1-p}}\paren{1-\rho}^{-\ell}\paren{1-\rho}^t\paren{-\rho}^{\ell-t} = - \sqrt{\frac{p}{1-p}}\paren{-\frac{\rho}{1-\rho}}^{\ell-t},
    \end{align*}
    which proves our desired identity.
    \end{proof}
    Fix a candidate hyperedge set $\alpha \subseteq \cE_{\ell}$ and let $F_{\alpha}$ be the corresponding hypergraph     with set of edges given by $\alpha$ and $V(\alpha)$ as the corresponding vertex set. We also let $d_{\alpha}(u) = \deg_{F_{\alpha}}(u)$.

    \begin{lemma}
        Fix a set of candidate hyperedges $\alpha \subseteq \cE_{\ell}$. Then under the planted distribution $\sfP'$,
        \begin{align*}
            \bbE_{\sfP'}\Brac{h_{\alpha}}=\lambda^{\abs{\alpha}}\prod_{u\in V(\alpha)}\bbE\Brac{\paren{\zeta-\rho}^{d_{\alpha}(u)}}, \quad \zeta \sim \Ber(\rho).
        \end{align*}
        Further, if a subhypergraph $F_{\alpha}$ on the set of candidate hyperedges $\alpha$  contains a vertex of degree $1$, we have $\bbE_{\sfP'}\Brac{h_{\alpha}}=0$. If every vertex in $V(\alpha)$ has degree atleast $2$, then for a constant $B_{\ell} \geq 2^{2\ell+1}$ we have,
        \begin{align*}
            \paren{\bbE_{\sfP'}\Brac{h_{\alpha}}}^2 \leq \paren{B_{\ell}p}^{\abs{\alpha}}\rho^{2\abs{V(\alpha)}}.
        \end{align*}
    \end{lemma}
    
    \begin{proof}
        We start by noting that for a set $\alpha$, by the construction of planted distribution $\sfP'$, its constituent hyperedges are conditionally independent given $\chi$. Using \prettyref{claim:quick-claim-hypergraph} we have that,
        \begin{align*}
            \bbE_{\sfP'}\Brac{h_{\alpha} \mid \chi} &= \bbE_{\sfP'}\Brac{\prod_{e\in \alpha}Y_e \mid \chi} = \prod_{e \in \alpha}\bbE_{\sfP'}\Brac{Y_e \mid \chi}=\lambda^{\abs{\alpha}}\prod_{e\in \alpha}\prod_{i\in e}\paren{\chi_i-\rho} =\lambda^{\abs{\alpha}}\prod_{u\in V(\alpha)}\paren{\chi_u-\rho}^{d_{\alpha}(u)},
        \end{align*}
        where last inequality holds because every vertex $v$ appears once for every incident hyperedge. Now we take the expectation over $\chi$ and using independence of $\chi_u$'s
        we have,
        \begin{align*}
            \bbE_{\sfP'}\Brac{h_{\alpha}} = \bbE\Brac{\bbE_{\sfP'}\Brac{h_{\alpha} \mid \chi}} = \lambda^{\abs{\alpha}}\prod_{u\in V(\alpha)}\bbE\Brac{\paren{\zeta - \rho}^{d_{\alpha}(u)}}
        \end{align*}
        where $\zeta \sim \Ber(\rho)$ and thus $\bbE\Brac{\zeta - \rho}=0$. Hence, for a vertex of degree $1$ we have some $u$ has $d_{\alpha}(u)=1$ and thus some factor of $\bbE\Brac{\zeta-\rho}$ makes the coefficient $\bbE_{\sfP'}\Brac{h_{\alpha}}=0$. If all vertices in $F_{\alpha}$ have degree atleast $2$, since $\zeta \in \set{0,1}$, and for $t\geq 2$, we have $\abs{\bbE\Brac{\paren{\zeta-\rho}^t}} \leq \rho$. Hence it follows that every factor in expression of $\bbE_{\sfP'}\Brac{h_{\alpha}}$ above is bounded by $\rho$ and since $\rho \leq 1$ we have,
        \begin{align*}
            \Big\lvert{\bbE_{\sfP'}\Brac{h_{\alpha}}}\Big\rvert \leq \abs{\lambda}^{\abs{\alpha}} \prod_{u\in V(\alpha)}\bbE\Brac{\paren{\zeta-\rho}^{d_{\alpha}(u)}} \leq \abs{\lambda}^{\abs{\alpha}}\prod_{u\in V(\alpha)}\paren{\rho} \leq \abs{\lambda}^{\abs{\alpha}}\cdot \rho^{\abs{V(\alpha)}}.
        \end{align*}
        Using $\rho \leq 1/2$ and $p \leq 1/2$ so that $\abs{\lambda}^2 \leq 2^{2\ell+1}p \leq  B_{\ell}p$, and we have,
        \begin{align*}
            \paren{\bbE_{\sfP'}\Brac{h_{\alpha}}}^2 \leq \paren{B_{\ell}p}^{\abs{\alpha}}\rho^{2\abs{V(\alpha)}}.
        \end{align*}
    \end{proof}
    \begin{claim}\label{claim:handshake-hypergraph}
        Fix an $\alpha \subseteq \cE_{\ell}$, if we have $\deg\paren{F_{\alpha}} \geq 2$. Then $\abs{V\paren{\alpha}} \leq \abs{\alpha}\cdot \ell/2$.
    \end{claim}

    \begin{proof}
        This simply follows by double counting the incidences between vertices and hyperedges,
        \begin{align*}
            \ell\abs{\alpha} =  \ell \cdot \abs{\text{edges}} = \sum_{u \in V(\alpha)}d_{\alpha}(u) \geq 2\abs{V(\alpha)}.
        \end{align*}
    \end{proof}
    We first bound the advantage for the unconditioned planted distribution $\sfP'$.
    Now note that for every hypergraph where the minimum degree is $2$, there must be at least $\ell+1$ vertices. Now we fix $\abs{\alpha},\abs{V(\alpha)}$ and count how many such subhypergraphs $F_{\alpha}$ exists. We can first choose $\abs{V(\alpha)}$ many vertices and then there are $\binom{\abs{V(\alpha)}}{\ell}$ choices for edges from which we choose $\abs{\alpha}$. We can now use \prettyref{prop:advantage-computation} and plug the subhypergraph count, the coefficients and obtain,
    \begin{align*}
        \Adv^2_{\leq D}\paren{{\sfP'},\sfQ}- 1 &\leq \sum_{\abs{\alpha} \leq D}\paren{\bbE_{\sfP'}\Brac{h_{\alpha}}}^2 \leq \sum_{\abs{\alpha} \leq D}\sum_{\abs{V(\alpha)}=\ell+1}^{\abs{\alpha}\cdot \ell/2}\binom{n}{\abs{V(\alpha)}} \binom{\binom{\abs{V(\alpha)}}{\ell}}{\abs{\alpha}} \paren{B_{\ell}p}^{\abs{\alpha}} \rho^{2\abs{V(\alpha)}}
    \end{align*}
    Exchanging the order of summation we can rewrite the expression as,
    \begin{align*}\numberthis\label{eq:advantage-equation-exchanged-hypergraph}
        \Adv^2_{\leq D}\paren{{\sfP'},\sfQ}- 1 \leq \sum_{\abs{V(\alpha)}=\ell+1}^{\ell D/2}\binom{n}{\abs{V(\alpha)}}\rho^{2\abs{V(\alpha)}} \sum_{\abs{\alpha}=2\abs{V(\alpha)}/\ell}^{\min\set{D,\binom{\abs{V(\alpha)}}{\ell}}}\binom{\binom{\abs{V(\alpha)}}{\ell}}{\abs{\alpha}}\paren{B_{\ell}p}^{\abs{\alpha}}.
    \end{align*}
    Using crude bounds on binomial terms, where we let $\binom{n}{\abs{V(\alpha)}}\leq n^{\abs{V(\alpha}}$ and $\binom{\binom{V(\alpha)}{\ell}}{\abs{\alpha}} \leq \abs{V(\alpha)}^{\ell
    \cdot \abs{\alpha}}$ and using $4^{\abs{V(\alpha)}} \leq 4^{\ell\abs{\alpha}/2}$ , we have that for a constant $C_{\ell}=2^{\ell}\cdot B_{\ell}$,
    \begin{align*}
        \Adv^2_{\leq D}\paren{{\sfP'},\sfQ}- 1 
        &\leq \sum_{\abs{V(\alpha)}=\ell+1}^{\ell D/2}\sum_{\abs{\alpha}=2\abs{V(\alpha)}/\ell}^{\min\set{D,\binom{\abs{V(\alpha)}}{\ell}}}
        \paren{C_{\ell}\abs{V(\alpha)}^{\ell}p}^ {\abs{\alpha}}\paren{\frac{k^2}{n}}^{\abs{V(\alpha)}}\\
        &= \sum_{\abs{V(\alpha)}=\ell+1}^{\ell D/2}\sum_{\abs{\alpha}=2\abs{V(\alpha)}/\ell}^{\min\set{D,\binom{\abs{V(\alpha)}}{\ell}}}
        \paren{C_{\ell}\abs{V(\alpha)}^{\ell}}^ {\abs{\alpha}}\cdot p^{\abs{\alpha}-\frac{2\abs{V(\alpha)}}{\ell}}\cdot \paren{\frac{p^{2/\ell}k^2}{n}}^{\abs{V(\alpha)}}
    \end{align*}
    Using $\abs{\alpha} \geq 2\abs{V(\alpha)}/\ell$ and $p \leq 1$ the middle term involving $p$ can be upper bounded by $1$ so that,
    \begin{align*}
        \Adv^2_{\leq D}\paren{{\sfP'},\sfQ}- 1 &\leq \sum_{\abs{V(\alpha)}=\ell+1}^{\ell D/2}\paren{\frac{p^{2/\ell}k^2}{n}}^{\abs{V(\alpha)}}\sum_{\abs{\alpha}=2\abs{V(\alpha)}/\ell}^{\min\set{D,\binom{\abs{V(\alpha)}}{\ell}}}
        \paren{C_{\ell}\abs{V(\alpha)}^{\ell}}^ {\abs{\alpha}} 
    \end{align*}
    For a fixed $\abs{V(\alpha)}$, let the largest possible value for $\abs{\alpha}$ be denoted $m_{\ell,V_{\alpha}}= \min\set{D,\binom{\abs{V(\alpha)}}{\ell}}$. Now the inner summation above involves atmost $D$ terms and the expression for advantage simplifies,
    \begin{align*}\numberthis\label{eq:bound-from-here-hypergraph}
        \Adv^2_{\leq D}\paren{{\sfP'},\sfQ}- 1 \leq D\sum_{\abs{V(\alpha)}=\ell+1}^{\ell D/2}\paren{C_{\ell}\abs{V(\alpha)}^{\ell}}^ {m_{\ell,V_{\alpha}}}\cdot \paren{\frac{p^{2/\ell}k^2}{n}}^{\abs{V(\alpha)}}
    \end{align*}
    \begin{claim}\label{claim:bound-hypergraph}
        There exists $A_{\ell}\geq 0$, depending only on $\ell$ such that for every $D \geq 1$ and every $\abs{V(\alpha)}$ satisfying $\ell +1 \leq \abs{V(\alpha)} \leq \ell D/2$, and for any $C_{\ell} \geq 1$ we have that,
        \begin{align*}
            \paren{C_{\ell}\abs{V(\alpha)}^{\ell}}^{m_{\ell,V_{\alpha}}} \leq 2^{g\paren{\ell,D}\cdot\abs{V(\alpha)}},\quad \text{for} \quad g\paren{\ell,D} = A_{\ell}D^{1-\frac{1}{\ell}}\log(D).
        \end{align*}
    \end{claim}

    \begin{proof}
        Note that by definition $m_{\ell,V_{\alpha}} \leq \min\set{D,\binom{V(\alpha)}{\ell}} \leq \min \set{D,\abs{V(\alpha)}^{\ell}}$. Using AM-GM inequality, we have that for $a,b \geq 0$ and $\theta  \in \Brac{0,1}$ we have,
        \begin{align*}
            \min \set{a,b} \leq a^{1-\theta}b^{\theta}.
        \end{align*}
        Setting $a=D,b=\abs{V(\alpha)}^{\ell}$, and $\theta=1/\ell$, we have $m_{\ell,V_{\alpha}} \leq D^{1-(1/\ell)}\abs{V(\alpha)}$. Also since $\abs{V(\alpha)} \leq \ell D/2$ we have $\log\paren{C_{\ell}\abs{V(\alpha)}^{\ell}} \leq C'_{\ell}\log\paren{D}$,for a constant $C'_{\ell}$ depending only on $\ell$. Therefore,
        \begin{align*}
            \log\paren{\paren{C_{\ell}\abs{V(\alpha)}^{\ell}}^{m_{\ell,V_{\alpha}}}} = m_{\ell,V_{\alpha}}\log\paren{C_{\ell}\abs{V(\alpha)}^{\ell}} \leq C'_{\ell}\abs{V(\alpha)}D^{1-1/\ell}\log(D).
        \end{align*}
        Now taking $A_{\ell} \geq C'_{\ell}$ finishes our proof.
    \end{proof}

    Now in eqn.~\prettyref{eq:bound-from-here-hypergraph} we note that $C_{\ell}\abs{V(\alpha)}^{\ell} \geq 1$, the inner sum has atmost $D$ terms and is hence bounded by $D\paren{C_{\ell}\abs{V(\alpha)}^{\ell}}^{m_{\ell,V_{\alpha}}}$. Using \prettyref{claim:bound-hypergraph} it follows that there is a constant $A_{\ell}>0$ such that $\paren{C_{\ell}\abs{V(\alpha)}^{\ell}}^{m_{\ell,V_{\alpha}}} \leq 2^{g(\ell,D)\cdot \abs{V(\alpha)}}$. Now putting it together we have,
    \begin{align*}
        \Adv^2_{\leq D}\paren{\sfP',\sfQ}-1 \leq D\sum_{\abs{V(\alpha)}=\ell+1}^{\ell D/2}\paren{\paren{\frac{p^{2/\ell}k^2}{n}}\cdot 2^{g(\ell,D)}}^{\abs{V(\alpha)}}
    \end{align*}
    For our given regimes of $k,n,\ell,p,D,\varepsilon$ we have that $\paren{p^{2/\ell}\cdot k^2/n} \leq \kappa_{\ell,1}^2\cdot \varepsilon^{1/(\ell+1)}\cdot 2^{-2C_{\ell}\cdot g'(\ell,D)}$ where $g'(\ell,D)=D^{1-(1/\ell)}\log(D)$. Now choosing a large enough $C_{\ell}>A_{\ell}$ and $\kappa_{\ell,1}>0$ small enough (choice of $\kappa_{\ell,1} \leq 1/\sqrt{2}$ works) so that for every $D \geq 1$ and every $\varepsilon \in (0,1)$ we have,
    \begin{align*}
         \Adv^2_{\leq D}\paren{\sfP',\sfQ}-1 \leq D\sum_{\abs{V(\alpha)}=\ell+1}^{\infty}\paren{\kappa_{\ell,1}^2\varepsilon^{1/(\ell+1)}2^{-\paren{2C_{\ell}-A_{\ell}}g'(\ell,D)}}^{\abs{V(\alpha)}} \leq \frac{\varepsilon}{4},
    \end{align*}
    where we used that $2^{-(2C_{\ell}-A_{\ell})g'(\ell,D)} \leq 1$ and the geometric series has a multiplicative ratio $\kappa_{\ell,1}^2 \leq 1/2$ and we can bound it by $2\paren{\kappa_{\ell,1}^2\varepsilon^{1/(\ell+1)}}^{\ell+1}$
    and for $\ell \geq 2$, we do obtain $\Adv^2_{\leq D}\paren{\sfP',\sfQ} \leq 1+\varepsilon/4$.  Now for the planted distribution to be supported on independent set of size atleast $k$, denoted by good event $\cG_k = \set{\abs{S} \geq k}$ and recall $\widetilde{\sfP}\defeq \sfP'\paren{\cdot \mid \cG_k}$. Under our construction in \prettyref{def:hypergraph-quiet-planted} it is easy to verify that $\widetilde{\sfP}$ is supported $\cR(k)$. Now note that $\bbE\brac{\abs{S}}=2k$ by Chernoff bounds for a constant $c'$ of our choice, we have,
    \begin{align*}
        \delta_n \defeq {\sfP'}\paren{\cG_k^{c}} = \Pr{\abs{S} < k} \leq e^{-c'k}.
    \end{align*}
    Next, note that for any $\alpha$ where $\abs{\alpha} \leq D$,  for $p \leq 1/2$ we have $\abs{Y_e}\leq 1/\sqrt{p}$ and $\abs{h_{\alpha}}\leq p^{-\abs{\alpha}/2}$. So,
    \begin{align*}
        \abs{\bbE_{\widetilde{\sfP}}\Brac{h_{\alpha}} - \bbE_{\sfP'}\Brac{h_{\alpha}}} \leq 2\delta_n\cdot p^{-\abs{\alpha}/2}
    \end{align*}
    Now using the inequality that $(a+b)^2 \leq 2a^2+2b^2$ we can rewrite the above as,
    \begin{align*}
        \sum_{\abs{\alpha}=1}^D\paren{\bbE_{\widetilde{\sfP}}\Brac{h_{\alpha}}}^2 \leq 2\sum_{\abs{\alpha}=1}^D \paren{\bbE_{\sfP'}\Brac{h_{\alpha}}}^2 + 8\delta_n^2 \sum_{\abs{\alpha}=1}^D\binom{\binom{n}{\ell}}{\abs{\alpha}}p^{-\abs{\alpha}}.
    \end{align*}
    We already argued that the first term is $\varepsilon/2$ and for the second term,
    \begin{align*}
        \sum_{\abs{\alpha}=1}^D\binom{\binom{n}{\ell}}{\abs{\alpha}}p^{-\abs{\alpha}} \leq D\paren{\frac{en^\ell}{p}}^D.
    \end{align*}
    Therefore the conditioning error is now atmost,
    \begin{align*}
        8D\cdot \exp\paren{-2c'k + D\log \paren{en^\ell/p}}.
    \end{align*}
    By choosing sufficiently large enough $\kappa_{\ell,2} \geq 1/2c'$ (depending only on $\ell$), the bound given by,
    \begin{align*}
        k \geq \kappa_{\ell,2}\paren{D\log \paren{\frac{en^{\ell}}{p}} + \log \paren{\frac{16D}{\varepsilon}}},
    \end{align*}
    ensures the conditioning error is atmost $\varepsilon/2$. Putting together we obtain,
    \begin{align*}
        \Adv^2_{\leq D}\paren{\widetilde{\sfP},\sfQ} \leq 1 + 2\cdot \frac{\varepsilon}{4} + \frac{\varepsilon}{2} = 1+\varepsilon.
    \end{align*}
\end{proof}

\paragraph{Acknowledgements.}
This research was supported in part by the International Centre for Theoretical Sciences (ICTS) for the discussion meeting on Geometry, Probability, and Algorithms (ICTS/\!/gpa2025/05).
AL was supported by Walmart Centre for Tech Excellence at IISc, and SERB award CRG/2023/002896. RP was supported by Prime Minister's Research Fellowship, India. We thank Santosh Vempala and Jeff Xu for helpful discussions. We thank the anonymous reviewers for their comments.

\newcommand{\etalchar}[1]{$^{#1}$}
\providecommand{\bysame}{\leavevmode\hbox to3em{\hrulefill}\thinspace}
\providecommand{\MR}{\relax\ifhmode\unskip\space\fi MR }
\providecommand{\MRhref}[2]{%
  \href{http://www.ams.org/mathscinet-getitem?mr=#1}{#2}
}

\appendix
\section{Reduction to Planted Cliques in Graphs}
\label{app:graph_reduction}
We start by discussing the special case of $r=1$, where the planted substructure is an independent set. The setup is often referred to as the planted hypergraph clique problem, as the independent set corresponds to a clique in the complement graph. A key observation, due to \cite{ZX18} is that the problem admits a natural reduction to the well-studied planted clique problem in graphs.

\begin{proposition}[Section 4.1 in \cite{ZX18}]
	For an $\ell$-uniform hypergraph $H=(V,E)$ with a planted independent set $S$ (as in \prettyref{def:semirandom_colorable_hypergraph} with $r=1$), with high probability can be reduced to graph $G=(\cU,\cF)$ on $n-\paren{\ell-2}$ vertices with planted independent set $S' \subset S$ of $k-\paren{\ell-2}$ vertices.
\end{proposition}

\begin{proof}
	We consider a natural reduction from hypergraph planted clique/independent set to graph planted clique/independent set problem as discussed in \cite{ZX18},
	\begin{construction}\label{cons:hypergraph_to_graph}
		Given an instance of a hypergraph $H=(V,E)$ generated according to \prettyref{def:semirandom_colorable_hypergraph}, the reduction proceeds as follows:
		\begin{enumerate}
			\item Fix a set $U=\set{u_1,u_2,\dots,u_{\ell-2}} \subset V$ of $\ell-2$ vertices chosen uniformly at random.
			\item Construct a graph $G=(\cU,\cF)$ where $\cU=V\setminus U$ and for any $i,j\in \cU$, include the edge $\set{i,j}\cup U$ to $\cF$ if the corresponding hyperedge $\set{i,j}\cup U \in E$ is present in the hypergraph.
		\end{enumerate}
	\end{construction}
	We observe that the reduction in \prettyref{cons:hypergraph_to_graph} is an ideal reduction for recovery (yields a graph where $\cU$ is a clique/independent)
	set if the randomly chosen set $U$ lies entirely inside the planted set $S$. In that case, the reduced graph $G$ over $\cU$ has $S\setminus U$ as an independent set of size $k-(\ell-2)$ in a random \Erdos-\Renyi graph. One can repeat the reduction $T$ times by selecting sets $U_1,U_2,\dots,U_T$ independently at random. The probability that a fixed set $U_t$ is fully contained in $S$ is given by,
	\begin{align*}
		\Pr{U_t \subset S} = \frac{\dbinom{k}{\ell-2}}{\dbinom{n}{\ell-2}}= \bigO_{\ell}\paren{\frac{k}{n}}^{\ell-2} \text{and hence for } T=\theta_{\ell}\paren{\paren{\frac{n}{k}}^{\ell-2}\log(1/\delta)}
	\end{align*}
	many repetitions, we have that with probability $1-\delta$ there exists a $t^*\in [T]$ such that $U_{t^*} \subset S$.
\end{proof}

In the light of this reduction, algorithmic results for the graph setting, such as those in \cite{AKS98,FK00} can be directly applied. 
In particular, one can recover the planted clique when $k=\tcohm_p\paren{\sqrt{n}}$, matching the conjectured computational threshold for the problem in graphs and hypergraph settings. This shows that for the $r=1$ hypergraph case , the problem of detecting and recovering planted clique/independent set is no harder than the graph counterpart; the computational threshold on the planted size $k$ in hypergraphs is at most that for the graphs.

However, this is also an issue for us since the algorithms for graphs in \cite{AKS98,FK00} also require a lower bound of $p \gtrsim \sqrt{\log n/n}$, beyond these values of $p$ the certificates yield vacuous bounds (bounds larger than $n$) on the certificate of independent set in $G(n,p)$. Since we are reducing our hypergraph problem to a graph with same probability $p$, we need the same lower bound of $p \gtrsim \sqrt{\log n/n}$ for the above analysis to yield non-vacuous results.

The problem of detecting and recovering a planted $r$-colorable subhypergraph in an $\ell$-uniform random hypergraph is significantly more challenging when $r \geq 2$. 
In our setting, the reduction in \prettyref{cons:hypergraph_to_graph} fails for $r \geq 2$, since the reduction may no longer preserve the coloring structure.

\begin{figure}[htbp]
	\includegraphics[scale=0.3]{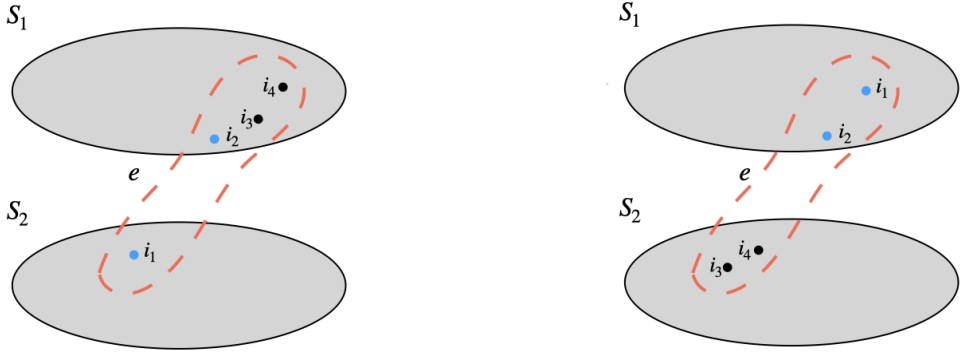}
	\caption{Difficulties with Reduction from $2$-Colorable $4$-Uniform Planted Hypergraphs to Graphs}
	\centering
\end{figure}

\begin{itemize}
	\item \textbf{Attempt 1 (Fail).} If the set $U \subseteq S$ (the $\ell-2$ vertices) contain vertices from different color classes, a hyperedge $e=U \cup \set{i_{\ell-1},i_{\ell}}$ might have $i_{\ell-1}$ and $i_{\ell}$ belonging to the same color class (say $S_1$). The reduction would create a graph edge $\set{i_{\ell-1},i_{\ell}}$ which is an edge within a color class $S_1$, thus violating the $r$-colorability. This is shown on the left in figure above, where $U = \set{i_1,i_2}$ and the violating edge $\set{i_3,i_4}$.

	\item \textbf{Attempt 2 (Fail).} Even if one attempts to refine this reduction by requiring that the set $U$ belongs to a single color class (say $ U \subset S_1$), the issue persists. There may still be an hyperedge $e = U \cup \set{i_{\ell-1},i_{\ell}}$ such that the vertices $i_{\ell-1}$ and $i_{\ell}$ both belong to the same color class (say $S_2$). The reduction then creates a graph edge $\set{i_{\ell-1},i_{\ell}}$ which is again within a single color class and violates the planted coloring structure. This is depicted in the right of the figure, where $U = \set{i_1,i_2} \subset S_1$ and the violating edge $\set{i_3,i_4}\subset S_2$
\end{itemize}

\section{Supplemental Proofs and Technical Discussion}
\subsection{Independence Number Certificates in \texorpdfstring{\cite{AOW15}}{AOW15}}
\label{app:aow-proof-reconstruct}
We state the main technical result in \cite{AOW15} on polynomials with random coefficients.
\begin{theorem}[Theorem 4.1 in \cite{AOW15}]\label{thm:random-poly-bound-aow15}
	Let $\set{w(T)}_{T\in [n]^{\ell}}$ be independent random variables such that for each $T \in [n]^{\ell}$ we have $\E[w(T)]=0, \Pr{w(T)\neq 0}\leq p,\abs{w(T)}\leq 1$, then for $p \geq 1/n^{\ell/2}$ and $\ell \geq 2$ there is an efficient algorithm that certifies for all $x \in \bbR^n$ such that $\norm{x}_{\infty} \leq 1$ that with high probability,
	\begin{align*}
		\sum_{T\in [n]^{\ell}}\paren{w(T)\prod_{u\in T}x_u} \leq \bigO_{\ell} \paren{\sqrt{p}\cdot n^{3\ell/4} \cdot \log^{3/2}(n)}.
	\end{align*}
\end{theorem}
Now they construct a collection of random variables for all potential hyperedges
$e \in {\binom{\Brac{n}}{\ell}}$ of an $\ell$-uniform random hypergraph $H=(V,E)$ they associate an arbitrary tuple $T_e$ among the $\ell!$ tuples in $[n]^{\ell}$ containing the same set of $\ell$ elements. For each $T \in [n]^{\ell}$, they define random variables as,
\begin{align*}
	w(T)  =\begin{cases}
		-(1-p) &\quad \text{if }T=T_e, e\in E\\
		\qquad p &\quad \text{if }T=T_e,e\notin E\\
		\qquad 0 &\quad \text{otherwise}
	\end{cases}
\end{align*}
Now for the indicator $x \in \set{0,1}^n$ of an independent set $S$ so that $\prod_{u\in T}x_u=1$ if $T \subset S$ and $0$ otherwise. Then they observe (see proof of Theorem 2.9 in \cite{AOW15}) that,
\begin{align*}
	\sum_{T\in [n]^{\ell}}w(T)\prod_{u\in T}x_u = \sum_{e\in \binom{[n]}{\ell}}w(T_e)\prod_{u\in E}x_u = p\sum_{e\in \binom{[n]}{\ell}}\prod_{u\in e}x_u - \paren{\sum_{e\in \binom{[n]}{\ell}}\one_{e\in E}\prod_{u\in e}x_u}
\end{align*}
Now the second term in expression above is non-zero iff  $e\in E$ and $\prod_{u\in e}x_u=1$ which means $e \subseteq S$ but since $S$ is an independent set this can never happen simultaneously. The summand in the first term is also $1$ if $\prod_{u\in S}x_u=1$ which happens if $e \subseteq S$ and hence it simply counts number of size-$\ell$ subsets of $S$ and by \prettyref{thm:random-poly-bound-aow15} one obtains a certifiable bound of,
\begin{align*}
	p\binom{\abs{S}}{\ell} = \sum_{T\in [n]^{\ell}}w(T)\prod_{u\in T}x_u \leq \bigO_{\ell} \paren{\sqrt{p}\cdot n^{3\ell/4} \cdot \log^{3/2}(n)}
\end{align*}
Solving this for $k=\abs{S}$ yields a certifiable bound on the size of independent set as,
\begin{align*}
	k = \abs{S} \leq \bigO_{\ell}\paren{\frac{n^{3/4}\cdot \log^{3/2\ell}(n)}{\sqrt{p}}}.
\end{align*}
The bound obtained in \cite{AOW15} is lossy for independent set problem as they use a generic bound that holds for a large class of random polynomials, and apply it to an indirect proxy polynomial for the problem. We on the other hand analyze the natural polynomial that directly captures the independent set problem, which enables us to use the algebraic structure in our analysis and this gets rid of the unnecessary $n^{1/4}\cdot \mathsf{polylog}(n)$ factors.

\subsection{Analyzing the Canonical Planted Distribution for Hypergraphs}
\label{app:analyzing-canonical-planted-distribution-hypergraphs}
We recall the canonical planted distribution in \prettyref{cons:canonical-planted-distribution-hypergraph}. This is the natural planted distribution associated with the problem, and appears in the pseudocalibration approach for SoS lower bounds for planted cliques \cite{BHK+16}, and also considered in the work \cite{JPR+21} for sparse independent sets where they discuss it's limitation towards proving sharp low-degree polynomial lower bounds.

For an $\ell$-uniform hypergraph the distribution $\sfP'$ can be constructed as below:
\begin{enumerate}
    \item Sample a random hypergraph $\cH_0=(V,E)$ as per \prettyref{def:random-hypergraph}.
    \item Choose a subset $S \subset V$ by picking each vertex at random with probability $\rho = 2k/n$. Let $\chi_i$ be an indicator random variable for a vertex $i$ which denotes if the vertex is chosen or not.
    \item Remove every hyperedge $e \in E$ if $e \subseteq S$, i.e. set $X_{e}=0$.
\end{enumerate}

\begin{theorem}\label{thm:canonical-planting-ldp-bound}
    Fix integers $\ell \geq 2, D \geq 1$. Let $\sfQ=\cH_0\paren{n,\ell,p}$, with $p \leq 1/2$, and $\sfP'$ as defined above. Let $\cG_k$ denote the good event that $\abs{S} \geq k$. Then $\sfP \defeq \sfP'\paren{\cdot \mid \cG_k}$ is supported on $\cR(k)$, and we have,
    \begin{align*}
        \Adv_{\leq D}\paren{\sfP,\sfQ} = 1 + o(1), \quad \text{for} \quad k=o\paren{\frac{\sqrt{n}}{p^{1/2\ell}}}
    \end{align*}
\end{theorem}

\begin{proof}
    Fix a candidate edge set $\alpha \subseteq \cE_{\ell}$, and let $F_{\alpha}$ be the corresponding hypergraph with set of hyperedges given by $\alpha$ and $V(\alpha)$ be the corresponding vertex set.
    \begin{claim}\label{claim:canonical-coefficient-bounds}
        For a fixed $\alpha \subseteq \cE_{\ell}$, and for the planted distribution $\sfP'$ as defined above, we have that,
        \begin{align*}
            \paren{\bbE_{\sfP'}\Brac{h_{\alpha}}}^2 \leq \paren{2p}^{\abs{\alpha}}\cdot \rho^{2\abs{V(\alpha)}}.
        \end{align*}
    \end{claim}
    \begin{proof}
        Conditioned on $S$, we have the edge variables $\set{Y_e}_{e\in \alpha}$ are independent. Hence,
        \begin{align*}
            \bbE_{\sfP'}\Brac{h_{\alpha} \mid S} = \bbE_{\sfP'}\Brac{\prod_{e\in \alpha}Y_e \mid S} = \prod_{e\in \alpha}\bbE_{\sfP'}\Brac{Y_e \mid S} = -\paren{\sqrt{\frac{p}{1-p}}}^{\abs{\alpha}} \prod_{e\in \alpha}, \one_{e\subseteq S}
        \end{align*}
        where the last inequality uses the fact that if $\alpha$ contains some $e \nsubseteq S$ then $Y_e$ is a normalized $\Ber(p)$ random variable that has $0$ mean. Now, the event every hyperedge in $\alpha$ is contained in $S$ is equivalent to the event that $V(\alpha) \subseteq S$. 
        \begin{align*}
            \bbE_{\sfP'}\Brac{h_{\alpha} \mid S}= \paren{-\sqrt{\frac{p}{1-p}}}^{\abs{\alpha}}\prod_{e\in \alpha}\one_{\alpha \subseteq S}  = \paren{-\sqrt{\frac{p}{1-p}}}^{\abs{\alpha}}\prod_{e\in \alpha}\one_{V(\alpha) \subseteq S}
        \end{align*}
        Taking expectation over $S$, and by the construction of $\sfP'$ we have vertices are included with probability $\rho$, which allows for an easy computation that yields,
        \begin{align*}
            \bbE_{\sfP'}\Brac{h_{\alpha}} = \paren{-\sqrt{\frac{p}{1-p}}}^{\abs{\alpha}}\cdot \Pr{V(\alpha)\subseteq S} = \paren{-\sqrt{\frac{p}{1-p}}}^{\abs{\alpha}} \cdot \rho^{\abs{V(\alpha)}}
        \end{align*}
        For $p \leq 1/2$ we have $p/(1-p)\leq 2p$, and squaring the expression above finishes the proof.
    \end{proof}
    Now we consider the conditioned planted distribution $\sfP$ defined as $\sfP=\sfP'\paren{\cdot \mid \cG_k}$ where $\cG_k$ is the event that $\abs{S} \geq k$. Note that $\abs{S} \sim \mathsf{Bin}\paren{n,\rho}$ and $\bbE\Brac{\abs{S}}=n\cdot \rho=2k$. Using Chernoff bounds,
    \begin{align*}
        \Pr{\cG_k^{c}} = \Pr{\abs{S}<k} = \Pr{\abs{S}< \frac{\bbE\Brac{\abs{S}}}{2}} \leq e^{-\tcohm(k)} = o_n(1), \quad \implies \quad \Pr{\cG_k} = 1-o_n(1).
    \end{align*}
    For the conditional distribution $\sfP$ we have $\bbE_{\sfP}\Brac{h_{\alpha}} = \bbE_{\sfP'}\Brac{h_{\alpha} \mid \cG_k}$. Now using the conditional calculation in \prettyref{claim:canonical-coefficient-bounds}, we now have with additional conditioning on $\cG_k$ as,
    \begin{align*}\numberthis\label{eq:conditional-canonical-expression}
        \paren{\bbE_{\sfP}\Brac{h_{\alpha}}}^2 = \paren{\frac{p}{1-p}}^{\abs{\alpha}} \cdot \paren{\Pr{V(\alpha)\subseteq S \mid \cG_k}}^2
    \end{align*}
    Again for $p \leq 1/2$ the first term is atmost $2p$ and conditional probability can be written as,
    \begin{align*}
        \Pr{V(\alpha) \subseteq S \mid \cG_k} = \frac{\Pr{V(\alpha)\subseteq S,\cG_k}}{\Pr{\cG_k}} \leq \frac{\Pr{V(\alpha)\subseteq S}}{\Pr{\cG_k}}.
    \end{align*}
    Now putting this back in eqn.~\prettyref{eq:conditional-canonical-expression} and using $\Pr{\cG_k}=1-o(1)$ we have,
    \begin{align*}
        \paren{\bbE_{\sfP}\Brac{h_{\alpha}}}^2 \leq \paren{\frac{p}{1-p}}^{\abs{\alpha}}\cdot \frac{\rho^{2\abs{V(\alpha)}}}{{\Pr{\cG_k}}^2} \leq \paren{1+o_n(1)}\paren{2p}^{\abs{\alpha}}\rho^{2\abs{V(\alpha)}}.
    \end{align*}
    Now using \prettyref{prop:advantage-computation} and plugging the above bound we have,
    \begin{align*}
        \Adv^2_{\leq D}\paren{\sfP,\sfQ} - 1 = \sum_{1\leq \abs{\alpha}\leq D}\paren{\bbE_{\sfP}\Brac{h_{\alpha}}}^2 \leq \paren{1+o(1)}\sum_{\abs{\alpha}=1}^D\sum_{\alpha \in \cE_{\ell}}\paren{2p}^{\abs{\alpha}}\cdot \rho^{2\abs{V(\alpha)}}.
    \end{align*}
    Now we group terms according to $\abs{V(\alpha)}$. For any $\ell$-uniform candidate hyperedge set $\alpha$, at least one $\ell$ edge contributes $\ell$ vertices and $\abs{\alpha}$ edges contribute at most $\abs{V(\alpha)}\ell$ vertices, so $\ell \leq \abs{V(\alpha)}\leq \abs{\alpha}\cdot \ell$. For a fixed $\abs{\alpha},\abs{V(\alpha)}$, the number of possible subhyperpgraphs $F_{\alpha}=\paren{V(\alpha),\alpha}$ with $\abs{\alpha}$ hyperedges and $\abs{V(\alpha)}$ vertices is atmost,
    \begin{align*}
        \binom{n}{\abs{V(\alpha)}} \cdot \binom{\binom{\abs{V(\alpha)}}{\ell}}{\abs{\alpha}},
    \end{align*}
    which follows by first choosing $\abs{V(\alpha)}$ many vertices and then there are $\binom{V(\alpha)}{\ell}$ choices for hyperedges out of which one chooses $\abs{\alpha}$. Now $\ell,D$ are fixed and so  we can upper bound as,
    \begin{align*}
         \binom{n}{\abs{V(\alpha)}} \cdot \binom{\binom{\abs{V(\alpha)}}{\ell}}{\abs{\alpha}} \leq n^{\abs{V(\alpha)}}\cdot {\binom{\abs{V(\alpha)}}{\ell}}^{\abs{\alpha}} \leq n^{\abs{V(\alpha)}}\cdot \abs{V(\alpha)}^{\abs{\alpha}\ell} \leq n^{\abs{V(\alpha)}}\cdot \paren{\ell D}^{\ell D}.
    \end{align*}
    Let $C''_{\ell,D}=\paren{\ell D}^{\ell D}$ and plugging it back into the $\Adv_{\leq D}$ computation, grouping by $\abs{V(\alpha)}$,
    \begin{align*}
        \Adv^2_{\leq D}\paren{\sfP,\sfQ}-1 \leq \paren{1+o(1)}C''_{\ell,D}\sum_{\abs{\alpha}=1}^D\sum_{\abs{V(\alpha)}=\ell}^{\abs{\alpha}\ell}n^{\abs{V(\alpha)}}\paren{2p}^{\abs{\alpha}}\rho^{2\abs{V(\alpha)}}.
    \end{align*}
    Again since $\abs{\alpha} \leq D$, for a fixed $D$ we let $C'_{\ell,D}=C''_{\ell,D}\cdot 2^D$ and we rewrite above as,
    \begin{align*}\numberthis\label{eq:important-equation}
        \Adv^2_{\leq D}\paren{\sfP,\sfQ}-1 \leq \paren{1+o(1)}C'_{\ell,D}\sum_{\abs{\alpha}=1}^D\sum_{\abs{V(\alpha)}=\ell}^{\abs{\alpha}\ell}n^{\abs{V(\alpha)}}{p}^{\abs{\alpha}}\rho^{2\abs{V(\alpha)}}.
    \end{align*}
    As argued above, for an $\ell$-uniform hypergraph $\abs{\alpha} \geq \abs{V(\alpha)}/\ell$ and since $p<1$ we have $p^{\abs{\alpha}} \leq p^{\abs{V(\alpha)/\ell}}$, and plugging into the expression above we have,
    \begin{align*}
        \Adv^2_{\leq D}\paren{\sfP,\sfQ}-1 \leq \paren{1+o(1)}C'_{\ell,D}\sum_{\abs{\alpha}=1}^D\sum_{\abs{V(\alpha)}=\ell}^{\abs{\alpha}\ell}\paren{n\rho^2p^{1/\ell}}^{\abs{V(\alpha)}},
    \end{align*}
    We upper bound the summation as $D$ times the inner summation and letting $C_{\ell,D}=D\cdot C'_{\ell,D}$ yields,
    \begin{align*}
         \Adv^2_{\leq D}\paren{\sfP,\sfQ}-1 \leq \paren{1+o(1)}C_{\ell,D}\sum_{\abs{V(\alpha)}=\ell}^{\abs{\alpha}\ell}\paren{n\rho^2p^{1/\ell}}^{\abs{V(\alpha)}}
    \end{align*}
    By our assumption on the parameter regimes for $k$ we have that each term in expression above is $o_n(1)$ and since $D,\ell$ are fixed we have,
    \begin{align*}
        C_{\ell,D}\sum_{\abs{V(\alpha)}=\ell}^{\abs{\alpha}\ell}\paren{n\rho^2p^{1/\ell}}^{\abs{V(\alpha)}} \leq C_{\ell,D}\cdot  \leq C_{\ell,D}\cdot D\ell \cdot o_n(1) = o_n(1)
    \end{align*}
    Now we obtain the desired guarantee as $\Adv^2_{\leq D}-1 \leq \paren{1+o_n(1)}\cdot o_n(1)$ and hence $\Adv^2_{\leq D} \leq 1 + o_n(1)$. Now it remains to check that $\sfP$ is supported on $\cR\paren{k}$. We note that $\sfP$ is the conditional distribution on the good event $\cG_k$ that the planted set $\abs{S} \geq k$. By our construction, every hyperedge contained inside $S$ is deleted and $S$ is an independent set of size atleast $k$. Therefore, the sampled hypergraph belongs to $\cR(k)$ with probability $1$ under $\sfP$.
\end{proof}

\end{document}